\documentclass[journal]{IEEEtran}
\IEEEoverridecommandlockouts
\usepackage[cmex10]{mathtools}
\usepackage{mathrsfs}
\usepackage[export]{adjustbox}

\usepackage{cite}
\usepackage{amsmath,amssymb,amsfonts}

\usepackage{mathtools}
\usepackage{algorithmic}
\usepackage{graphicx}
\usepackage{textcomp}
\usepackage{amsthm}
\usepackage{bm}
\usepackage{dsfont}
\usepackage{enumitem}
\usepackage{comment}

\usepackage{color,soul}
\usepackage{multirow,array}
\usepackage{blkarray}
\graphicspath{{./images/}}
\usepackage{amssymb}
\usepackage{amsmath}

\newtheorem{theorem}{Theorem}
\newtheorem{lemma}{Lemma}

\newtheorem{claim}{Claim}

\newtheorem{definition}{Definition}
\newtheorem{condition}{Condition}
\usepackage[T1]{fontenc}
\usepackage{tabularx,ragged2e,booktabs}

\newcommand*{\Scale}[2][4]{\scalebox{#1}{$#2$}}%

\newcommand{\rv}[1]{{\color{black}#1}}
\newcommand{\rvv}[1]{{\color{black}#1}}

\usepackage{lineno}

\newcommand*{\set}{\fontfamily{qag}\selectfont}
\DeclareTextFontCommand{\textset}{\set}
\newcommand{\CPX}{{\textset{CKT}(r,cn^s)}}
\newcommand{\intg}{{\textbf{int}}}
\newcommand{\Cin}{{\mathcal{C}_{\mathrm{in}}}}
\newcommand{\Cout}{{\mathcal{C}_{\mathrm{out}}}}
\newcommand{\Be}{{\mathcal{B}_{pn}(0)}}

\usepackage{ifthen}
\newboolean{extend_v}
\newboolean{editor}
\begin{document}

% Add line numbers
%\linenumbers

\setboolean{extend_v}{false} %Set to true to compile extended version, false to compile ISIT version
\setboolean{editor}{false} %Set to true to compile editor comments, false to hide them

\title{Channel Capacity for Adversaries with Computationally Bounded Observations \\
\thanks{This work was supported in part by the Office of Naval Research under ONR Grant N00014-21-1-2472, by NSF Grants CCF-1618475, CCF-1816013, CCF-2008527, CNS-2107363, CIF-2309887 and also by National Spectrum Consortium (NSC) under grant W15QKN-15-9-1004. 

E. Ruzomberka is with the Department of Electrical and Computer Engineering, Princeton University, USA (email: er6214@princeton.edu). This work was done while E. Ruzomberka was with the Elmore Family School of Electrical and Computer Engineering, Purdue University, West Lafayette, USA. C.-C. Wang and D. J. Love are with the Elmore Family School of Electrical and Computer Engineering, Purdue University, West Lafayette, USA (email: {chihw,djlove}@purdue.edu). A preliminary version of the work was presented at the 2022 IEEE International Symposium on Information Theory \cite{Ruzomberka2022}.}
}

\author{\IEEEauthorblockN{Eric Ruzomberka, Chih-Chun Wang and David J. Love}
%\IEEEauthorblockA{\textit{School of Electrical and Computer Engineering, Purdue University, West Lafayette, USA}\\
%email: \{eruzombe,chihw,djlove\}@purdue.edu}
}

%\author{\IEEEauthorblockN{Eric Ruzomberka}
%\IEEEauthorblockA{\textit{Purdue University}\\
%West Lafayette, USA \\
%eruzombe@purdue.edu}
%\and
%\IEEEauthorblockN{David J. Love}
%\IEEEauthorblockA{\textit{Purdue University}\\
%West Lafayette, USA \\
%djlove@purdue.edu}
%\and
%\IEEEauthorblockN{Chih-Chun Wang}
%\IEEEauthorblockA{\textit{Purdue University}\\
%West Lafayette, USA \\
%chihw@purdue.edu}
%}

\maketitle

%\begin{abstract}
%\hl{We should begin by introducing computationally bounded channels and information bounded channels, and say we are studying an intermediate of the two.}
%We study unique decoding over adversarial channels in which an adversary seeks to disrupt point-to-point communication. We consider 2 practical restrictions on the adversary's power. First, the adversary is partially \textit{oblivious}, i.e., it has limited knowledge of the transmitted codeword. It obtains this knowledge by first choosing an observation function with a bounded number of output bits followed by observing a mapping of the codeword via its chosen observation function. Second, the adversary is \textit{computationally bounded}, i.e., it must choose an observation function that can be computed using a finite amount of computational resources. We characterize the capacity of this channel under certain parameters.
%\end{abstract}

\begin{abstract}
We study reliable communication over point-to-point adversarial channels in which the adversary can observe the transmitted codeword via some function that takes the $n$-bit codeword as input and computes an $rn$-bit output for some given $r \in [0,1]$. We consider the scenario where the $rn$-bit observation is \textit{computationally bounded} -- the adversary is free to choose an arbitrary observation function as long as the function can be computed using a polynomial amount of computational resources. This observation-based restriction differs from conventional channel-based computational limitations, where in the later case, the resource limitation applies to the computation of the (adversarial) channel error/corruption. For all $r \in [0,1-H(p)]$ where $H(\cdot)$ is the binary entropy function and $p$ is the adversary's error budget, we characterize the capacity of the above channel and find that the capacity is identical to the completely oblivious setting ($r=0$). This result can be viewed as a generalization of known results on myopic adversaries and on channels with active eavesdroppers for which the observation process depends on a fixed distribution and fixed-linear structure, respectively, that cannot be chosen arbitrarily by the adversary.
\end{abstract}
% We analyze the above channel by extending a framework proposed by Langberg 2008.
%  This threshold behavior can be viewed as a generalization of the existing myopic adversaries or active eavesdroppers for which the observation process depends on a fixed distribution or a variable mapping with a fixed linear structure, respectively, that cannot be chosen arbitrarily by the adversary.

\begin{IEEEkeywords}
Adversarial channels, capacity, arbitrarily varying channels
\end{IEEEkeywords}

  \ifthenelse{\boolean{editor}}{
  { \color{red} \textbf{Things to do before submitting.}

  }}{}
 
  \section{Introduction} \label{sec:intro}

  Beginning with Shannon's seminal paper \cite{Shannon1948ACommunication}, early channel coding research observed that fundamental coding limits are highly sensitive to channel modeling assumptions. This sensitivity is demonstrated by a gap in capacity between the two classical models: the \textit{Shannon model} in which channel errors are random and follow a known distribution and the \textit{Hamming model} in which error patterns are worst-case for some fixed number of bit errors. In the design of robust codes, the more conservative Hamming model is particularly attractive as it makes no assumptions about the channel distribution and thus any resulting conclusion is \textit{robust} against a wide variety of channel imperfections. The downside of the Hamming model, however, is that it admits a smaller capacity than the Shannon model. In many cases, the gap in capacity is large \cite{McEliece1977a}.
  
  \subsection{Closing the gap}
  
  Recent research efforts have made progress in closing this gap by considering settings in between the two classical models. Ideally, the following two properties hold for a good channel model:
  \begin{enumerate}
  \item[]\textbf{Property 1:} The channel is \textit{mild} in the sense that its capacity coincides with the Shannon model capacity. 
  \item[]\textbf{Property 2:} The channel inherits conservative aspects of the Hamming model. In particular, the channel may be altered in an arbitrary manner unknown to the communicating parties.
  \end{enumerate}
  
  In the following Section \ref{sec:intro_cmplx}, we focus on two different approaches which have had some success towards producing good channel models. These approaches are 1) to bound the channel's computing power (i.e., computational complexity) \cite{Lipton1994ATheory,Guruswami2016OptimalChannels} and 2) to bound the information known to the channel about the communication scheme \cite{Sarwate2010a,Dey2019a,Budkuley2020SymmetrizabilityAVCs,Chen2015AChannels,Dey2016AErasures,Suresh2021Stochastic-AdversarialSnooping,Csiszar1989CapacityChannels,Csiszar1988TheConstraints,Langberg2008ObliviousCapacity}.

  \subsection{Complexity Bounded Models vs Partially Oblivious Models} \label{sec:intro_cmplx}
  
  Consider a transmitter Alice who wishes to communicate a message $\rv{m_0}$ drawn randomly from a set of $M$ possible messages over a noisy channel to a receiver Bob. To protect the message from noise corruption, Alice encodes $m_0$ into an $n$-bit codeword $\bm{x}$ of rate $R = (1/n)\log M $ and transmits $\bm{x}$ over the channel. The channel adds an $n$-bit error vector $\bm{e}$ to $\bm{x}$, and Bob receives the $n$-bit word $\bm{y} = \bm{x} \rv{\oplus} \bm{e}$. The channel is controlled by an \textit{adversary} who chooses $\bm{e}$ to prevent reliable (unique) decoding by Bob. For an error budget $p \in (0,1/2)$, the adversary can induce at most $pn$ bit flips, i.e., the Hamming weight of $\bm{e}$ must be bounded above by $pn$. \rv{We focus on  \textit{deterministic codes} in which the codeword $\bm{x}$ is a deterministic function of the message $m_0$, and in turn, consider the \textit{average error criterion} in which decoding is permitted to fail over an arbitrarily small fraction of Alice's messages.\footnote{\rv{Alternatively, one may consider \textit{stochastic codes} in which $\bm{x}$ is a function of both $m_0$ and a private random key known only to Alice. Note that a deterministic code is a degenerate stochastic code where the set of private random keys is empty. Compared to the average error criterion, a stronger decoding criterion which is of interest but not considered here is the \textit{maximum error criterion} in which decoding is permitted to fail for an arbitrarily small fraction of Alice's keys.}}}

  \rv{We define the Shannon model capacity as $C_{\mathrm{Shannon}(p)} = 1 - H(p)$ for $p \in [0,1/2]$ where $H(\cdot)$ is the binary entropy function, which coincides with the capacity of a binary symmetric channel with crossover probability $p \in [0,1/2]$. In general, $C_{\mathrm{Shannon}}(p)$ is an upper bound of any rate achievable by any communication scheme used by Alice and Bob, but may be tight depending on additional assumptions made about the adversary's capabilities and limitations. A surprising result of Csisz\'{a}r and Narayan \cite{Csiszar1988TheConstraints} is that $C_{\mathrm{Shannon}}(p)$ is the channel capacity when the adversary must choose error vector $\bm{e}$ without knowledge of the codeword $\bm{x}$ or message $m_0$.}

  In the \textit{computationally bounded model} (first proposed by Lipton \cite{Lipton1994ATheory}), the adversary takes $\bm{x}$ as input and computes $\bm{e}$ using limited computational resources, e.g., via an algorithm that takes a bounded number of computational steps. This model has the appeal of sufficiently describing practical channels, including channels with memory and channels governed by natural, but unknown processes. However, the computationally bounded model can be \textit{severe} -- an impossibility result of Guruswami and Smith \cite{Guruswami2016OptimalChannels} is that the model's capacity can be less than the Shannon capacity, and can even be $0$ when the latter is positive. Thus, Property~1 does not hold for the computationally bounded model.\footnote{\rv{Specifically, a channel which uses logarithmic space to process the codeword $\bm{x}$ has a capacity of $0$ when $p \geq 1/4$. In light of this impossibility result, recent studies on the computationally bounded model study either unique decoding when $p \in (0,1/4)$ \cite{Shaltiel2021_b,Shaltiel2022} or relax the objective of unique decoding and instead consider list-decoding when $p \in (0,1/2)$ \cite{Kopparty2019,Shaltiel2021}. The works \cite{Guruswami2016OptimalChannels,Shaltiel2021_b,Shaltiel2022,Kopparty2019,Shaltiel2021} employ stochastic codes together with pseudorandom sequences to complicate the channels task of computing an effective error pattern $\bm{e}$. In contrast to the above works, we consider deterministic codes and unique decoding for all $p \in (0,1/2)$.}}

  Another existing approach is the \textit{partially oblivious model}, where the adversary chooses $\bm{e}$ based on incomplete side-information about the transmitted codeword $\bm{x}$. This model includes myopic channels, e.g., \cite{Sarwate2010a,Dey2019a,Budkuley2020SymmetrizabilityAVCs}, causal channels, e.g., \cite{Chen2015AChannels,Dey2016AErasures,Suresh2021Stochastic-AdversarialSnooping}, channels with active eavesdroppers, e.g., \cite{Wang2017OnChannel}, and some arbitrarily varying channels (AVCs), e.g., \cite{Csiszar1989CapacityChannels,Csiszar1988TheConstraints}. \rv{We focus on the following setting which captures a special case of the partially oblivious model}: for $r \in [0,1]$ and some observation function $f_n:\{0,1\}^n \rightarrow \{0,1\}^{rn}$, the adversary makes an $rn$-bit observation $f_n(\bm{x})$ of codeword $\bm{x}$, and in turn, chooses $\bm{e}$. \rv{We emphasize that the error vector can depend non-causally on the $rn$-bit observation and, thus, causal channels are not captured by our setting.} The special cases $r = 0$ and $r=1$ correspond to no information (i.e., completely oblivious) and perfect information (i.e., omniscient), respectively. 
  
  Property 1 can hold for the partially oblivious model when $r$ is sufficiently small.\footnote{This fact is an analog to a channel being \textit{sufficiently myopic} (see \cite{Dey2019a}).} However, Property 2 does not hold for many partially oblivious channels in the literature. \rv{While all partially oblivious channels allow error vector $\bm{e}$ to chosen in an arbitrarily manner unknown to Alice and Bob, the observation function $f_n$ is not always chosen arbitrarily. For example, a myopic channel in our setting corresponds to $f_n$ being drawn randomly from some distribution known to Alice and Bob.} \textit{For Property 2 to hold, however, we must allow $f_n$ to be chosen arbitrarily and require Alice and Bob to devise their communication scheme without knowledge of $f_n$}. This is equivalent to the adversary choosing a \textit{worst-case} $f_n$ for a fixed $r$, a model defined and studied by Langberg \cite{Langberg2008ObliviousCapacity} under the name of the \textit{$(1-r)$-oblivious channel}.\footnote{\rv{An alternative interpretation of Property 2 is that the adversary may choose a worst-case $f_n$ from some \textit{subset of functions} from $\{0,1\}^n$ to $\{0,1\}^{rn}$, and where the subset is known to Alice and Bob. In the $(1-r)$-oblivious channel, this subset is the improper subset of all functions. Another channel that satisfies Property 2 under this alternative interpretation is the adversarial wiretap channel of type II \cite{Wang2017OnChannel}, in which $f_n$ is chosen from the set of all linear mappings from $\{0,1\}^n$ to $\{0,1\}^{rn}$. Depending on the specific application for which the channel model serves, it may be unrealistic to assume that this subset contains only linear mappings.}} The capacity of the $(1-r)$-oblivious channel remains an open problem, where the best known lower bound is given by \cite{Langberg2008ObliviousCapacity}.

  \subsection{This Work} \label{sec:contribution}
  
  In this paper, we define and study a channel model that has qualities of both the computationally bounded model and the partially oblivious model. Roughly speaking, we define this model by requiring the adversary to observe $\bm{x}$ via an $rn$-bit observation function $f_n$ that is computationally bounded.
  
 Specifically, for fixed positive integers $c$ and $s$, the adversary chooses a sequence of observation functions $f_n(\cdot)$, $\forall n \geq 1$ that belongs to $\textset{CKT}(r,cn^s)$ -- the set of observation functions with $n$ input bits and $rn$ output bits that can be computed by a Boolean circuit with at most $cn^s$ gates. We allow the choice of $f_n$ to be unknown to Alice or Bob. On the other hand, the $f_n$ chosen by the adversary can depend on the codebook of Alice but cannot depend on the actual message being sent. Using the observation function $f_n$ of its choice, the adversary observes $f_n(\bm{x})$ and chooses $\bm{e}$ with no computational bound. \rv{Our model differs from the prior works \cite{Lipton1994ATheory,Guruswami2016OptimalChannels,Shaltiel2021_b,Shaltiel2022,Kopparty2019,Shaltiel2021}, where in the latter, the channel has a complete view of $\bm{x}$ but must choose $\bm{e}$ subject to a computational bound.} We refer to our adversary as a \textit{$\textset{CKT}(r,cn^s)$-oblivious adversary}. By construction, Property 2 holds for a channel controlled by a $\textset{CKT}(r,cn^s)$-oblivious adversary due to $f_n$ being unknown to Alice or Bob.
  
Our computational restriction is modeled after realistic adversarial channels. A channel controlled by a $\textset{CKT}(r,cn^s)$-oblivious adversary closely approximates a \textit{$(1-r)$-oblivious channel} \cite{Langberg2008ObliviousCapacity} (i.e., the definition therein is equivalent to the $\textset{CKT}(r,\infty)$-oblivious adversary) without weakening the power of the adversary too much. Indeed, the adversary is quite strong. To illustrate its strength, if for a sequence of functions $\{f_n\}_{n=1}^{\infty}$ satisfying  $\forall c,s \geq 1$ there exists a finite $n_0$ such that for all $n \geq n_0$ $f_n \not\in \CPX$, then the sequence is widely regarded to be an \textit{infeasible computation} \cite{Sipser2006IntroductionComputation}. The technical value of the computational constraint is to bound the number of observation functions that the adversary can choose from while still including a wide range of important observation functions in the problem formulation.

In this paper, for any \rv{fixed finite} integers $c,s$, and all $p \in (0,1/2)$ and $r \in [0,1-H(p))$, we study the channel controlled by a $\CPX$-oblivious adversary with error budget $p$ by characterizing the channel capacity $C(p,r,\rv{c,s})$. As our main result, we show that $C(p,r,\rv{c,s})$ is exactly $1-H(p)$, and thus the capacity is independent of parameters $c,s,r$ for the stated parameter regime. It follows that $C(r,p,\rv{c,s})$ coincides with the Shannon model capacity $C_{\textrm{Shannon}}(p)$ and thus Property 1 holds. \rv{Futhermore, in this regime, deterministic codes are optimal.}\footnote{\rv{For the parameter regime $r \geq 1-H(p)$ and $c\geq1$, $s \geq 1$, deterministic codes may not be optimal. We remark that our proof techniques, which involve a random coding argument over a set of deterministic codes, only work for the regime $r < 1- H(p)$.  For $r \geq 1 - H(p)$, the channel to the adversary is ``less noisy'' than the channel to Bob, such that when a deterministic code is used at rate less than $1-H(p)$, the adversary is likely to decode Alice's codeword (with high probability over the code selection) and thus the adversary is effectively omniscient (i.e., $r=1$). For omniscient adversaries, the GV bound of $1-H(2p) $\cite{Gilbert1952a,Varshamov1957a} is the best-known achievable rate. However, when a stochastic code is used, one may find achievable rates exceeding the GV bound for some values of $r \geq 1-H(p)$ and $p \in (0,1/2)$. In fact, in some channel models, stochastic codes are known to achieve rates significantly larger than the GV bound for certain parameters when the channel to the adversary is ``less noisy'' than the channel to Bob (see, e.g., \cite{Dey2019}).}} This main result was first presented at the International Symposium on Information Theory (ISIT) \cite{Ruzomberka2022}.

The remainder of this paper is organized as follows. In Section \ref{sec:model}, we present the precise channel model and main result. The main result is discussed in the context of related work on myopic channels, channels controlled by active eavesdroppers, and $(1-r)$-oblivious channels. In Section \ref{sec:outline_and_overview}, we present the overview of the proof of the main result and discuss our proof techniques in the context of related work. In Section \ref{sec:analysis}, we present the detailed proof of the main result.

  \section{Channel Model \& Results} \label{sec:model}
  \subsection{Notation}
  All vectors are in bold notation. Let $d(\bm{z},\bm{z}')$ denote the Hamming distance between two binary vectors $\bm{z}$ and $\bm{z}'$. For $t>0$ and $\bm{z} \in \{0,1\}^n$ we define $\mathcal{B}_{t}(\bm{z}) = \{\bm{z}' \in \{0,1\}^n: d(\bm{z},\bm{z}') \leq t \}$ to be the Hamming ball of radius $t$ centered around $\bm{z}$. The functions $\log(\cdot)$ and $\ln(\cdot)$ denote the base $2$ and base $e$ logarithms, respectively. For a number $K \geq 1$, let $[K]$ denote the set $\{1,\ldots, \lfloor K \rfloor \}$. For an integer blocklength $n \geq 1$ and rate $R \in (0,1]$, an $(n,Rn)$ codebook $\mathcal{C}_n$ is a function $\mathcal{C}_n: [2^{R n}] \rightarrow \{0,1\}^n$. When useful, we will think of $\mathcal{C}_n = \{\mathcal{C}_n(1), \ldots, \mathcal{C}_n(2^{R n})\}$ as a subset of $\{0,1\}^{n}$ and the $i$th codeword $\mathcal{C}_n(i)$ as a vector in $\{0,1\}^{n}$. For a number $\rho>0$ and a binary vector $\bm{a} =(a_1,\ldots,a_{\rho n}) \in \{0,1\}^{\rho n}$, we define the integer representation of $\bm{a}$ to be the integer $\intg(\bm{a}) = 1+\sum_{j=1}^{\rho n} \rv{a_j} 2^{j-1} \in [2^{\rho n}]$. For functions $g(n)$ and $h(n)$, we adopt standard ``little O'', ``big O'' and ``big Omega'' notation: $g = o(h(n))$ if $\lim_{n \rightarrow \infty} \frac{g(n)}{h(n)} = 0$, $g = O(h(n))$ if $\exists k$ s.t. for large enough $n$, $g(n) \leq k h(n)$, and $g = \Omega(h(n))$ if $\exists k$ s.t. for large enough $n$, $g(n) \geq k h(n)$.
  
  \subsection{Channel Model}
  
  \textbf{Alice's Encoding:} A transmitter Alice communicates over a noisy channel with a receiver Bob in the following manner. For a rate $R \in (0,1]$ and integer blocklength $n \geq 1$, Alice randomly draws a message $m_0$ uniformly from a message set $[2^{Rn}]$. For a $(n,Rn)$ codebook $\mathcal{C}_n$, Alice encodes $m_0$ into a codeword $\bm{x} \in \{0,1\}^\rv{n}$ by computing $\bm{x} = \mathcal{C}_n(m_0)$. Since $\bm{x} = \mathcal{C}_n(m_0)$ is a deterministic function of $m_0$, we say that Alice is using a deterministic code. After encoding, Alice transmits $\bm{x}$ into the channel.
  
  \textbf{Bob's Decoding:} At the channel output, Bob receives word $\bm{y} = \bm{x} \rv{\oplus} \bm{e}$ where $\bm{e} \in \{0,1\}^n$ is an error vector added by the channel \rv{and where the symbol `$\rv{\oplus}$' denotes the bit-wise XOR}. Bob outputs a message estimate $\hat{m}$ based on the received word $\bm{y}$. We say that a decoding error occurs if $\hat{m} \neq m_0$. 
  
  \textbf{Adversary:} The channel is controlled by an adversary who has side-information about Alice's and Bob's communication scheme but not exact knowledge of the transmitted message $m_0$. In particular, the adversary knows Alice's codebook $\mathcal{C}_n$ and is \textit{partially oblivious} to the transmitted codeword $\bm{x}$. By partially oblivious, we mean that for observation rate $r \in [0,1]$, the adversary randomly draws a function $f_n:\{0,1\}^n \rightarrow \{0,1\}^{rn}$ with probability $U_f(f_n)$ and observes a realization $\bm{\psi}$ of the random variable $\bm{\Psi} = \bm{\Psi}(m_0) = f_n(\mathcal{C}_n(m_0)) = f_n(\bm{x})$.\footnote{The fact that $\bm{\Psi}$ is a random variable follows from its dependency on the random variable $m_0$.} Using its knowledge of $\mathcal{C}_n$ but without knowledge of the realization of $m_0$, the adversary randomly draws $f_n$ with probability $U_f(f_n)$. Due to the adversary's computational bound, for positive integers $c,s$, $U_f(f_n)=0$ for all $f_n \not\in \textset{CKT}(r,cn^s)$ (we provide a rigorous definition of $\textset{CKT}(r,cn^s)$ in Section \ref{sec:CPX_def}). Neither the actual choice of $f_n$ nor the distribution $U_f(\cdot)$ is revealed to Alice or Bob. As a result, the model falls into the adversarial setting in which the adversary has full freedom of using any specific function (by choosing $U_f(\cdot)$ to be a delta distribution) or any random selection of functions (by choosing $U_f(\cdot)$ to be of general distribution). 
  
  Finally, using knowledge of the codebook $\mathcal{C}_n$ and observation function $f_n$, the adversary chooses the conditional probability $U_{\bm{e}|\bm{\psi}}(\bm{e}|\bm{\psi})$ that the error vector $\bm{e}$ is added to the channel given that it observes $\bm{\Psi}(m_0) = \bm{\psi}$. For $p\in (0,1/2)$, we impose an error budget constraint such that $\bm{e}$ has a Hamming weight bounded above by $pn$, i.e., $U_{\bm{e}|\bm{\psi}}(\bm{e}|\bm{\psi})=0$ for all $\bm{e} \not\in \Be$ and $\bm{\psi} \in \{0,1\}^{rn}$. We refer to the above adversary as the $\textset{CKT}(r,cn^s)$-oblivious adversary with error budget $p$. We note that the distribution $U_f(f_n)$ and $U_{\bm{e}|\bm{\psi}}(\bm{e}|\bm{\psi})$ are used so that some randomness can be embedded in the adversary's action. For simplicity, the reader may assume that the adversary chooses deterministically an observation function $f_n\in \CPX$, uses the observation function to observe $\bm{\Psi}=f_n(\mathcal{C}_n(m_0)) = \bm{\psi}$, and then chooses deterministically an error vector $\bm{e}$ under the given error budget $p$.

  \subsection{Adversary's Computational Bound} \label{sec:CPX_def}
 
  For observation rate $r \in [0,1]$, positive integers $c,s,n$, we define the set $\textset{CKT}(r,cn^s)$. Let $\mathcal{F}_{n,r}$ denote the set of all Boolean functions of the form $f_n:\{0,1\}^n \rightarrow \{0,1\}^{r n}$. To define $\textset{CKT}(r,cn^s)$, we first define the circuit complexity of a function $f_n \in \mathcal{F}_{n,r}$.
  
  A Boolean circuit $B_n$ is an acyclic directed graph where each node is either an input node (with in-degree 0) or a logic gate (with in-degree 2). All nodes in $B_n$ have out-degree $1$ with unbounded fan-out and each logic gate computes an arbitrary Boolean function from $\{0,1\}^2$ to $\{0,1\}$. The \textit{size} of $B_n$ is the total number of nodes in $B_n$ (input nodes and logic gates). Note that an observation function $f_n \in \mathcal{F}_{n,r}$ can be computed by some Boolean circuit that takes $n$ bits as input and produces $rn$ bits as output. The \textit{circuit (size) complexity} of an observation function $f_n \in \mathcal{F}_{n,r}$ is the size of the smallest size Boolean circuit $B_n$ that can compute $f_n$. We define $\textset{CKT}(r,cn^s)$ to be the set of all functions $f_n \in \mathcal{F}_{n,r}$ with a circuit complexity of at most $cn^s$. In modern complexity theory, the study of circuit complexity is a common approach for proving lower bounds on the complexity of certain problems \cite{Sipser2006IntroductionComputation}.
  
  %For $f_n \in \mathcal{F}_{n,r}$,  $m_0 \in [M]$ and realization $\bm{\psi} = f_n(\mathcal{C}_n(m_0))$, define the observation set $\mathcal{O}_{\bm{\psi}} = \{\bm{z} \in \{0,1\}^{n}: f_n(\bm{z}) = \bm{\psi} \}$. We note that observing $\bm{\Psi}(m_0) = \bm{\psi}$ is equivalent to knowing that $\mathcal{C}_n(m_0) \in \mathcal{O}_{\bm{\psi}}$. Hence, the following two perspectives are equivalent: the adversary chooses $f_n \in \mathcal{F}_{n,r}$ and the adversary chooses a partition $\vec{\mathcal{O}} = (\mathcal{O}_{1}, \ldots, \mathcal{O}_{2^{rn}})$ of the space $\{0,1\}^n$ consisting of $2^{rn}$ non-empty observation sets. With an abuse of notation, for $f_n \in \CPX$, we write $\vec{\mathcal{O}} \in \CPX$ to denote the partitioned observation sets corresponding to $f_n$ with circuit complexity upper bounded by $cn^s$.
  
  \subsection{Capacity}
  
  For an $(n,Rn)$ codebook $\mathcal{C}_n$, the (average) probability of decoding error is defined as
  \begin{equation} \label{eq:prob_error_def}
  \begin{aligned}
  & \bar{P}_e(\mathcal{C}_n) = \\
  &\max_{f_n \in \CPX} \mathbb{E}_{\bm{\Psi}} \left[ \max_{\bm{e} \in \Be} \mathbb{P}_{m_0}(\hat{m}(\bm{e},m_0) \neq m_0 |\bm{\Psi} = \bm{\psi}) \right]
  \end{aligned}
  \end{equation}
  where the probability measure $\mathbb{P}_{m_0}(\cdot)$ is w.r.t. the distribution of $m_0$, and the expectation $\mathbb{E}_{\bm{\Psi}}[\cdot] = \sum_{\bm{\psi} \in \{0,1\}^{rn}} (\cdot) \mathbb{P}_{m_0}(\bm{\Psi}(m_0) = \bm{\psi})$. Given the above channel model, we can define achievable rate in the usual way.
  \begin{definition}[Achievable Rate] 
  For $p \in (0,1/2)$, $r \in [0,1]$, and positive integers $c,s$, a rate $R \in (0,1]$ is said to be ($c,s$)-achievable if for any $\epsilon_e >0$, there exists an $n_0$ such that for all $n \geq n_0$, there exists an $(n,Rn)$ codebook $\mathcal{C}_n$ such that $\bar{P}_e(\mathcal{C}_n) \leq \epsilon_e$.
  \end{definition}
  For $p \in (0,1/2)$, $r \in [0,1]$, and positive integers $c,s$, we define the capacity $C(p,r,\rv{c,s})$ of a channel controlled by a $\CPX$-oblivious adversary as the supremum of $(c,s)$-achievable rates. Let $C(p,r,\infty,\infty)$ denote the capacity of $(1-r)$-oblivious channel for which there is no constraint on the computational complexity when computing the $rn$-bit observation, see \cite{Langberg2008ObliviousCapacity}. 
  
  \subsection{Main Result} \label{sec:main_result}
  
  Under the above model, the Shannon capacity is $C_{\mathrm{Shannon}}(p) = 1-H(p)$ where $H(p) = -p \log p - (1-p) \log(1-p)$ is the binary entropy function \cite{Csiszar1988TheConstraints,Langberg2008ObliviousCapacity}. The following result shows that Property 1 holds for our model for a wide range of $r$.
  
  \begin{theorem} \label{thm:main_result}
  For $p \in (0,1/2)$, $r \in [0,C_{\mathrm{Shannon}}(p))$, and $c,s \geq 1$, $C(p,r,\rv{c,s}) = C(p,0,\rv{c,s}) = C(p,0,\infty,\infty) = C_{\mathrm{Shannon}}(p)$.
  \end{theorem}
  
  \begin{figure}[t]
    \includegraphics[width=\columnwidth]{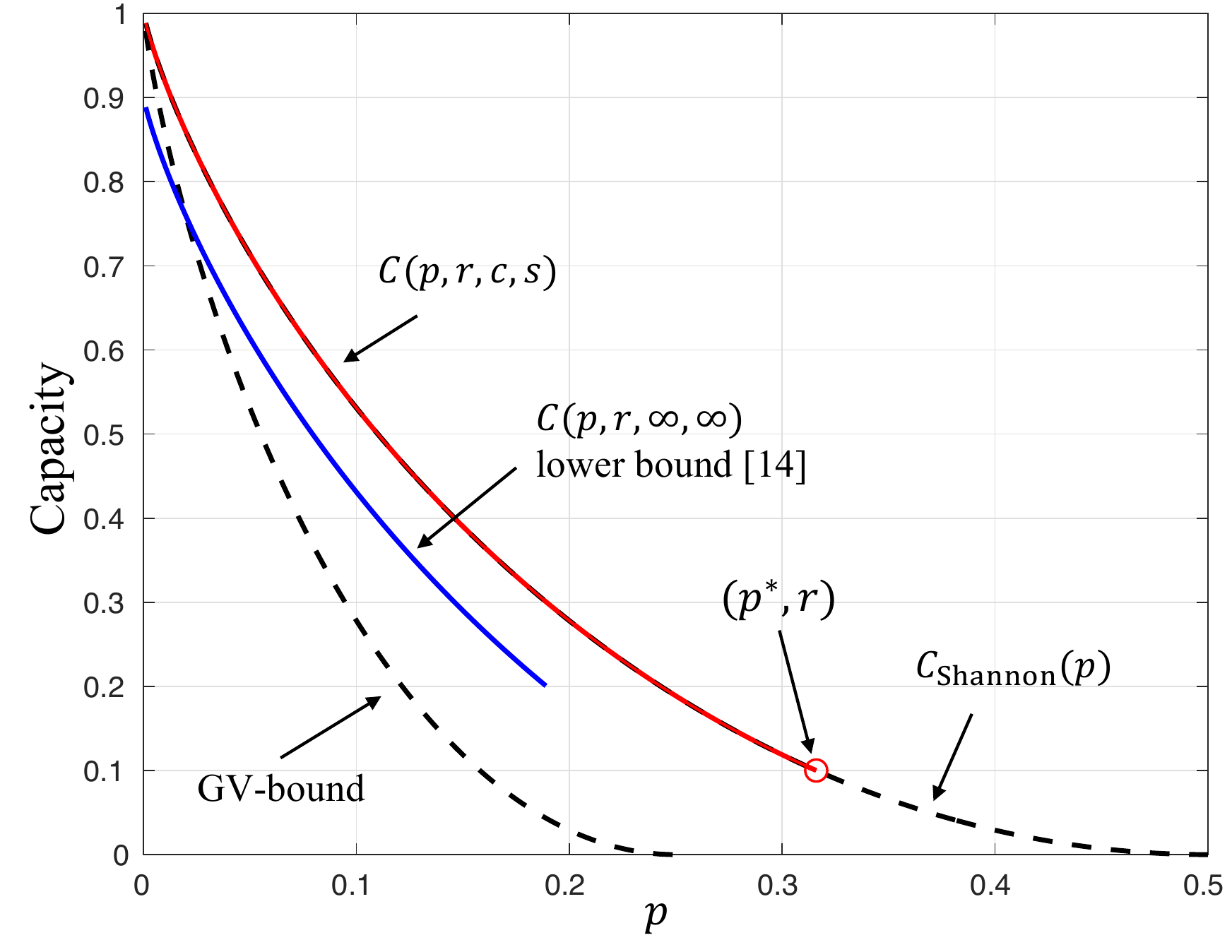}
    \centering
    \caption{Capacity when $r=0.1$ and $c,s$ are finite positive integers. Herein the value $p^*$ satisfies $C_{\mathrm{Shannon}}(p^*) = r = 0.1$.}
    %\vspace{-0.2in}
    \label{fig:capacity}
  \end{figure}
  
 We share a few remarks on the above theorem. When $r < C_{\mathrm{Shannon}}(p) = 1-H(p)$, Theorem \ref{thm:main_result} implies that the adversary can do no better than to \textit{ignore} its side-information $f_n(\bm{x})$ and choose $\bm{e}$ randomly from the set of all $n$-bit vectors with Hamming weight $pn$. Additionally, we note that the largest known lower bound on $C(p,r,\infty,\infty)$ is $1-H(p)-r$ for $r \in [0,\frac{1-H(p)}{3})$ \cite{Langberg2008ObliviousCapacity}. Since $C(p,r,\infty,\infty)$ is a lower bound to $C(p,r,\rv{c,s})$, Theorem 1 significantly sharpens the best known lower bound of $C(p,r,\rv{c,s})$ to an exactly tight characterization.\rv{\footnote{\rvv{One can show that $C(p,r,\infty,\infty)$ is strictly less than $C(p,r,c,s)$ for some values of $p \in (0,1/2)$ and $r < 1 - H(p)$. See Section \ref{sec:conc} for a proof sketch. The intuition behind this result follows from the fact that we have imposed a complexity bound of $f$ while allowing the codebook $\mathcal{C}_n$ to have unbounded complexity. Allowing encoding/decoding to use unlimited computation power while the adversary is $\textset{CKT}(r,cn^s)$-oblivious may give Alice and Bob an advantage compared the setting where both the codebook and observation function have similar complexity constraints.}}} For $r>C_{\mathrm{Shannon}}(p)$, an immediate lower bound of $C(p,r,\rv{c,s})$ is given by the Gilbert-Varshamov (GV) bound (i.e. $C(p,r,\infty,\infty) \geq 1-H(2p)$) \cite{Gilbert1952a,Varshamov1957a}.\footnote{\rvv{As discussed above, when $r>C_{\mathrm{Shannon}}(p)$ one may find achievable rates strictly greater than the GV bound when stochastic codes are considered. One such stochastic coding scheme is the following. Suppose that the encoder passes its clean codeword $\bm{u}$ through a BSC($q$) ($q$ to be determined) to obtain the transmitted codeword $\bm{x}$. If the effective mutual information between the clean codeword $\bm{u}$ and the adversary’s observation is less than the rate $R$ (the ``right'' notion of sufficient myopicity in this scenario), then the above stochastic coding scheme can be shown to achieve rate $R = 1-H(p')$ where $p' = q(1-p)+p(1-q)$. As can be verified numerically, there exists some values of $r$, $p$ and $q$ such that $r>1-H(p)$ and $R>1-H(2p)$ (the GV bound).}} All results discussed thus far are summarized in Fig. \ref{fig:capacity}.
 
 Theorem \ref{thm:main_result} generalizes a few results on myopic channels and on channels with active eavesdroppers. For \rv{$p \in [0,1/2]$} and $r < C_{\mathrm{Shannon}}(p)$, $C_{\mathrm{Shannon}}(p)$ is known to be the capacity of the \textit{\rv{binary-erasure bit-flip} myopic channel} where the adversary a) non-causally views $\bm{x}$ through a binary erasure channel with erasure probability $1-r$ (\rv{denoted as BEC($1-r$) in the literature}) then b) injects $pn$ bit errors \cite[Theorem~III.12]{Dey2019a}. It is clear that this result is generalized by Theorem \ref{thm:main_result} after observing that a $\textset{CKT}(r,cn^s)$-oblivious adversary can choose $f_n$ \textit{randomly} in a way that simulates a BEC($1-r$). Similarly, for $r<C_{\mathrm{Shannon}}(p)$, $C_{\mathrm{Shannon}}(p)$ is known to be the capacity of the \textit{adversarial wiretap channel of type II} where the adversary a) chooses $rn$ indices in $\{1,\ldots,n\}$ and observes $rn$-bits of $\bm{x}$ at the chosen indices then b) injects $pn$ bit errors \cite[Theorem~4.2]{Wang2017OnChannel}. It is clear that \cite[Theorem~4.2]{Wang2017OnChannel} is a special case of Theorem \ref{thm:main_result} after observing that a $\textset{CKT}(r,cn^s)$-oblivious adversary can choose $f_n(\bm{x})$ to output a subset of $rn$ bits of $\bm{x}$.

  %For $p \in (0,1/2)$, the capacity of the above model when the adversary is oblivious is known to be $C(p,0,\cdot) = C_{\mathrm{Shannon}}(p) = 1-H(p)$ \cite{Csiszar1988TheConstraints,Langberg2008ObliviousCapacity}. Our main result is as follows.
  
  %\begin{theorem} \label{thm:main_result}
  %For $p \in (0,1/2)$, $r \in [0,C_{\mathrm{Shannon}}(p))$, and $c,s \geq 1$, $C(p,r,cn^s) = C(p,0,cn^s) = C(p,0,\infty) = C_{\mathrm{Shannon}}(p)$.
  %\end{theorem}
  
  \section{Proof Outline, Overview of Proof Technique} \label{sec:outline_and_overview}
  In this section, we outline the proof of Theorem 1 and discuss an overview of our proof technique. A detailed proof of Theorem 1 can be found in Section \ref{sec:analysis}.
  
  \subsection{Achievability Scheme} \label{sec:ach_scheme}
  
  For our proof of Theorem \ref{thm:main_result}, we construct a specific $\mathcal{C}_n$.  %We will show that our construction is \textit{optimal} for $p \in (0,1/2)$ and $r \in [0,1-H(p))$ in the sense that for $\epsilon_R >0$, the rate $R = C_{\mathrm{Shannon}}(p) - \epsilon_R$ is ($c,s$)-achievable using our construction of $\mathcal{C}_n$.
  
  \textbf{Encoder Construction:} Alice's $(n,Rn)$ codebook $\mathcal{C}_n$ is constructed as follows. Let $\rho \in (R,C_{\mathrm{Shannon}}(p))$. Codebook $\mathcal{C}_n$ is a concatenation of two codebooks: an \textit{outer} $(\rho n,R n)$ codebook $\mathcal{C}_{\mathrm{out}}: [2^{R n}] \rightarrow \{0,1\}^{\rho n}$ and an \textit{inner} $(n, \rho n)$ codebook $\mathcal{C}_{\mathrm{in}}:\{0,1\}^{\rho n} \rightarrow \{0,1\}^{n}$. Encoding proceeds as follows. First, Alice encodes $m_0$ with $\mathcal{C}_{\mathrm{out}}$ where we denote the resulting codeword as $\mathcal{C}_{\mathrm{out}}(m_0)$. Subsequently, Alice encodes $\mathcal{C}_{\mathrm{out}}(m_0)$ with $\mathcal{C}_{\mathrm{in}}$ where we denote the resulting codeword as $\mathcal{C}_n(m_0) = \mathcal{C}_{\mathrm{in}}(\mathcal{C}_{\mathrm{out}}(m_0))$. After encoding, Alice transmits the codeword $\bm{x} = \mathcal{C}_n(m_0)$ over the channel. We denote the concatenated $(n,Rn)$ codebook as $\mathcal{C}_{n} = \mathcal{C}_{\mathrm{in}} \circ \mathcal{C}_{\mathrm{out}}$.
  
 \textbf{Decoder Construction:} Bob's list decoder is constructed as follows. Given the received word $\bm{y} = \mathcal{C}_n(m_0) \rv{\oplus} \bm{e}$, Bob first performs list decoding by forming a list $\mathcal{L}_{\mathrm{in}}(\bm{y},\mathcal{C}_{\mathrm{in}})$ of all words $\bm{w} \in \{0,1\}^{\rho n}$ such that $\mathcal{C}_{\mathrm{in}}(\bm{w})$ is contained in the ball $\mathcal{B}_{pn}(\bm{y})$. After list decoding, Bob refines the list (i.e., Bob performs disambiguation) by removing all words $\bm{w} \in \mathcal{L}_{\mathrm{in}}$ that are \textit{not consistent} with $\mathcal{C}_{\mathrm{out}}$: we say that a word $\bm{w}$ is consistent with $\mathcal{C}_{\mathrm{out}}$ if there exists an $m \in [2^{R n}]$ such that $\bm{w} = \mathcal{C}_{\mathrm{out}}(m)$.

 Denote the refined list as $\mathcal{L}_{\mathrm{out}}$ and note that $\mathcal{L}_{\mathrm{out}} \subseteq \mathcal{L}_{\mathrm{in}} \subseteq \{0,1\}^{\rho n}$. After $\mathcal{L}_{\mathrm{in}}$ is refined to $\mathcal{L}_{\mathrm{out}}$, a decoding decision is made according to the following rules. If $|\mathcal{L}_{\mathrm{out}}|=1$, then we have exactly one $m \in [2^{R n}]$ s.t. $C_{\mathrm{out}}(m) \in \mathcal{L}_{\mathrm{out}}$, and the decoder outputs $\hat{m} = m$. If $\mathcal{L}_{\mathrm{out}}$ is empty or if $|\mathcal{L}_{\mathrm{out}}| > 1$, then the decoder declares an error by setting $\hat{m}$ to an error symbol. We say that a decoding error occurs if $\hat{m} \neq m_0$. However, by the list decoding logic and the adversary error budget constraint $pn$, $\mathcal{C}_{\mathrm{out}}(m_0)$ is guaranteed to be in $\mathcal{L}_{\mathrm{out}}$, and so a decoding error occurs if and only if $|\mathcal{L}_{\mathrm{out}}|>1$.

  \textbf{Probability of Error:} For $i = 1, \ldots, 2^{\rho n}$, define
  \begin{align} 
  & \bm{w}_i(m_0,\bm{e},\mathcal{C}_{\mathrm{out}},\mathcal{C}_{\mathrm{in}}) = \arg \min_{\bm{w} \in \mathcal{W}_i(m_0,\bm{e},\mathcal{C}_{\mathrm{out}},\mathcal{C}_{\mathrm{in}})} \intg (\bm{w}) \label{eq:error_idx_ineq} \\ 
  & \text{such that} \nonumber \\
  & \mathcal{W}_i(m_0,\bm{e},\mathcal{C}_{\mathrm{out}},\mathcal{C}_{\mathrm{in}}) = \substack{\arg \min \\ \bm{w} \in \{0,1\}^{\rho n} \setminus \{\bm{w}_1, \ldots,\bm{w}_{i-1} \}} d(\bm{y},\mathcal{C}_{\mathrm{in}}(\bm{w})). \nonumber
  \end{align}
  That is, we sort the \textit{message/word vectors} $\bm{w}$ according to the distance between the observation $\bm{y}$ and the inner codeword $\Cin(\bm{w})$, where the term $\intg(\bm{w})$ in (\ref{eq:error_idx_ineq}) is used to break any tie and ensure that the $i$th closest codeword to $\bm{y}$ is uniquely defined. Note that $\bm{w}_i \in \mathcal{L}_{\mathrm{in}}$ iff $i \leq |\mathcal{L}_{\mathrm{in}}|$. Also define
  \begin{equation} \label{eq:def_incon_set}
  \mathcal{I}_{m_0} = \{ \mathcal{C}_{\mathrm{out}}(m'): m' \neq m_0\}
  \end{equation}
  to be the set of words in $\{0,1\}^{\rho n}$ that are consistent with $\mathcal{C}_{\mathrm{out}}$ but do not correspond to the true message $m_0$. Under the code construction of Section \ref{sec:ach_scheme}, the probability of decoding error can be written as  
  \begin{align} 
  & \bar{P}_{e}(\mathcal{C}_{\mathrm{out}},\mathcal{C}_{\mathrm{in}})  \nonumber \\
  & = \hspace{-1em} \max_{f_n \in \CPX} \mathbb{E}_{\bm{\Psi}} \left[ \max_{\bm{e} \in \Be}\mathbb{P}_{m_0}(|\mathcal{L}_{\mathrm{out}}|>1|\bm{\Psi}(m_0) = \bm{\psi}) \right] \nonumber \\
  & = \hspace{-1em} \max_{f_n \in \CPX} \mathbb{E}_{\bm{\Psi}} \left[ \max_{\bm{e} \in \Be} \mathbb{P}_{m_0}(\bigcup_{i=1}^{|\mathcal{L}_{\mathrm{in}}|} \{ \bm{w}_i \in \mathcal{I}_{m_0} \} |\bm{\Psi} = \bm{\psi}) \right]. \label{eq:prob_error_inside}
  \end{align}
  
  \subsection{Preliminaries}
  
  The following preliminary results characterize the list-decodability properties of a random codebook. Let $\rv{Q}(n, \rho n)$ be the distribution of $(n, \rho n)$ codebooks such that all codewords of $\mathcal{C}_{\mathrm{in}}$ are independently and uniformly distributed in $\{0,1\}^n$.
  
  \begin{definition}
  For $L>0$, an $(n, \rho n)$ codebook $\mathcal{C}_{\mathrm{in}}$ is said to be $[L,p]$ list decodable if $|\mathcal{C}_{\mathrm{in}} \cap \mathcal{B}_{pn}(\bm{y})| \leq L$ for every $\bm{y} \in \{0,1\}^{n}$.
  \end{definition}
  
  \begin{lemma} \label{thm:LD_lb}
  Let $\ell = \ell(n) > 0$ be $\omega(n)$ (i.e., $\lim_{n \rightarrow \infty} \ell(n)/n = \infty$). For large enough $n$, a codebook $\mathcal{C}_{\mathrm{in}}$ drawn from distribution $\rv{Q}(n, \rho n)$ is $[\ell,p]$ list decodable w.p. greater than $1 - 2^{-\ell(n)/4}$. Proof is in Appendix \ref{sec:proof_LD_lb}.
  \end{lemma}
  
  Similar results hold even if the list size is constant in $n$. 
  
  \begin{lemma}[{\cite[Claim~A.15]{Chen2015AChannels}}] \label{thm:LD_lb_bounded}
  Let $\epsilon_{\rho} \in (0,C_{\mathrm{Shannon}}(p))$ and set $\rho = C_{\mathrm{Shannon}}(p)-\epsilon_{\rho}$. For $L > \frac{1}{\epsilon_{\rho}}$ and for large enough $n$ (depending only on $\epsilon_{\rho}$), an $(n, \rho n)$ codebook $\mathcal{C}_{\mathrm{in}}$ drawn from distribution $\rv{Q}(n, \rho n)$ is $[L,p]$ list decodable w.p. greater than $1 - \frac{1}{n}$.
  \end{lemma}
  
  \begin{lemma} \label{thm:size_AandC}
  Consider an arbitrary $1$-to-$1$ $(\rho n, R n)$ codebook $\mathcal{C}_{\mathrm{out}}$ and randomly draw an $(n, \rho n)$ codebook $\mathcal{C}_{\mathrm{in}}$ from distribution $\rv{Q}(n, \rho n)$. Recall that $\mathcal{C}_{n} = \mathcal{C}_{\mathrm{in}} \circ \mathcal{C}_{\mathrm{out}}$. For any subset $\mathcal{A} \subseteq \{0,1\}^n$, we have that $\mu = \mathbb{E}_{\mathcal{C}_{\mathrm{in}}}|\mathcal{A} \cap \mathcal{C}_n| = 2^{-(1-R)n}|\mathcal{A}|$, and for $t_L < \mu$ and $t_U > \mu$,
  \begin{equation} \nonumber
  \mathbb{P}_{\mathcal{C}_{\mathrm{in}}}\left(|\mathcal{A}\cap\mathcal{C}_n| < t_L \right) \leq 2 \exp\left\{ \frac{-(\mu-t_L)^2}{4 \mu} \right\}
  \end{equation}
  and 
  \begin{equation} \nonumber
  \mathbb{P}_{\mathcal{C}_{\mathrm{in}}}\left(|\mathcal{A}\cap\mathcal{C}_n| > t_U \right) \leq 2 \exp \left\{ \frac{-(t_U- \mu)^2}{4(t_U + \mu)} \right\}.
  \end{equation}
  Proof is in Appendix \ref{sec:proof_size_AandC}.
  \end{lemma}
  
  \subsection{Overview of the proof of Theorem \ref{thm:main_result}} \label{sec:overview}
  
  For any error budget $p \in (0,1/2)$ and observation rate parameter $r\in(0,C_{\mathrm{Shannon}}(p))$, the goal of our proof of Theorem \ref{thm:main_result} is to prove that Alice and Bob can communicate at rate $R$ that is arbitrarily close to $C_{\mathrm{Shannon}}(p)$.  Our proof idea is to prove the following slightly different statement instead: for any $p \in (0,1/2)$ and for any $r\in(0,C_{\mathrm{Shannon}}(p))$, there exists an $R\in (r, C_{\mathrm{Shannon}}(p))$ such that Alice and Bob can communicate at rate $R$. Such a (seemingly weaker) statement implies Theorem~\ref{thm:main_result} immediately, since for any $r'>r$, the achievable rate under $r'$ is a lower bound of the achievable rate under $r$. We can then let $r'\rightarrow C_{\mathrm{Shannon}}(p)$ and use the (seemingly weaker) statement to derive Theorem 1.  We now present the setup of our proof. 
  
  \textbf{Setup:} The following setup will be used throughout the proof of Theorem \ref{thm:main_result}:
  \begin{enumerate}
  \item Fix any error budget $p \in(0,1/2)$ and observation rate $r \in (0,C_{\mathrm{Shannon}}(p))$, and fix observation complexity bound parameters $c,s$ to be positive integers.
  \item We can always find parameters $\delta_0, \delta_1, \epsilon_{\rho},\epsilon_R > 0$ such that the following two conditions hold:
    \begin{condition} \label{cond:small_parameters}
  $r < C_{\mathrm{Shannon}}(p) - \delta_0- \delta_1 - \epsilon_{\rho} - \epsilon_R$
  \end{condition}
 
  \begin{condition} \label{cond:small_R}
  $\epsilon_R \in (0,(\frac{5}{13} - \frac{1}{30}) \delta_0)$
  \end{condition}
  Set the inner-code rate $\rho = C_{\mathrm{Shannon}}(p) - \epsilon_{\rho}$ and the inner-outer concatenated code rate $R = \rho - \epsilon_R$. One can easily verify that the above choice of parameters guarantees that $r<R<\rho<C_{\mathrm{Shannon}}(p)$. 
  \item For blocklength $n = 1, 2, \ldots$, let the codebook $\mathcal{C}_n$ be the code construction described in Section \ref{sec:ach_scheme}.
  \item Fix the $(\rho n,R n)$ outer codebook $\mathcal{C}_{\mathrm{out}}$ to be any $1$-to-$1$ function from $\{0,1\}^{R n}$ to $\{0,1\}^{\rho n}$. Let the $(n, \rho n)$ inner codebook $\mathcal{C}_{\mathrm{in}}$ be drawn from distribution $\rv{Q}(n, \rho n)$. Note that $\mathcal{C}_n = \mathcal{C}_{\mathrm{in}} \circ \mathcal{C}_{\mathrm{out}}$ is Alice's $(n,Rn)$ codebook.
  \end{enumerate}

   We now show that the rate $R$ is $(c,s)$-achievable by using a \textit{random-coding argument}, i.e., we show that for any $\epsilon_e>0$ and for large enough $n$, $\mathbb{P}_{\mathcal{C}_{\mathrm{in}}}(\bar{P}_{\rv{e}}(\mathcal{C}_{\mathrm{out}},\mathcal{C}_{\mathrm{in}})> \epsilon_e)<1$ and thus there exists a sequence of $(\rho n, R n)$ codebooks $\mathcal{C}_{\mathrm{out}}$ and $(n, \rho n)$ codebooks $\mathcal{C}_{\mathrm{in}}$ such that $\bar{P}_{\rv{e}}(\mathcal{C}_{\mathrm{out}},\mathcal{C}_{\mathrm{in}}) \leq \epsilon_e$ for all $n$ large enough. 
  
  %We say that an $(n, \rho n)$ codebook $\mathcal{C}_{\mathrm{in}} = \{\mathcal{C}_{\mathrm{in}}(1), \ldots, \mathcal{C}_{\mathrm{in}}(2^{\rho n}) \}$ is drawn from distribution $\Omega(n, \rho n)$ if the $2^{\rho n}$ codewords of $\mathcal{C}_{\mathrm{in}}$ are independent and uniformly distributed in $\{0,1\}^n$. Over the course of the random-coding argument, we let the inner code $\mathcal{C}_{\mathrm{in}}$ be drawn from distribution $\Omega(n, \rho n)$ while we keep the outer code $\mathcal{C}_{\mathrm{out}}$ fixed to be any $(\rho n,R n)$ codebook that is a $1:1$ function from $[2^{Rn}]$ to $\{0,1\}^{\rho n}$.
  
   \textbf{Random-Coding:}  In the sequel, we drop the dependency on $\mathcal{C}_{\mathrm{out}}$ from all notation due to the outer codebook being fixed. We write $\bar{P}_e(\mathcal{C}_{\mathrm{in}})$ to denote the probability of decoding error evaluated at the $(n,Rn)$ codebook $\mathcal{C}_{n} = \mathcal{C}_{\mathrm{in}} \circ \mathcal{C}_{\mathrm{out}}$.
   
   For $f_n \in \CPX$, $\bm{\psi} \in \{0,1\}^{rn}$, $\bm{e} \in \Be$ and $i \in [2^{\rho n}]$, define
   \begin{equation}
   q_i(f_n,\bm{\psi},\bm{e},\mathcal{C}_{\mathrm{in}}) = \mathbb{P}_{m_0}(\bm{w}_i \in \mathcal{L}_{\mathrm{in}}, \bm{w}_i \in \mathcal{I}_{m_0}| \bm{\Psi}(m_0)=\bm{\psi})
   \end{equation}
   to be the probability that word $\bm{w}_i(m_0,\bm{e},\mathcal{C}_{\mathrm{in}})$ results in a decoding error given that the adversary observes $\bm{\Psi}(m_0)=\bm{\psi}$. To apply the random-coding argument, we first apply a simple union bound to $\bar{P}_e(\mathcal{C}_{\mathrm{in}})$ in (\ref{eq:prob_error_inside}) to bound the quantity above by 
   \begin{equation} \label{eq:prob_error_bd}
   \bar{P}^{\mathrm{ub}}_{e} (\mathcal{C}_{\mathrm{in}}) =\max_{f_n \in \CPX} \sum_{i=1}^{2^{\rho n}} \mathbb{E}_{\bm{\Psi}} \left[ \max_{\bm{e} \in \Be} q_i(f_n,\bm{\psi},\bm{e},\mathcal{C}_{\mathrm{in}}) \right].
   \end{equation}
   
  We now prepare to state a sufficient condition for the rate $R$ to be $(c,s)$-achievable. Recall that $\epsilon_\rho$, $\delta_0$, and $\delta_1$ are the parameters used to construct $\mathcal{C}_{\mathrm{in}}$ and $\mathcal{C}_{\mathrm{out}}$. For integer $L \in (0,2^{\rho n}]$, define the product set $\mathcal{P}(L) =  [L] \times \CPX \times \{0,1\}^{rn} \times \Be$. For $\epsilon_e>0$, we define the set $\mathcal{H}(L,\epsilon_e)$ to be the set of all $(n, \rho n)$ codebooks $\mathcal{C}_{\mathrm{in}}$ such that for all $(i,f_n,\bm{\psi},\bm{e}) \in \mathcal{P}(L)$, either $\mathbb{P}_{m_0}(\bm{\Psi}(m_0) = \bm{\psi}) < 2^{(\delta_0+ \delta_1 - R)n}$ or $q_i(f_n,\bm{\psi},\bm{e},\mathcal{C}_{\mathrm{in}}) \leq \frac{\epsilon_e}{2L}$.
  
  The intuition behind the definition of $\mathcal{H}(L,\epsilon_e)$ is as follows. For any observation-function/observation-pair $(f_n, \bm{\psi})$, we say that this pair is \textit{informative} if $\rv{\mathbb{P}}_{m_0}(\bm{\Psi}(m_0)=\bm{\psi})<2^{(\delta_0+\delta_1-R)n}$. Namely, if the adversary picks $f_n$ and observes $\bm{\Psi}(m_0) = \bm{\psi}$, then there are not many other messages $m \neq m_0$ such that $\bm{\Psi}(m) = \bm{\psi}$. As a result, the adversary knows that the true message $m_0$ must be in a very small set of possibilities, thus the name "informative". The set $\mathcal{H}(L,\epsilon_e)$ then considers the set of inner codebooks such that for any $(f_n,\psi)$ that is \textit{not} informative, no matter how the adversary designs the error vector $\bm{e}$, with high probability $1-\epsilon_e/(2L)$, each of the $L$ inner codewords that are closest to $\bm{y}=\bm{x} \rv{\oplus} \bm{e}$ is either outside the Hamming ball $B_{pn}(\bm{y})$ or can be ruled out by the outer codebook $\mathcal{C}_{\mathrm{out}}$.  

  Given the above intuition, we may consider any codebook in $\mathcal{H}(L,\epsilon_e)$ to be a good choice of $\mathcal{C}_{\mathrm{in}}$.  The reason is that when the pair $(f_n,\bm{\psi})$ is informative, the adversary knows very accurately which message is likely to be $m_0$ and thus it is hard to keep the error probability small. However, $\mathcal{H}(L,\epsilon_e)$ ensures that under a more favorable situation in which the $(f_n,\bm{\psi})$ is not informative, the inner codebook $\mathcal{C}_{\mathrm{in}}$ can take advantage of this ambiguity at the adversary and guarantee small error probability for the $L$ closest inner codewords (thus the enumerating index $i$) and regardless of how the adversary chooses the error vector $\bm{e}$.
  
  \begin{lemma}[Sufficient Condition for Achievability] \label{thm:overview_goal}
  Let $L > 1/\epsilon_{\rho}$ be a constant.  If for any $\epsilon_e>0$, there exists an $n_0$ such that for all $n \geq n_0$, the probability $\mathbb{P}_{\mathcal{C}_{\mathrm{in}}} (\mathcal{C}_{\mathrm{in}} \not\in \mathcal{H}(L,\epsilon_e)) < 1-1/n$, then the rate $R$ is $(c,s)$-achievable.
  \end{lemma}
  
  \begin{proof}
  Let $L>1/\epsilon_{\rho}$ and let $\epsilon_e>0$. Consider the probability
  \begin{equation} \label{eq:prob_ov_goal}
  \mathbb{P}_{\mathcal{C}_{\mathrm{in}}}\bigg(\mathcal{C}_{\mathrm{in}} \text{ is not } [L,p] \text{ list decodable or }  \mathcal{C}_{\mathrm{in}} \not\in \mathcal{H}(L,\epsilon_e) \bigg).
  \end{equation}
  By a simple union bound, probability (\ref{eq:prob_ov_goal}) is bounded above by 
  \begin{equation} \label{eq:prob_ov_goal_ub}
  \mathbb{P}_{\mathcal{C}_{\mathrm{in}}}(\mathcal{C}_{\mathrm{in}} \text{ is not } [L,p] \text{ list dec.}) + \mathbb{P}_{\mathcal{C}_{\mathrm{in}}}(\mathcal{C}_{\mathrm{in}} \not\in \mathcal{H}(L,\epsilon_{\rho})).
  \end{equation}
  By Lemma \ref{thm:LD_lb_bounded}, there exists an $n_1$ such that for all $n\geq n_1$, the first term in equation (\ref{eq:prob_ov_goal_ub}) is bounded above by $1/n$. In turn, since for all $n \geq n_0$ the second term in equation (\ref{eq:prob_ov_goal_ub}) is strictly smaller than $1-1/n$, it follows that for all $n \geq \max \{n_0,n_1\}$  probability (\ref{eq:prob_ov_goal}) is strictly less than $1$. Thus, for each $n\geq \max \{n_0,n_1\}$, there exists an $(n, \rho n)$ codebook $\mathcal{C}^*_{\mathrm{in}}$ such that $\mathcal{C}_{\mathrm{in}}^*$ is $[L,p]$ list decodable and $\mathcal{C}^*_{\mathrm{in}} \in \mathcal{H}(L,\epsilon_e)$. 
  
  The above shows the existence of a special codebook $\mathcal{C}_{\mathrm{in}}^*$. In the following, we show that the error probability evaluated at $\mathcal{C}_{\mathrm{in}}^*$ can be upper bounded analytically. Specifically, since $\mathcal{C}_{\mathrm{in}}^*$ is $[L,p]$ list decodable, we have the identity $q_i(f_n,\bm{\psi},\bm{e},\mathcal{C}^*_{\mathrm{in}})=0$ for all $i > L \geq \max_{\bm{y} \in \{0,1\}^n} |\mathcal{L}_{\mathrm{in}}(\bm{y},\mathcal{C}_{\mathrm{in}})|$, and therefore, $\bar{P}^{\mathrm{ub}}_e(\mathcal{C}^*_{\mathrm{in}})$ in (\ref{eq:prob_error_bd}) is equal to
  \begin{equation} \nonumber
  \max_{f_n \in \CPX} \sum_{i=1}^L \mathbb{E}_{\bm{\Psi}} \left[ \max_{\bm{e} \in \Be} q_i(f_n,\bm{\psi},\bm{e},\mathcal{C}_{\mathrm{in}}^*) \right].
  \end{equation}
  For any fixed $f_n \in \CPX$ and fixed $i\in[1,L]$, we have
  \begin{align}
  & \mathbb{E}_{\bm{\Psi}} \left[\max_{\bm{e} \in \Be} q_i(f_n,\bm{\psi},\bm{e},\mathcal{C}_{\mathrm{in}}^*) \right] \nonumber \\
  & =\sum_{\substack{\bm{\psi}: (f_n, \bm{\psi})  \text{ is} \\ \text{not informative}}} \mathbb{P}_{m_0}(\bm{\Psi}(m_0)=\bm{\psi})\cdot \max_{\bm{e} \in \Be} q_i(f_n,\bm{\psi},\bm{e},\mathcal{C}_{\mathrm{in}}^*) \label{eq:suff_term1}\\
  & + \sum_{\substack{\bm{\psi}: (f_n, \bm{\psi})  \text{ is} \\ \text{ informative}}} \mathbb{P}_{m_0}(\bm{\Psi}(m_0)=\bm{\psi}) \cdot \max_{\bm{e} \in \Be} q_i(f_n,\bm{\psi},\bm{e},\mathcal{C}_{\mathrm{in}}^*) \label{eq:suff_term2}
  \end{align}
  for which we partition based on the events that $\bm{\psi}$ and the given $f_n$ are informative or not. Since $\mathcal{C}_{\mathrm{in}}^* \in  \mathcal{H}(L,\epsilon_e)$, by the definition of $\mathcal{H}(L,\epsilon_e)$, the first summation (\ref{eq:suff_term1}) can be upper bounded by 
  \begin{equation} \label{eq:suff_term1_ub}
  \sum_{\substack{\bm{\psi}: (f_n,\psi) \text{ is} \\ \text{not informative}}}  \mathbb{P}_{m_0}(\bm{\Psi}(m_0)=\bm{\psi}) \cdot \frac{\epsilon_e}{2L} \leq \frac{\epsilon_e}{2L}.
  \end{equation}
  Since $q(\cdot)$ is a probability, we have $\max_{ \bm{e} \in \Be} q_i(f_n,\bm{\psi},\bm{e},\mathcal{C}_{\mathrm{in}}^*)\leq 1$. By the definition of $\mathcal{H}(L,\epsilon_e)$ the second summation (\ref{eq:suff_term2}) can be upper bounded by 
  \begin{equation} \label{eq:suff_term2_ub}
  \sum_{ \bm{\psi}: (f_n, \psi) \text{ is informative}} \mathbb{P}_{m_0}(\bm{\Psi}(m_0)=\bm{\psi})\leq 2^{r} 2^{(\delta_0+\delta_1-R)n}.
  \end{equation}
  By (\ref{eq:suff_term1_ub}) and (\ref{eq:suff_term2_ub}) and by summing over $i=1, \ldots, L$, we have that $\bar{P}_e^{\mathrm{ub}}(\mathcal{C}_{\mathrm{in}}^*)$ is bounded above by \begin{equation} \label{eq:suff_total_ub}
   \max_{f_n\in \CPX } L \left( \frac{\epsilon_e}{2L}+2^{(r+\delta_0+\delta_1-R)n} \right).
  \end{equation}
  Following Condition \ref{cond:small_parameters}, the exponent $r + \delta_0 + \delta_1 - R$ is strictly negative, and thus for large enough $n$, the quantity (\ref{eq:suff_total_ub}) is bounded above by $\epsilon_e$. In conclusion, for large enough $n$, $\bar{P}_e(\mathcal{C}^*_{\mathrm{in}}) \leq \epsilon_e$.

  %following the fact that $\mathcal{C}_{\mathrm{in}}^* \in \mathcal{H}(L,\epsilon_e)$ and thus for all $(i,f_n,\bm{\psi},\bm{e}) \in \mathcal{P}(L)$ either $q_i(f_n,\bm{\psi},\bm{e},\mathcal{C}_{\mathrm{in}}^*) \leq \frac{\epsilon_e}{2L}$ or $\mathbb{P}_{m_0}(\bm{\Psi}(m_0) = \bm{\psi}) < 2^{(\delta_0+ \delta_1 - R)n}$, is bounded above by
  %\begin{equation} \nonumber
  %\begin{aligned}
  %& \sum_{i=1}^L \mathbb{E}_{\bm{\Psi}} \left[\frac{\epsilon_e}{2L} \right] + \sum_{i=1}^L\mathbb{E}_{\bm{\Psi}}  \left[\mathds{1}_{\bm{\Psi}} \right] \\
  %& = \epsilon_e/2 + L \sum_{\bm{\psi} \in \{0,1\}^{rn}} \mathds{1}_{\bm{\psi}} \mathbb{P}_{m_0}(\bm{\Psi}(m_0)=\bm{\psi}) \\
  %& \text{where } \mathds{1}_{\bm{\psi}} = \mathds{1}\{  \mathbb{P}_{m_0}(\bm{\Psi}(m_0)=\bm{\psi}) < 2^{(\delta_0+ \delta_1 - R)n} \}
  %\end{aligned}
  %\end{equation}
  %which in turn is bounded above by $\epsilon_e/2 + L2^{(r + \delta_0 + \delta_1 - R)n}$. Following Condition \ref{cond:small_parameters}, the exponent $r + \delta_0 + \delta_1 - R$ is strictly negative, and thus for large enough $n$, the quantity $\epsilon_e/2 + L2^{(r + \delta_0 + \delta_1 - R)n}$ is bounded above by $\epsilon_e$. In conclusion, for large enough $n$, $\bar{P}_e(\mathcal{C}^*_{\mathrm{in}}) \leq \epsilon_e$.

  \end{proof}
  
  As a result, for $L>1/\epsilon_{\rho}$ and $\epsilon_e>0$, our strategy will be to lower bound the probability $\mathbb{P}_{\mathcal{C}_{\mathrm{in}}}(\mathcal{C}_{\mathrm{in}} \in \mathcal{H}(L,\epsilon_e))$ and apply Lemma \ref{thm:overview_goal}. In this strategy, a significant amount of work is needed to show that the following statement holds with probability greater than $1/n$ over the choice of $\Cin$: 
  \begin{equation} \label{eq:ovv_max_q}
  \max_{\substack{(i,f_n,\bm{\psi},\bm{e}) \in \mathcal{P}(L): \\ (f_n\bm{\psi}) \text{ is not informative}}} q_i(f_n,\bm{\psi},\bm{e},\Cin) \leq \frac{\epsilon_e}{2L}.
  \end{equation}
  The first step in this work is to show that each of the  $q_i(f_n,\bm{\psi},\bm{e},\Cin)$ terms in (\ref{eq:ovv_max_q}) has a small expectation (w.r.t $\Cin$), i.e.,
  \begin{equation} \label{eq:ovv_ex_lim}
  \lim_{n \rightarrow \infty}\max_{(i,f_n,\bm{\psi},\bm{e}) \in \mathcal{P}(L)} \mathbb{E}_{\Cin}[q_i(f_n,\bm{\psi},\bm{e},\Cin)] = 0.
  \end{equation} 
  We prove this result in Lemma \ref{thm:expected_t_new}. The next step is to show that each of the $q_i(f_n,\bm{\psi},\bm{e},\Cin)$ terms is \textit{concentrated} around its expectation $\mathbb{E}_{\Cin}[q_i(f_n,\bm{\psi},\bm{e},\Cin)]$, i.e., for any $\epsilon_e' \in (0,\frac{\epsilon_e}{2L})$, for large enough $n$ and with probability greater than $1/n$ over the choice of $\Cin$, the following inequality holds:
  \begin{equation} \label{eq:ovv_max_conc}
  \begin{aligned}
  & \max_{\substack{(i,f_n,\bm{\psi},\bm{e}) \in \mathcal{P}(L): \\ (f_n,\bm{\psi}) \text{ not informative}}} \left( q_i(f_n,\bm{\psi},\bm{e},\Cin) - \mathbb{E}_{\Cin}[q_i(f_n,\bm{\psi},\bm{e},\Cin)] \right) \\
  & \hspace{21em} \leq  \epsilon_e'.
  \end{aligned}
  \end{equation}
  The bulk of our proof is dedicated towards this step. Since $\epsilon'_e < \epsilon_e/2L$ with strict inequality, (\ref{eq:ovv_ex_lim}) and (\ref{eq:ovv_max_conc}) together imply that (\ref{eq:ovv_max_q}) holds with probability strictly greater than $1/n$. In the remainder of this overview, we outline our approach for studying the concentration of measure of $q_i(f_n,\bm{\psi},\bm{e},\Cin)$.

  %We now provide an overview of the techniques used in our proof. A detailed proof of Theorem \ref{thm:main_result} is provided in Section \ref{sec:analysis}.
 
   \textbf{Concentration:}  When confusion can be avoided, we drop the notated dependencies and subscripts of $q_i(f_n,\bm{\psi},\bm{e},\mathcal{C}_{\mathrm{in}})$ and simply write $q(\mathcal{C}_{\mathrm{in}})$ to emphasize the dependency on $\mathcal{C}_{\mathrm{in}}$. For integer $L>1/\epsilon_{\rho}$, fixed $(i,f_n,\bm{\psi},\bm{e}) \in \mathcal{P}(L)$ such that $(f_n,\bm{\psi})$ is not informative, and for $n=1,2,3,\ldots$, we study the concentration of measure of $q(\mathcal{C}_{\mathrm{in}})$ around its expectation $\mathbb{E}_{\mathcal{C}_{\mathrm{in}}}[q]$ by deriving concentration inequalities from the logarithmic Sobolev inequalities, e.g., \cite{Boucheron2003ConcentrationMethod}. This method of deriving concentration inequalities is also known as the entropy method.

   At a high level, the concentration inequalities tell us that if a function $g$ from the set of $(n, \rho n)$ codebooks to the real numbers is \textit{smooth} for \textit{most} $(n, \rho n)$ codebooks, then $g$ is concentrated (around its expectation). To define ``most'', we define a subset $\mathcal{T}$ of $(n, \rho n)$ codebooks with the property that $\mathbb{P}_{\Cin}(\Cin \not\in \mathcal{T}) = \exp\{-2^{\Omega(n)} \}$ (Definition \ref{def:typ_param}). We refer to the set $\mathcal{T}$ as a typical set of $(n, \rho n)$ codebooks and say a codebook $\Cin$ is typical if $\Cin \in \mathcal{T}$. To define ``smooth'', define the variation of $g$ as
   \begin{equation} \nonumber
   V(\mathcal{C}_{\mathrm{in}}) = \sum_{j=1}^{2^{\rho n}} \mathbb{E}_{\bm{z}}|g(\mathcal{C}_{\mathrm{in}}) - g(\mathcal{C}_{\mathrm{in}}(j,\bm{z}))|^2
   \end{equation}
   where codebook $\Cin(j,\bm{z})$ is equal to $\Cin$ with the $j$th codeword replaced with the word $\bm{z}$ uniformly distributed in $\{0,1\}^n$. We say that a number $a_G>0$ is a \textit{global variation coefficient} of $g$ if for any $(n, \rho n)$ codebook $\Cin$, $V(\Cin) \leq a_G$. Similarly, we say that $a_T>0$ is a \textit{typical variation coefficient} of $g$ if for any $\Cin \in \mathcal{T}$, $V(\Cin) \leq a_T$. Finally, we say that $g$ is smooth for most codebooks if $g$ has typical and global variation coefficients that are both sufficiently small.
   
   Given the above definitions, the following statement summarizes our concentration inequalities: If $g$ has a typical variation coefficient $a_T = \exp \{ -2^{\Omega(n)}\}$ and a global variation coefficient $a_G = O(1)$, then\footnote{See Lemma \ref{thm:global_typical_conc} for additional conditions on $a_T$ and $a_G$.}
   \begin{equation} \nonumber
   \mathbb{P}_{\Cin}(g - \mathbb{E}_{\Cin}[g] > \epsilon_e') = \exp \{ -2^{\Omega(n)} \}.
   \end{equation}
   Three remarks are at hand. First, the double-exponential bound ensures that a union bound can be successfully applied to the probability that event (\ref{eq:ovv_max_conc}) occurs (more on this below). Second, the $O(1)$ global variation coefficient prevents the inequalities from blowing up over the set $\mathcal{T}^c$. Lastly, these inequalities cannot be directly applied in our setting to show that $q$ is concentrated. This last remark follows from the fact that while one can find a typical variation coefficient of $q$ that is $\exp\{-2^{\Omega(n)}\}$, it is difficult to find a global variation coefficient of $q$ that is $O(1)$. To circumvent this issue, we proceed with the following additional steps.

   \begin{itemize}
   \item For an $(n, \rho n)$ codebook $\Cin$, we define an approximation function $q'(\Cin)$ to approximate $q(\Cin)$ (Definition \ref{def:approx_q}). The approximation function has the following properties:
   \begin{itemize}
   \item For any $(n, \rho n)$ codebook $\Cin$, $q'(\Cin) \leq q(\Cin)$ which holds with equality if $\Cin \in \mathcal{T}$. Hence, the function $q'$ is a good approximation of $q$ in the sense that $\mathbb{P}_{\Cin}(q \neq q') \leq \mathbb{P}_{\Cin}(\Cin \in \mathcal{T})$.
   \item The function $q'$ has a global variation coefficient that is $O(1)$ (Lemma \ref{thm:global_coeff}). We remark that the concatenated structure of our code construction simplifies the proof of this bound. %Specifically, by isolating the effect on $\bar{P}_e$ of the $i^{\text{th}}$ closest codeword in inner codebook $\Cin$ to $\bm{y}$ and subsequently performing list refinement with the outer codebook $\Cout$, we are able to show that for some $b>0$, $q'$ has a global variation coefficient of $bi = O(1)$.
   \end{itemize}
   \item Given a global variation coefficient that is $O(1)$, we show that $q'$ is concentrated, i.e., $\mathbb{P}_{\Cin}(q' - \mathbb{E}_{\Cin}[q'] > \epsilon_e') = \exp\{-2^{\Omega(n)} \}$ (Lemma \ref{thm:global_typical_conc}).
   \item We show the concentration of $q$ by proving that our special construction of $q'$ satisfies the following approximation bound: 
   \begin{align}
   & \mathbb{P}_{\Cin}(q - \mathbb{E}_{\Cin}[q] > \epsilon_e') \nonumber \\
   &\leq \mathbb{P}_{\Cin}(q' - \mathbb{E}_{\Cin}[q'] > \epsilon_e') + \mathbb{P}_{\Cin}(q \neq q') \nonumber \\ 
   &\leq \mathbb{P}_{\Cin}(q' - \mathbb{E}_{\Cin}[q'] > \epsilon_e') + \mathbb{P}_{\Cin}(\Cin \not\in \mathcal{T}) \nonumber
   \end{align}
   (Lemma \ref{thm:approx_ineq}). It follows that $\mathbb{P}_{\Cin}(q - \mathbb{E}_{\Cin}[q] > \epsilon_e') = \exp \{-2^{\Omega(n)}\}$.
   \end{itemize}
   
   To complete the proof that inequality (\ref{eq:ovv_max_q}) holds with probability greater than $1/n$, we apply a simple union bound:
   \begin{align} \nonumber
   & \mathbb{P}_{\Cin} \left(\max_{\substack{(i,f_n,\bm{\psi},\bm{e}) \in \mathcal{P}(L): \\ (f_n,\bm{\psi}) \text{ not inform.}}} \left(q-\mathbb{E}[q] \right) > \epsilon_e' \right) \\
   & \hspace{4em} \leq  |\mathcal{P}(L)| \cdot \max_{\substack{(i,f_n,\bm{\psi},\bm{e}) \in \mathcal{P}(L): \\ (f_n,\bm{\psi}) \text{ not inform.}}} \mathbb{P}_{\Cin}(q - \mathbb{E}[q] > \epsilon_e') \nonumber \\
   & \hspace{4em} \leq |\mathcal{P}(L)| \exp \{-2^{\Omega(n)} \} \label{eq:plzbtl}
   \end{align}
   where $|\mathcal{P}(L)|$ denotes the number of elements in the product space $[L] \times \CPX \times \{0,1\}^{\rho n} \times \mathcal{B}_{pn}(0)$. The final step is to show that (\ref{eq:plzbtl}) is bounded above by $1-1/n$, which we show by verifying that $|\mathcal{P}(L)| = 2^{\mathrm{poly}(n)}$ and thus $|\mathcal{P}(L)|$ is growing much slower than the double exponential $\exp\{ 2^{\Omega(n)}\}$. The key idea in this final step is to use the adversary's computational bound and show that the number of functions in $\CPX$ is $2^{\mathrm{poly}(n)}$. We remark that the bounded observation complexity is critical to the proof since if we allow for unbounded circuit complexity, the set $\textset{CKT}(r,\infty)$ contains $\exp \{2^{\Omega(n)} \}$ functions.
   
   %We say that $q(\mathcal{C}_{\mathrm{in}})$ is smooth for \textit{most} $(n, \rho n)$ codebooks if for some set $\mathcal{T}$ of typical $(n, \rho n)$ codebooks we have that a) for all $\mathcal{C}_{\mathrm{in}} \in \mathcal{T}$, $q$ is smooth at $\mathcal{C}_{\mathrm{in}}$, and b) $\mathbb{P}_{\mathcal{C}_{\mathrm{in}}} (\mathcal{C}_{\mathrm{in}} \in \mathcal{T})$ is close to $1$.

   \textbf{Prior work:} The above approach is inspired by Langberg's framework \cite{Langberg2008ObliviousCapacity} to study concentration of measure when the function under analysis is smooth over a typical set $\mathcal{T}$ of codebooks. The main technical contribution of \cite{Langberg2008ObliviousCapacity} is to carefully define $\mathcal{T}$ based on the codebooks' list decodable properties in a way where one can then apply Vu's martingale-type concentration inequalities for non-smooth functions \cite{Vu2002ConcentrationApplications}. We follow Langberg's framework by also defining typicality in terms of list decodability. However, we use concentration inequalities derived via the entropy method.
   
   The major technical difference between our work and Langberg's \cite{Langberg2008ObliviousCapacity} lies at the definition of smoothness. Langberg adopts a Lipschitz criterion of smoothness which states that a function $g$ is smooth if $g$ has a sufficiently small \textit{typical Lipschitz coefficient} $K_T>0$; $K_T$ is said to be a typical Lipschitz coefficient if for any $\Cin \in \mathcal{T}$, the quantity
   \begin{equation} \label{eq:def_lipW} \nonumber
   W(\mathcal{C}_{\mathrm{in}}) = 2^{\rho n} \max_{j \in [2^{\rho n}], \bm{z} \in \{0,1\}^{n}} |g(\mathcal{C}_{\mathrm{in}}) - g(\mathcal{C}_{\mathrm{in}}(j,\bm{z}))|^2
   \end{equation}
   is bounded above by $K_T$. Similarly, a number $K_G$ is said to be a global Lipschitz coefficient if for any $(n, \rho n)$ codebook $\Cin$, $W(\Cin) \leq K_G$. Our work identifies and exploits two advantages of using the variation criterion over the Lipschitz criterion for characterizing smoothness in our setting. First, for a typical codebook $\Cin \in \mathcal{T}$, variation $V(\Cin)$ captures more information about the behavior of $g$ locally around codebook $\mathcal{C}_{\mathrm{in}}$ than $W(\Cin)$. We leverage this additional information to find typical variation coefficients for $q'$ that are smaller than any typical Lipschitz coefficient. Second, for a non-typical codebook $\Cin \not\in \mathcal{T}$, the best bound on $W(\Cin)$ is $O(2^{\rho n})$. Thus, a good global Lipschitz coefficient of $q'$ is much larger than the $O(1)$ global variation coefficient established in our proof.
   
   Similar to our work and the work of \cite{Langberg2008ObliviousCapacity}, other works adopt proof techniques that are combinatorial in nature. These include the studies by Csisz\'{a}r and Narayan \cite{Csiszar1988TheConstraints,Csiszar1989CapacityChannels} on ACVs with input and state constraints which adopt a method-of-types approach. We note that the channel controlled by a $\CPX$-oblivious adversary with error budget $p$ can be formulated as an AVC with state constraints. \rv{These works also include the study by Dey, Jaggi and Langberg \cite{Dey2019a} on myopic adversarial channels.\footnote{\rv{We note that the proof techniques of \cite{Dey2019a} can provide a simple alternative proof of Theorem \ref{thm:main_result}. We provide an outline of this alternative proof in Appendix \ref{sec:alt_proof}.}}}

   \rv{Lastly, we remark that our proof techniques and analysis allow for generalization to other channel models. For example, straightforward modifications of our techniques/analysis can allow the bit flip channel from Alice to Bob to be generalized to a $q$-ary error/erasure channel for $q \geq 2$ in which Alice sends symbols from a $q$-ary alphabet and the adversary can induce both symbol errors and symbol erasures.}

  \section{Proof of Theorem \ref{thm:main_result}} \label{sec:analysis}
   
  In the sequel, we use the setup described in Section \ref{sec:overview}.
 
  \subsection{Notation} \label{sec:up_pe}
  
  The following notation will assist in our proof of Theorem \ref{thm:main_result}. For an observation function $f_n \in \CPX$ and observation $\bm{\psi} \in \{0,1\}^{rn}$, we define the observation set $\mathcal{O}_{\bm{\psi}} = \{\bm{z} \in \{0,1\}^n: f_n(\bm{z}) = \bm{\psi} \}$. Note that observing $\bm{\Psi}(\rv{m_0}) \triangleq f_n(\mathcal{C}_n(m_0)) = \psi$ is equivalent to knowing that $\mathcal{C}_n(m_0) \in \mathcal{O}_{\bm{\psi}}$. Hence, the following two perspectives are equivalent: a) the adversary draws an observation function $f_n$ and b) the adversary draws a partition $\vec{\mathcal{O}} = (\mathcal{O}_1, \ldots, \mathcal{O}_{2^{rn}})$ of the space $\{0,1\}^n$ consisting of $2^{rn}$ non-empty observation sets. With an abuse of notation, for $f_n \in \CPX$, we write $\vec{\mathcal{O}} \in \CPX$ to denote the partition of observation sets corresponding to $f_n$ with circuit complexity upper bounded by $cn^s$. Along the same lines, for $(i,\vec{\mathcal{O}},\bm{\psi},\bm{e}) \in \mathcal{P}(2^{\rho n})$ and for an $(n, \rho n)$ codebook $\mathcal{C}_{\mathrm{in}}$, we write $q_i(\vec{\mathcal{O}},\bm{\psi},\bm{e},\mathcal{C}_{\mathrm{in}})$ to denote $q_i(f_n,\bm{\psi},\bm{e},\mathcal{C}_{\mathrm{in}})$.

  \subsection{Expectation of $q$} \label{sec:exp_q}
  
  For any $(i,\vec{\mathcal{O}},\bm{\psi},\bm{e}) \in \mathcal{P}(2^{\rho n})$, the following result characterizes the expectation of $q_i(\vec{\mathcal{O}},\bm{\psi},\bm{e},\mathcal{C}_{\mathrm{in}})$ when the $(n, \rho n)$ codebook $\mathcal{C}_{\mathrm{in}}$ is drawn from distribution $\rv{Q}(n, \rho n)$.
   
  \begin{lemma} \label{thm:expected_t_new}
  \begin{equation} \nonumber
  \lim_{n\rightarrow \infty}\max_{(i,\vec{\mathcal{O}},\bm{\psi},\bm{e}) \in \mathcal{P}(2^{\rho n})} \mathbb{E}_{\mathcal{C}_{\mathrm{in}}}[q_i(\vec{\mathcal{O}},\bm{\psi},\bm{e},\mathcal{C}_{\mathrm{in}})] =0.
  \end{equation}
  \end{lemma} 
  
  \begin{proof}
  Our approach is to find a small upper bound of the expectation $\mathbb{E}_{\Cin}[q_i(\vec{\mathcal{O}},\bm{\psi},\bm{e},\Cin)]$ that is independent of the parameters $i$, $\vec{\mathcal{O}}$, $\bm{\psi}$ and $\bm{e}$. Let $\emptyset$ denote the empty set and let $(i,\vec{\mathcal{O}},\bm{\psi},\bm{e}) \in \mathcal{P}(2^{\rho n})$. Observe that for an $(n, \rho n)$ codebook $\mathcal{C}_{\mathrm{in}}$, $q_i(\vec{\mathcal{O}},\bm{\psi},\bm{e},\mathcal{C}_{\mathrm{in}})$ is bounded above by $\mathbb{P}_{m_0}(\mathcal{I}_{m_0} \cap \mathcal{L}_{\mathrm{in}}(\bm{y},\mathcal{C}_{\mathrm{in}}) \neq \emptyset|\bm{\Psi}(m_0) = \bm{\psi})$. Hence, for any $(i,\vec{\mathcal{O}},\bm{\psi},\bm{e})\in \mathcal{P}(2^{\rho n})$ we have
  \begin{align}
  & \mathbb{E}_{\mathcal{C}_{\mathrm{in}}}[q_i(\vec{\mathcal{O}},\bm{\psi},\bm{e},\mathcal{C}_{\mathrm{in}})] \nonumber \\
  & \leq \mathbb{E}_{\mathcal{C}_{\mathrm{in}}} \left[ \mathbb{P}_{m_0}(\mathcal{I}_{m_0} \cap \mathcal{L}_{\mathrm{in}}(m_0,\bm{e},\mathcal{C}_{\mathrm{in}}) \neq \emptyset|\bm{\Psi}(m_0) = \bm{\psi}) \right] \nonumber \\
  & = \mathbb{E}_{\mathcal{C}_{\mathrm{in}}} \left[ \mathbb{E}_{m_0|\bm{\Psi}=\bm{\psi}} \left[\mathds{1}\{ \mathcal{I}_{m_0} \cap \mathcal{L}_{\mathrm{in}}(\bm{y},\mathcal{C}_{\mathrm{in}}) \neq \emptyset\} \right] \right] \nonumber \\
  & = \mathbb{E}_{m_0|\bm{\Psi}=\bm{\psi}} \left[ \mathbb{E}_{\mathcal{C}_{\mathrm{in}}} \left[\mathds{1}\{ \mathcal{I}_{m_0} \cap \mathcal{L}_{\mathrm{in}}(\bm{y},\mathcal{C}_{\mathrm{in}}) \neq \emptyset\} \right] \right] \nonumber 
  %& = \sum_{m \in [2^{Rn}]} \mathbb{E}_{\mathcal{C}_{\mathrm{in}}} \left[\mathds{1}\{ \mathcal{I}_{m} \cap \mathcal{L}_{\mathrm{in}}(m,\bm{e},\mathcal{C}_{\mathrm{in}}) \neq \emptyset\} \big| m_0 = m \right] \mathbb{P}_{m_0}(m_0 = m | \bm{\Psi} = \psi)
  \end{align} 
  Thus, to prove Lemma \ref{thm:expected_t_new}, it is sufficient to show that the quantity
  \begin{equation} \label{thm:expected_t_new_suff}
  \max_{\substack{(m,i,\vec{\mathcal{O}},\bm{\psi},\bm{e}) \\ \in [2^{Rn}] \times \mathcal{P}(2^{\rho n})}} \mathbb{E}_{\mathcal{C}_{\mathrm{in}}} \left[\mathds{1}\{ \mathcal{I}_{m} \cap \mathcal{L}_{\mathrm{in}}(\bm{y},\mathcal{C}_{\mathrm{in}}) \neq \emptyset\} | m_0 = m \right]
  \end{equation}
  is going to zero in the limit as $n \rightarrow \infty$. This sufficient condition simplifies our problem as the expectation inside the maximum of (\ref{thm:expected_t_new_suff}) only depends on parameters $m$, $i$ and $\bm{e}$, but not on $\vec{\mathcal{O}}$ or $\bm{\psi}$. In the next steps, we will further bound the expectation in (\ref{thm:expected_t_new_suff}) to relax its dependency on parameters $i$ and $\bm{e}$ by exploiting the list-decodability properties of the random codebook $\Cin$.
  
  Let $m \in [2^{Rn}]$, $(i,\vec{\mathcal{O}},\bm{\psi},\bm{e}) \in \mathcal{P}(2^{\rho n})$ and let $\bm{y}_m = \mathcal{C}_n (m) \rv{\oplus} \bm{e}$. By the definition of the set $\mathcal{I}_m$ and a simple union bound, the quantity $\mathbb{E}_{\mathcal{C}_{\mathrm{in}}} \left[\mathds{1}\{ \mathcal{I}_{m} \cap \mathcal{L}_{\mathrm{in}}(\bm{y},\mathcal{C}_{\mathrm{in}}) \neq \emptyset\} | m_0 =m \right]$ is bounded above by 
  \begin{equation} \nonumber
  \sum_{m' \in [2^{Rn}] \setminus \{m\}} \mathbb{P}_{\mathcal{C}_{\mathrm{in}}}(\mathcal{C}_{\mathrm{out}}(m') \in \mathcal{L}_{\mathrm{in}}(\bm{y}_m,\mathcal{C}_{\mathrm{in}}))
  \end{equation}
  which in turn, by letting $\mathcal{E}$ be the event that $\mathcal{C}_{\mathrm{in}}$ is $[n^2+1,p]$ list decodable and by the law of total probability, is bounded above by
  \begin{equation} \label{eq:exp_ub_final}
  \sum_{m' \in [2^{Rn}]\setminus \{m\}} \left( \mathbb{P}_{\mathcal{C}_{\mathrm{in}}}(\mathcal{C}_{\mathrm{out}}(m') \in \mathcal{L}_{\mathrm{in}}(\bm{y}_m,\mathcal{C}_{\mathrm{in}})| \mathcal{E}) + \mathbb{P}_{\mathcal{C}_{\mathrm{in}}}(\mathcal{E}^c) \right).
  \end{equation}
  Note that for $m' \in [2^{Rn}] \setminus \{m\}$, $\mathbb{P}_{\mathcal{C}_{\mathrm{in}}}(\mathcal{C}_{\mathrm{out}}(m') \in \mathcal{L}_{\mathrm{in}}(\bm{y}_m,\mathcal{C}_{\mathrm{in}})|\mathcal{E})$ is bounded above by $\frac{n^2}{2^{\rho n}-1}$ following that the codeword $\mathcal{C}_{\mathrm{in}} \circ \mathcal{C}_{\mathrm{out}}(m')$ can be one of at most $n^2$ codewords of $\mathcal{C}_{\mathrm{in}}$ randomly chosen from $2^{\rho n}-1$ codewords (nb. we can exclude codeword $\mathcal{C}_{\mathrm{in}} \circ \mathcal{C}_{\mathrm{out}}(m))$ contained in $\mathcal{L}_{\mathrm{in}}(\bm{y}_m,\mathcal{C}_{\mathrm{in}})$ by the list decodability properties of $\mathcal{C}_{\mathrm{in}}$. Also, by Lemma \ref{thm:LD_lb}, for large enough $n$, $\mathbb{P}_{\mathcal{C}_{\mathrm{in}}}(\mathcal{E}^c)$ is bounded above by $2^{-(n^2+1)/4}$. It follows that quantity (\ref{eq:exp_ub_final}) is bounded above by $\frac{2^{Rn}-1}{2^{\rho n}-1}n^2 + (2^{Rn}-1)2^{-(n^2+1)/4}$ which in turn is going to zero in the limit as $n \rightarrow \infty$ independent of $(m,i,\vec{\mathcal{O}},\bm{\psi},\bm{e}) \in [2^{Rn}] \times \mathcal{P}(2^{\rho n})$.
  \end{proof}

  \subsection{Approximation of $q$} \label{sec:approx_q}

  For $(i,\vec{\mathcal{O}},\bm{\psi},\bm{e}) \in \mathcal{P}(2^{\rho n})$, we will not directly show that the quantity $q(\vec{\mathcal{O}},\bm{\psi},\bm{e},\mathcal{C}_{\mathrm{in}})$ (as a function of $\mathcal{C}_{\mathrm{in}}$) is concentrated. Instead, we will approximate $q(\vec{\mathcal{O}},\bm{\psi},\bm{e},\mathcal{C}_{\mathrm{in}})$ with an approximation function $q'(\vec{\mathcal{O}},\bm{\psi},\bm{e},\mathcal{C}_{\mathrm{in}})$ and study the concentration of $q'(\vec{\mathcal{O}},\bm{\psi},\bm{e},\mathcal{C}_{\mathrm{in}})$. We carefully define the approximation function such that $q'(\vec{\mathcal{O}},\bm{\psi},\bm{e},\mathcal{C}_{\mathrm{in}})$ has good smoothness properties (and thus $q'(\vec{\mathcal{O}},\bm{\psi},\bm{e},\mathcal{C}_{\mathrm{in}})$ concentrates around its mean), and such that we can imply the concentration of $q(\vec{\mathcal{O}},\bm{\psi},\bm{e},\mathcal{C}_{\mathrm{in}})$ from the concentration of $q'(\vec{\mathcal{O}},\bm{\psi},\bm{e},\mathcal{C}_{\mathrm{in}})$.
   
  We first define typical codebooks. For the parameter $\delta_0 >0$, we define typical codebooks for all $\vec{\mathcal{O}} \in \CPX$ and $\bm{\psi} \in \{0,1\}^{rn}$ such that the observation set $\mathcal{O}_{\bm{\psi}}$ is larger than $2^{(1-R)n}2^{\delta_0 n}$. We remark that the condition $|\mathcal{O}_{\bm{\psi}}| \geq 2^{(1-R)n}2^{\delta_0 n}$ is related to the pair $(f_n,\bm{\psi})$ being \rv{not} informative (see definition of informative in the overview of Section \ref{sec:overview}). In our analysis of $\bar{P}_{e}^{\mathrm{ub}}$, we will let decoding fail for all $\mathcal{O}_{\bm{\psi}}$ that are smaller than $2^{(1-R)n}2^{\delta_0 n}$. Hence, there is no need to define typical codebooks for small observation sets. 
  
  \begin{definition}[Typical Codebooks] \label{def:typ_param}
  Suppose that $\vec{\mathcal{O}} \in \CPX$ and $\bm{\psi} \in \{0,1\}^{rn}$ such that $|\mathcal{O}_{\bm{\psi}}| 
  \geq 2^{(1-R)n} 2^{\delta_0 n}$. Set $\delta_0' = \delta_0'(\vec{\mathcal{O}},\bm{\psi},\epsilon_{\rho},\epsilon_R) \geq \delta_0$ to be the unique number such that $|\mathcal{O}_{\bm{\psi}}| = 2^{(1-R)n} 2^{\delta_0' n}$. Set  
  \begin{align}
  & \ell(\vec{\mathcal{O}},\bm{\psi},\epsilon_R,\epsilon_{\rho}) = 2^{\frac{4 \delta_0'}{13}n} \nonumber \\
  & t_L(\vec{\mathcal{O}},\bm{\psi},\epsilon_R,\epsilon_{\rho}) = 2^{-(1-R)n} |\mathcal{O}_{\bm{\psi}}| - 2^{\frac{3 \delta_0'}{4}n} = 2^{\delta_0' n} - 2^{\frac{3}{4}\delta_0' n} \nonumber
  \end{align}
  and 
  \begin{equation} \nonumber
  t_U(\vec{\mathcal{O}},\bm{\psi},\epsilon_R,\epsilon_{\rho}) = 2^{-(1-R)n} |\mathcal{O}_{\bm{\psi}}| + 2^{\frac{3 \delta_0'}{4}n} = 2^{\delta_0' n} + 2^{\frac{3}{4}\delta_0' n}.
  \end{equation}
  An $(n, \rho n)$ codebook $\mathcal{C}_{\mathrm{in}}$ is said to be typical w.r.t. the parameters $\vec{\mathcal{O}}$, $\bm{\psi}$, $\epsilon_R$, $\epsilon_{\rho}$ if $\mathcal{C}_{\mathrm{in}}$ is $[\ell,p]$ list decodable and $t_L \leq |\mathcal{O}_{\bm{\psi}} \cap \mathcal{C}_n| \leq t_U$ where $\mathcal{C}_n = \mathcal{C}_{\mathrm{in}} \circ \mathcal{C}_{\mathrm{out}}$. Define the typical set $\mathcal{T}_{\mathcal{O}_{\bm{\psi}}}$ as the set of all $(n, \rho n)$ codebooks that are typical w.r.t. $\vec{\mathcal{O}}$, $\bm{\psi}$, $\epsilon_R$, $\epsilon_{\rho}$.
  \end{definition}

  We now provide an equivalent expression of $q_i(\vec{\mathcal{O}},\bm{\psi},\bm{e},\mathcal{C}_{\mathrm{in}})$ that is convenient for defining our approximation function. For $m \in [2^{Rn}]$ and for $(i,\vec{\mathcal{O}},\bm{\psi},\bm{e}) \in \mathcal{P}(2^{\rho n})$, define
  \begin{equation} \nonumber
  \begin{aligned}
  &\phi_{i,m}(\vec{\mathcal{O}}, \bm{\psi},\bm{e},\mathcal{C}_{\mathrm{in}}) \\
  & = \mathds{1}\{\bm{w}_i(m,\bm{e},\mathcal{C}_{\mathrm{in}}) \in \mathcal{I}_{m} \cap \mathcal{L}_{\mathrm{in}}(\bm{y}_m,\mathcal{C}_{\mathrm{in}})\} \mathds{1}\{\mathcal{C}_n(m) \in \mathcal{O}_{\bm{\psi}}\}
  \end{aligned}
  \end{equation}
  and define 
  \begin{equation} \nonumber
  \Phi_i(\vec{\mathcal{O}}, \bm{\psi},\bm{e},\mathcal{C}_{\mathrm{in}}) = \sum_{m \in [2^{Rn}]} \phi_{i,m}(\vec{\mathcal{O}}, \bm{\psi}, \bm{e},\mathcal{C}_{\mathrm{in}}).
  \end{equation}
  For $(i,\vec{\mathcal{O}},\bm{\psi},\bm{e}) \in \mathcal{P}(2^{\rho n})$, note that
  \begin{equation} \label{eq:define_q}
  q_i(\vec{\mathcal{O}},\bm{\psi},\bm{e},\mathcal{C}_{\mathrm{in}}) = \frac{\Phi_{i}(\vec{\mathcal{O}},\bm{\psi},\bm{e},\mathcal{C}_{\mathrm{in}})}{|\mathcal{O}_{\bm{\psi}} \cap \mathcal{C}_n|}.
  \end{equation}

  \begin{definition}[Approximation function] \label{def:approx_q}
  Suppose that $\vec{\mathcal{O}} \in \CPX$ and $\bm{\psi} \in \{0,1\}^{rn}$ such that $|\mathcal{O}_{\bm{\psi}}| 
  \geq 2^{(1-R)n} 2^{\delta_0 n}$.  For $i \in [2^{\rho n}]$, $\bm{e} \in \Be$ and $(n, \rho n)$ codebook $\mathcal{C}_{\mathrm{in}}$, define the approximation function
  \begin{equation}
  q'_i(\vec{\mathcal{O}},\bm{\psi},\bm{e},\mathcal{C}_{\mathrm{in}}) = \frac{\Phi_i(\vec{\mathcal{O}}, \bm{\psi}, \bm{e},\mathcal{C}_{\mathrm{in}})}{t(\vec{\mathcal{O}},\bm{\psi}, \mathcal{C}_{\mathrm{in}})}
  \end{equation}
  where $t(\vec{\mathcal{O}},\bm{\psi}, \mathcal{C}_{\mathrm{in}}) = \max\{|\mathcal{O}_{\bm{\psi}}\cap \mathcal{C}_n|,t_L \}$. Notice that $q_i'(\vec{\mathcal{O}},\bm{\psi}, \bm{e},\mathcal{C}_{\mathrm{in}}) \leq q_i(\vec{\mathcal{O}},\bm{\psi},\bm{e},\mathcal{C}_{\mathrm{in}})$ with equality if $\mathcal{C}_{\mathrm{in}} \in \mathcal{T}_{\mathcal{O}_{\bm{\psi}}}$. Furthermore, define the variation of $q'_i(\vec{\mathcal{O}},\bm{\psi},\bm{e},\mathcal{C}_{\mathrm{in}})$ as
  \begin{equation}
    V_i'(\vec{\mathcal{O}},\bm{\psi}, \bm{e}, \mathcal{C}_{\mathrm{in}}) = \sum_{j=1}^{2^{\rho n}} \mathbb{E}_{\bm{z}} \bigg[ \hspace{0.3em} \Delta' (j,\bm{z},\mathcal{C}_{\mathrm{in}})^2 \hspace{0.2em} \bigg]
  \end{equation}
  where the bounded difference
  \begin{equation} \nonumber
  \Delta'(j,\bm{z},\mathcal{C}_{\mathrm{in}}) = |q_i'(\vec{\mathcal{O}},\bm{\psi}, \bm{e},\mathcal{C}_{\mathrm{in}}) - q_i'(\vec{\mathcal{O}},\bm{\psi}, \bm{e},\mathcal{C}_{\mathrm{in}}(j,\bm{z}))|
  \end{equation}
  and the expectation is taken over the random variable $\bm{z}$ uniformly distributed in $\{0,1\}^n$.
  \end{definition}
  
  The above definition of $q_i'(\vec{\mathcal{O}},\bm{\psi}, \bm{e},\mathcal{C}_{\mathrm{in}})$ is carefully set such that $V_i'(\vec{\mathcal{O}},\bm{\psi}, \bm{e},\mathcal{C}_{\mathrm{in}})$ is well behaved for all non-typical $\Cin$. This behavior is established in Section \ref{sec:char_V_prime}.
  
  \subsection{Combinatorial Preliminaries}
  
    In this section, we prove a few claims about the combinatorial properties of the quantities defined thus far. These claims will be used in the following section to characterize the smoothness properties of the approximation function $q'(\cdot)$.
    
    In the sequel, unless otherwise stated, we fix integer $L > 1/\epsilon_{\rho}$, fix $(i,\vec{\mathcal{O}},\bm{\psi},\bm{e}) \in \mathcal{P}(L)$, and allow only the $(n, \rho n)$ codebook $\mathcal{C}_{\mathrm{in}}$ to vary. We drop the fixed variables from our notation. We write $q(\mathcal{C}_{\mathrm{in}})$ to denote $q_i(\vec{\mathcal{O}},\bm{\psi},\bm{e},\mathcal{C}_{\mathrm{in}})$. Similarly, we write $\mathcal{T}$ to denote $\mathcal{T}_{\mathcal{O}_{\bm{\psi}}}$, $\Phi(\mathcal{C}_{\mathrm{in}})$ to denote $\Phi_i(\vec{\mathcal{O}},\bm{\psi},\bm{e},\mathcal{C}_{\mathrm{in}})$, $\phi_{m}(\mathcal{C}_{\mathrm{in}})$ to denote $\phi_{i,m}(\vec{\mathcal{O}},\bm{\psi},\bm{e},\mathcal{C}_{\mathrm{in}})$, $t(\mathcal{C}_{\mathrm{in}})$ to denote $t(\vec{\mathcal{O}},\bm{\psi},\mathcal{C}_{\mathrm{in}})$, $V(\mathcal{C}_{\mathrm{in}})$ to denote $V_i(\vec{\mathcal{O}},\bm{\psi},\bm{e},\mathcal{C}_{\mathrm{in}})$ and $V'(\mathcal{C}_{\mathrm{in}})$ to denote $V_i'(\vec{\mathcal{O}},\bm{\psi},\bm{e},\mathcal{C}_{\mathrm{in}})$.
  
    The following notation will be used throughout this section. For $m \in [2^{Rn}]$, let $\bm{y}_m = \mathcal{C}_{n}(m) \rv{\oplus} \bm{e}$. For $k = 1, \ldots, 2^{\rho n}$, let the notation $j_k$ denote the the index $\textbf{int}(\bm{w}_k(m,\bm{e},\mathcal{C}_{\mathrm{in}})) \in [2^{\rho n}]$.
    
    For $m \in [2^{Rn}]$ and $(n, \rho n)$ codebook $\Cin$, the first two claims characterize how the $i$th closest codeword in codebook $\Cin$ to received word $\bm{y}_m$ (i.e., $\Cin(\bm{w}_i(m,\bm{e},\Cin))$) changes when the $j_{k}^{\text{th}}$ codeword in $\Cin$ is replaced with a word $\bm{z} \not\in \mathcal{B}_{pn}(\bm{y}_m)$. The proof of these claims only require the definition (\ref{eq:error_idx_ineq}) of word $\bm{w}_i(m,\bm{e},\Cin)$.
  
  \begin{claim} \label{claim:ub_pn}
  Let $\mathcal{C}_{\mathrm{in}}$ be an $(n, \rho n)$ codebook, $m \in [2^{Rn}]$, $\bm{z} \not\in \mathcal{B}_{pn}(\bm{y}_m)$ and $k \in [2^{\rho n}]$. Let $C_{in}'$ denote the codebook $\mathcal{C}_{\mathrm{in}}(j_k,\bm{z})$. If $\mathcal{C}_{\mathrm{in}}(j_k) \not\in \mathcal{B}_{pn}(\bm{y}_m)$ and
  \begin{equation}  \label{eq:ub_pn_hyp} \nonumber
    \mathcal{C}_{\mathrm{in}}(\bm{w}_i(m,\bm{e},\mathcal{C}_{\mathrm{in}})) \not\in \mathcal{B}_{pn}(\bm{y}_m),
    \end{equation}
     then
    \begin{equation} \label{eq:ub_pn_result}
    \mathcal{C}_{\mathrm{in}}'(\bm{w}_i(m,\bm{e},\mathcal{C}_{\mathrm{in}}')) \not\in \mathcal{B}_{pn}(\bm{y}_m).
    \end{equation}
  \end{claim}
  
  \begin{figure}[t]
  \includegraphics[width=\columnwidth]{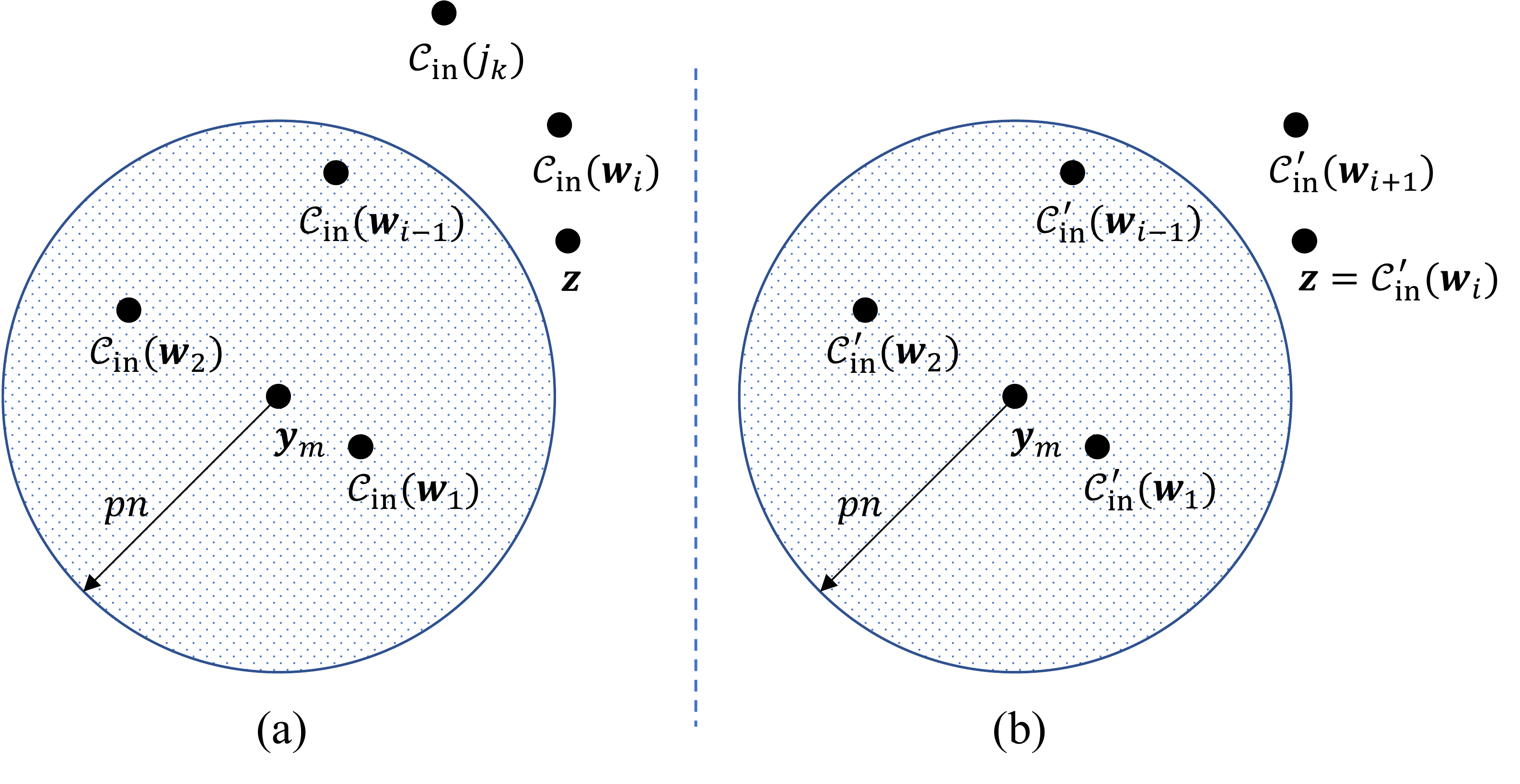}
  \caption{Location of the $i^{th}$ closest codeword to $\bm{y}_m$ (a) before replacing codeword $\mathcal{C}_{\mathrm{in}}(j_k)$ with word $\bm{z}$ and (b) after replacing codeword $\mathcal{C}_{\mathrm{in}}(j_k)$ with word $\bm{z}$. In this figure, the codebook $\mathcal{C}'_{\mathrm{in}}$ is equal to $\Cin(j_k,\bm{z})$.}
  \label{fig:w_change}
  \end{figure}
  
  \begin{proof}
  The location of codewords $\mathcal{C}_{\mathrm{in}}(\bm{w}_i(m,\bm{e},\Cin))$ and $\mathcal{C}_{\mathrm{in}}'(\bm{w}_i(m,\bm{e},\Cin'))$ around $\bm{y}_m$ are illustrated in Fig. \ref{fig:w_change}. We begin by observing that $\Cin(\bm{w}_i(m,\bm{e},\Cin)) \not\in \mathcal{B}_{pn}(\bm{y}_m)$ implies that $|\mathcal{C}_{\mathrm{in}} \cap \mathcal{B}_{pn}(\bm{y}_m)| \leq i-1$. Together with the fact that $\bm{z} \not\in \mathcal{B}_{pn}(\bm{y}_m)$ and $\mathcal{C}_{\mathrm{in}}(j_k) \not\in \mathcal{B}_{pn}(\bm{y}_m)$, it follows that $|\mathcal{C}_{\mathrm{in}}' \cap \mathcal{B}_{pn}(\bm{y}_m)|=|\mathcal{C}_{\mathrm{in}} \cap \mathcal{B}_{pn}(\bm{y}_m)| \leq i-1$. This implies equation (\ref{eq:ub_pn_result}).
  \end{proof}
  
  %\begin{claim} \label{claim:split1}
  %  Let $\mathcal{C}_{\mathrm{in}}$ be an $(n, \rho n)$ codebook, $m \in [2^{Rn}]$, $\bm{z} \in \{0,1\}^n$ and $k \in [2^{\rho n}]$ such that $\bm{w}_k(m,\bm{e},\mathcal{C}_{\mathrm{in}})\neq \mathcal{C}_{\mathrm{out}}(m)$. Let $C_{in}'$ denote the codebook $\mathcal{C}_{\mathrm{in}}(j_k,\bm{z})$. If $i<k$ and
  %  \begin{equation} \nonumber
  %  d(\bm{y}_m,\mathcal{C}_{\mathrm{in}}(\bm{w}_i(m,\bm{e},\mathcal{C}_{\mathrm{in}}))) \geq  d(\bm{y}_m,\bm{z}),
  %  \end{equation}
  %   then
  %  \begin{equation} \nonumber
  %  d(\bm{y}_m,\mathcal{C}_{\mathrm{in}}'(\bm{w}_i(m,\bm{e},\mathcal{C}_{\mathrm{in}}'))) \geq  d(\bm{y}_m,\bm{z}).
  %  \end{equation}
  %\end{claim}
  
  \begin{claim} \label{claim:split2}
    Let $\mathcal{C}_{\mathrm{in}}$ be an $(n, \rho n)$ codebook, $m \in [2^{Rn}]$, $\bm{z} \in \{0,1\}^n$ and $k \in \{i+1,\ldots,2^{\rho n}\}$ such that $\bm{w}_k(m,\bm{e},\mathcal{C}_{\mathrm{in}})\neq \mathcal{C}_{\mathrm{out}}(m)$. If 
    \begin{equation} \nonumber
    d(\bm{y}_m,\mathcal{C}_{\mathrm{in}}(\bm{w}_i(m,\bm{e},\mathcal{C}_{\mathrm{in}}))) < d(\bm{y}_m,\bm{z}),
    \end{equation}
    then $\bm{w}_i(m,\bm{e},\mathcal{C}_{\mathrm{in}}) = \bm{w}_i(m,\bm{e},\mathcal{C}_{\mathrm{in}}(j_k,\bm{z}))$.
  \end{claim}
  
  \begin{proof}
  The condition $\bm{w}_k(m,\bm{e},\mathcal{C}_{\mathrm{in}}) \neq \mathcal{C}_{\mathrm{out}}(m)$ ensures that $\bm{y}_m = \Cin \circ \Cout (m) \rv{\oplus} \bm{e}$ is equal to $\Cin(j_k,\bm{z}) \circ \Cout(m) \rv{\oplus} \bm{e}$, and thus the center of the ball of radius $pn$ around the received word does not change when the $j_k^{\text{th}}$ codeword of $\Cin$ is replaced with word $\bm{z}$. For $t = 1,\ldots, \rv{i}$, note that the $t^{\text{th}}$ closest codeword in $\Cin$ to $\bm{y}_m$ (i.e., $\Cin(\bm{w}_t(m,\bm{e},\Cin))$) is at least as close to $\bm{y}_m$ as codeword $\Cin(j_k)$, and is closer to $\bm{y}_m$ than word $\bm{z}$. Therefore, by replacing the codeword $\mathcal{C}_{\mathrm{in}}(j_k)$ with the word $\bm{z}$, we do not change the position of the $t^{\text{th}}$ closest codeword in $\mathcal{C}_{\mathrm{in}}$ to word $\bm{y}_m$. Hence, $\bm{w}_t(m,\bm{e},\mathcal{C}_{\mathrm{in}})= \bm{w}_t(m,\bm{e},\mathcal{C}_{\mathrm{in}}(j_k,\bm{z}))$.
  \end{proof}

  For $m \in [2^{Rn}]$ and $(n, \rho n)$ codebook $\Cin$, the next two claims build upon the first two claims and characterize how the term $\phi_m(\Cin)$ changes (and in turn, how the approximation function $q'(\Cin)$ changes) when the $j_k^{\text{th}}$ codeword in $\Cin$ is replaced with a word $\bm{z} \not\in \mathcal{B}_{pn}(\bm{y}_m)$. In the following section, these claims will help us in bounding the bounded difference $\Delta'(j,\bm{z},\mathcal{C}_{\mathrm{in}})$ and variation $V'(\mathcal{C}_{\mathrm{in}})$.
  
  \begin{claim} \label{claim:eq_phi_2}
  Let $\mathcal{C}_{\mathrm{in}}$ be an $(n, \rho n)$ codebook, $m \in [2^{Rn}]$, $\bm{z} \notin \mathcal{B}_{pn}(\bm{y}_m)$ and $k \in [2^{\rho n}]$ such that $\bm{w}_k(m,\bm{e},\mathcal{C}_{\mathrm{in}})\neq \mathcal{C}_{\mathrm{out}}(m)$. If either $k>i$ or $\mathcal{C}_{\mathrm{in}}(j_k) \not\in \mathcal{B}_{pn}(\bm{y}_m)$, then
  \begin{equation} \label{eq:eq_phi_2}
  \begin{aligned}
  & \mathds{1}\{\bm{w}_i(m,\bm{e},\mathcal{C}_{\mathrm{in}}) \in \mathcal{I}_{m} \cap \mathcal{L}_{\mathrm{in}}(\bm{y}_m,\mathcal{C}_{\mathrm{in}})\} \\
  & = \mathds{1}\{\bm{w}_i(m,\bm{e},\mathcal{C}_{\mathrm{in}}(j_k,\bm{z})) \in \mathcal{I}_{m} \cap \mathcal{L}_{\mathrm{in}}(\bm{y}_m,\mathcal{C}_{\mathrm{in}}(j_k,\bm{z}))\}.
  \end{aligned}
  \end{equation}
  \end{claim}
  
  \begin{proof}
  We consider 2 cases depending on the distance of codeword $\Cin(\bm{w}_i(m,\bm{e},\Cin))$ from received word $\bm{y}_m$. (Case 1): Suppose that $\mathcal{C}_{\mathrm{in}}(\bm{w}_i(m,\bm{e},\mathcal{C}_{\mathrm{in}})) \not\in \mathcal{B}_{pn}(\bm{y}_m)$ (i.e., $\bm{w}_i(m,\bm{e},\mathcal{C}_{\mathrm{in}}) \not\in \mathcal{L}_{\mathrm{in}}(\bm{y}_m,\mathcal{C}_{\mathrm{in}})$). Note that by hypothesis or by the condition that $k>i$, it follows that $\mathcal{C}_{\mathrm{in}}(j_k) \not\in \mathcal{B}_{pn}(\bm{y}_m)$, and in turn by Claim \ref{claim:ub_pn}, we have for $\mathcal{C}_{\mathrm{in}}'=\mathcal{C}_{\mathrm{in}}(j,\bm{z})$ that $\mathcal{C}_{\mathrm{in}}'(\bm{w}_i(m,\bm{e},\mathcal{C}_{\mathrm{in}}')) \not\in \mathcal{B}_{pn}(\bm{y}_m)$. Furthermore, $\bm{w}_i(m,\bm{e},\mathcal{C}_{\mathrm{in}}(j,\bm{z})) \not\in \mathcal{L}_{\mathrm{in}}(\bm{y}_m,\mathcal{C}_{\mathrm{in}}(j,\bm{z}))$. It follows that both sides of equation (\ref{eq:eq_phi_2}) are $0$, and thus, Claim \ref{claim:eq_phi_2} holds in this Case. (Case 2): Suppose that $\mathcal{C}_{\mathrm{in}}(\bm{w}_i(m,\bm{e},\mathcal{C}_{\mathrm{in}})) \in \mathcal{B}_{pn}(\bm{y}_m)$ (i.e., $\bm{w}_i(m,\bm{e},\mathcal{C}_{\mathrm{in}}) \in \mathcal{L}_{\mathrm{in}}(\bm{y}_m,\mathcal{C}_{\mathrm{in}})$). Then $k>i$ and $d(\bm{y}_m,\mathcal{C}_{\mathrm{in}}(\bm{w}_i(m,\bm{e},\mathcal{C}_{\mathrm{in}}))) < d(\bm{y}_m,\bm{z})$, and in turn, following Claim \ref{claim:split2}, we have that $\bm{w}_i(m,\bm{e},\mathcal{C}_{\mathrm{in}}) = \bm{w}_i(m,\bm{e},\mathcal{C}_{\mathrm{in}}(j,\bm{z}))$. Furthermore, since the received word $\mathcal{C}_{\mathrm{in}}(j_k,\bm{z}) \circ \mathcal{C}_{\mathrm{out}}(m)\rv{\oplus}\bm{e}$ is equal to $\bm{y}_m$, word $\bm{w}_i(m,\bm{e},\mathcal{C}_{\mathrm{in}}(j_k,\bm{z}))$ is in $\mathcal{L}_{\mathrm{in}}(\bm{y}_m,\mathcal{C}_{\mathrm{in}}(j_k,\bm{z}))$. Thus, equation (\ref{eq:eq_phi_2}) holds, and in turn, Claim \ref{claim:eq_phi_2} holds in this Case. 
  \end{proof}
 
  \begin{claim} \label{claim:eq_phi_3}
  Let $\mathcal{C}_{\mathrm{in}}$ be an $(n, \rho n)$ codebook, $m \in [2^{Rn}]$, $\bm{z} \notin \mathcal{B}_{pn}(\bm{y}_m)$ and $k \in [2^{\rho n}]$ such that $\bm{w}_k(m,\bm{e},\mathcal{C}_{\mathrm{in}})\neq \mathcal{C}_{\mathrm{out}}(m)$. If either $k>i$ or $\mathcal{C}_{\mathrm{in}}(j_k) \not\in \mathcal{B}_{pn}(\bm{y}_m)$, then $\phi_{m}(\mathcal{C}_{\mathrm{in}}) = \phi_{m}(\mathcal{C}_{\mathrm{in}}(j_k,\bm{z})).$
  \end{claim}
  
  \begin{proof}
  Claim \ref{claim:eq_phi_3} follows from Claim \ref{claim:eq_phi_2} and the observation that since $k \in [2^{\rho n}]$ such that $\bm{w}_k(m,\bm{e},\mathcal{C}_{\mathrm{in}}) \neq \mathcal{C}_{\mathrm{out}}(m)$, we have that $\mathds{1}\{\mathcal{C}_{\mathrm{in}} \circ \mathcal{C}_{\mathrm{out}}(m) \in \mathcal{O}_{\bm{\psi}}\} = \mathds{1}\{\mathcal{C}_{\mathrm{in}}(j_k,\bm{z}) \circ \mathcal{C}_{\mathrm{out}}(m) \in \mathcal{O}_{\bm{\psi}}\}$.
  \end{proof}
  
  \subsection{Smoothness of $q'$} \label{sec:char_V_prime}
  
  The goal of this subsection is to establish two bounds on $V'$. We say that a number $a_T>0$ is a \textit{typical variation coefficient} of $q'$ if for any $\mathcal{C}_{\mathrm{in}} \in \mathcal{T}$, we have $V'(\mathcal{C}_{\mathrm{in}}) \leq a_T$. We say that a number $a_G>0$ is a \textit{global variation coefficient} of $q'$ if for any $(n, \rho n)$ codebook $\mathcal{C}_{\mathrm{in}}$, we have that $V'(\mathcal{C}_{\mathrm{in}}) \leq a_G$. This subsection characterizes the smoothness of $q'$ by finding small typical and global variation coefficients that will later prove useful in establishing the concentration of $q'$. We start by finding a small global variation coefficient.

  \begin{lemma}[Global Variation Coefficient] \label{thm:global_coeff}
  If $\mathcal{O}_{\bm{\psi}}$ is bounded in size such that for some $\delta_0' \geq \delta_0$ we have $|\mathcal{O}_{\bm{\psi}}| = 2^{(1-R)n} 2^{\delta_0' n}$, then for any $(n, \rho n)$ codebook $\mathcal{C}_{\mathrm{in}}$ and for large enough $n$ (that depends only on $\delta_0$ and $\epsilon_{\rho}$), $V'(\mathcal{C}_{\mathrm{in}}) \leq a_G = 5i+14$.
  \end{lemma}
  
  Note that if the value of $i$ is small enough, the global variation coefficient given in Lemma \ref{thm:global_coeff} is an improvement over the trivial bound $V'(\cdot) \leq 2^{\rho n}$. The proof of this Lemma relies on the concatenated structure of the codebook construction. Indeed, Lemma \ref{thm:global_coeff} is our primary motivation for separating the codebook $\mathcal{C}_n$ into an inner codebook $\mathcal{C}_{\mathrm{in}}$ and outer codebook $\mathcal{C}_{\mathrm{out}}$. Before proving Lemma \ref{thm:global_coeff}, we prove the following useful inequality.

  \begin{lemma} \label{thm:global_inequality}
  For an $(n, \rho n)$ codebook $\mathcal{C}_{\mathrm{in}}$, for $m\in [2^{Rn}]$ and for $\bm{z} \not\in \mathcal{B}_{pn}(\bm{y}_m)$,
  \begin{equation} \nonumber
  \sum_{j = 1}^{2^{\rho n}} |\phi_{m}(\mathcal{C}_{\mathrm{in}}) - \phi_{m}(\mathcal{C}_{\mathrm{in}}(j,\bm{z}))| \leq i + 1.
  \end{equation}
  \end{lemma}
  
  \begin{proof}[Proof of Lemma \ref{thm:global_inequality}]
  For $k = 1, \ldots, 2^{\rho n}$, let the notation $j_k$ denote the the index $\textbf{int}(\bm{w}_k(m,\bm{e},\mathcal{C}_{\mathrm{in}})) \in [2^{\rho n}]$. Following Claim \ref{claim:eq_phi_3}, the quantity $\sum_{j=1}^{2^{\rho n}} |\phi_{m}(\mathcal{C}_{\mathrm{in}}) - \phi_{m}(\mathcal{C}_{\mathrm{in}}(j,\bm{z}))|$ is equal to 
  \begin{equation} \nonumber
  \sum_{\substack{k = 1,\ldots,i \\ \text{or } k: \bm{w}_k(m,\bm{e},\mathcal{C}_{\mathrm{in}}) = \mathcal{C}_{\mathrm{out}}(m)}} |\phi_m(\mathcal{C}_{\mathrm{in}}) - \phi_m(\mathcal{C}_{\mathrm{in}}(j_k,\bm{z}))|,
  \end{equation}
  which in turn is bounded above by $i + 1$.
  \end{proof}
  
  We are now ready to prove Lemma \ref{thm:global_coeff}.
  
  \begin{proof}[Proof of Lemma \ref{thm:global_coeff}]
  Let $\mathcal{C}_{\mathrm{in}}$ be an $(n, \rho n)$ codebook. Recall that $V'(\mathcal{C}_{\mathrm{in}})$ is equal to
  \begin{align} 
  & \sum_{j=1}^{2^{\rho n}} \sum_{\bm{z} \in \{0,1\}^n} \left\lvert \frac{\Phi(\mathcal{C}_{\mathrm{in}})}{t(\mathcal{C}_{\mathrm{in}})} - \frac{\Phi(\mathcal{C}_{\mathrm{in}}(j,\bm{z}))}{t(\mathcal{C}_{\mathrm{in}}(j,\bm{z}))} \right\rvert^2 2^{-n} \nonumber \\
  & \rv{\leq \sum_{j=1}^{2^{\rho n}} \sum_{\bm{z} \in \{0,1\}^n} \left\lvert \frac{\Phi(\mathcal{C}_{\mathrm{in}})}{t(\mathcal{C}_{\mathrm{in}})} - \frac{\Phi(\mathcal{C}_{\mathrm{in}}(j,\bm{z}))}{t(\mathcal{C}_{\mathrm{in}}(j,\bm{z}))} \right\rvert 2^{-n}} \label{eq:V_prime_extend}
  \end{align}
  \rv{where the inequality follows from the fact that both $\frac{\Phi(\mathcal{C}_{\mathrm{in}})}{t(\mathcal{C}_{\mathrm{in}})}$ and $\frac{\Phi(\mathcal{C}_{\mathrm{in}}(j,\bm{z}))}{t(\mathcal{C}_{\mathrm{in}}(j,\bm{z}))}$ are in $[0,1]$.}
  In expression (\ref{eq:V_prime_extend}), we can cut the summations by partitioning the set $\{ j \in [2^{\rho n}] \}$ into $\{j \in [2^{\rho n}]: \mathcal{C}_{\mathrm{in}}(j) \in \mathcal{O}_{\bm{\psi}} \}$ and $\{j \in [2^{\rho n}]: \mathcal{C}_{\mathrm{in}}(j) \not\in \mathcal{O}_{\bm{\psi}} \}$, and by partitioning the set $\{\bm{z} \in \{0,1 \}^n \}$ into $\{\bm{z} \in \{0,1\}^n: \bm{z} \in \mathcal{O}_{\bm{\psi}} \}$ and $\{\bm{z} \in \{0,1\}^n: \bm{z} \not\in \mathcal{O}_{\bm{\psi}} \}$. Hence, we write $V'(\mathcal{C}_{\mathrm{in}})$ as the sum of 4 terms:
  \begin{align}
  &\sum_{\substack{j \in [2^{\rho n}]: \\ \mathcal{C}_{\mathrm{in}}(j) \in \mathcal{O}_{\bm{\psi}}}} \sum_{\bm{z} \in \mathcal{O}_{\bm{\psi}}} \left\lvert \frac{\Phi(\mathcal{C}_{\mathrm{in}})}{t(\mathcal{C}_{\mathrm{in}})} - \frac{\Phi(\mathcal{C}_{\mathrm{in}}(j,\bm{z}))}{t(\mathcal{C}_{\mathrm{in}})} \right\rvert 2^{-n} \label{eq:term_1} \\
  &+ \sum_{\substack{j\in [2^{\rho n}]: \\
  \mathcal{C}_{\mathrm{in}}(j) \in \mathcal{O}_{\bm{\psi}}}} \sum_{\bm{z} \not\in \mathcal{O}_{\bm{\psi}}} \left\lvert \frac{\Phi(\mathcal{C}_{\mathrm{in}})}{t(\mathcal{C}_{\mathrm{in}})} - \frac{\Phi(\mathcal{C}_{\mathrm{in}}(j,\bm{z}))}{t(\mathcal{C}_{\mathrm{in}}(j,\bm{z}))} \right\rvert 2^{-n} \label{eq:term_2} \\
  &+ \sum_{\substack{j \in [2^{\rho n}]: \\ \mathcal{C}_{\mathrm{in}}(j) \not\in \mathcal{O}_{\bm{\psi}}}} \sum_{\bm{z} \in \mathcal{O}_{\bm{\psi}}} \left\lvert \frac{\Phi(\mathcal{C}_{\mathrm{in}})}{t(\mathcal{C}_{\mathrm{in}})} - \frac{\Phi(\mathcal{C}_{\mathrm{in}}(j,\bm{z}))}{t(\mathcal{C}_{\mathrm{in}}(j,\bm{z}))} \right\rvert 2^{-n} \label{eq:term_3} \\
  &+ \sum_{\substack{j \in [2^{\rho n}]: \\ \mathcal{C}_{\mathrm{in}}(j) \not\in \mathcal{O}_{\bm{\psi}}}} \sum_{\bm{z} \not\in \mathcal{O}_{\bm{\psi}}} \left\lvert \frac{\Phi(\mathcal{C}_{\mathrm{in}})}{t(\mathcal{C}_{\mathrm{in}})} - \frac{\Phi(\mathcal{C}_{\mathrm{in}}(j,\bm{z}))}{t(\mathcal{C}_{\mathrm{in}})} \right\rvert 2^{-n}. \label{eq:term_4}
  \end{align}
  We separately bound term (\ref{eq:term_1}) through term (\ref{eq:term_4}).
  
  \textbf{First Term:} We first bound term (\ref{eq:term_1}). Writing $\Phi$ as a sum of $\phi_{m}$ terms, term (\ref{eq:term_1}) is bounded above by
  \begin{equation} \nonumber
  \sum_{\substack{j \in [2^{\rho n}]: \\ \mathcal{C}_{\mathrm{in}}(j) \in \mathcal{O}_{\bm{\psi}}}} \sum_{\bm{z} \in \mathcal{O}_{\bm{\psi}}} \sum_{ \substack{m \in [2^{Rn}]: \\ \mathcal{C}_n(m) \in \mathcal{O}_{\bm{\psi}}}} \frac{|\phi_{m}(\mathcal{C}_{\mathrm{in}}) - \phi_{m}(\mathcal{C}_{\mathrm{in}}(j,\bm{z}))|}{t(\mathcal{C}_{\mathrm{in}})} 2^{-n}.
  \end{equation}
  which in turn can be bounded above by partitioning the set  $\{ \bm{z} \in \mathcal{O}_{\bm{\psi}} \}$ into $\{ \bm{z} \in \mathcal{O}_{\bm{\psi}} \cap \mathcal{B}_{pn}(\bm{y}_m) \}$ and $\{ \bm{z} \in \mathcal{O}_{\bm{\psi}} \cap \mathcal{B}^c_{pn}(\bm{y}_m) \}$, and applying the following inequalities: $|\mathcal{B}_{pn}(\bm{y}_m)| \leq 2^{H(p)n}$ and $|\mathcal{O}_{\bm{\psi}} \cap \mathcal{C}_n| \leq t(\mathcal{C}_{\mathrm{in}})$; the bound is as follows:
  \begin{equation} \nonumber
  \begin{aligned}
   &\sum_{ \substack{m \in [2^{Rn}]: \\ \mathcal{C}_n(m) \in \mathcal{O}_{\bm{\psi}}}} \sum_{\substack{\bm{z} \in \mathcal{O}_{\bm{\psi}}: \\
   \bm{z} \in \mathcal{B}^c_{pn}(\bm{y}_m)}}  \sum_{\substack{j \in [2^{\rho n}]: \\ \mathcal{C}_{\mathrm{in}}(j) \in \mathcal{O}_{\bm{\psi}}}} \frac{|\phi_{m}(\mathcal{C}_{\mathrm{in}}) - \phi_{m}(\mathcal{C}_{\mathrm{in}}(j,\bm{z}))|}{t(\mathcal{C}_{\mathrm{in}}) 2^{n}} \\
   & \hspace{20em} + 2^{- \epsilon_{\rho} n}
   \end{aligned}
  \end{equation}
  which in turn is bounded above by $i+1+2^{-\epsilon_{\rho} n}$ following Lemma \ref{thm:global_inequality}.
  
  \textbf{Second term:} Next, we bound term (\ref{eq:term_2}). Let the notation $j \in \mathcal{C}_{\mathrm{out}}$ denote an index $j \in [2^{\rho n}]$ that belongs to the set $\{\intg(\mathcal{C}_{\mathrm{out}}(1)),\intg(\mathcal{C}_{\mathrm{out}}(2)), \ldots, \intg(\mathcal{C}_{\mathrm{out}}(2^{Rn})) \}$. In term (\ref{eq:term_2}), we cut the summation over $j$ by partitioning the set $\{j \in [2^{\rho n}]: \mathcal{C}_{\mathrm{in}}(j) \in \mathcal{O}_{\bm{\psi}}\}$ into the sets $\{j \in [2^{\rho n}]: \mathcal{C}_{\mathrm{in}}(j) \in \mathcal{O}_{\bm{\psi}} , j \not\in \mathcal{C}_{\mathrm{out}}\}$ and $\{j \in [2^{\rho n}]: \mathcal{C}_{\mathrm{in}}(j) \in \mathcal{O}_{\bm{\psi}} , j \in \mathcal{C}_{\mathrm{out}}\}$. Since $t(\mathcal{C}_{\mathrm{in}}(j,\bm{z}))$ is equal to $t(\mathcal{C}_{\mathrm{in}})$ when $j \not\in \mathcal{C}_{\mathrm{out}}$, term (\ref{eq:term_2}) is equal to
  \begin{equation} \label{eq:cut_fn}
  \begin{aligned}
  &\sum_{\substack{j \in [2^{\rho n}]: \\ \mathcal{C}_{\mathrm{in}}(j) \in \mathcal{O}_{\bm{\psi}} \rv{\text{ and }} j \not\in \mathcal{C}_{\mathrm{out}}}} \sum_{\bm{z} \not\in \mathcal{O}_{\bm{\psi}}} \frac{|\Phi(\mathcal{C}_{\mathrm{in}}) - \Phi(\mathcal{C}_{\mathrm{in}}(j,\bm{z}))|}{t(\mathcal{C}_{\mathrm{in}}(j,\bm{z}))} 2^{-n}  \\
  & + \hspace{-1.5em} \sum_{\substack{j\in [2^{\rho n}]: \\ \mathcal{C}_{\mathrm{in}}(j) \in \mathcal{O}_{\bm{\psi}} \rv{\text{ and }} j \in \mathcal{C}_{\mathrm{out}}}} \sum_{\bm{z} \not\in \mathcal{O}_{\bm{\psi}}} \left\lvert \frac{\Phi(\mathcal{C}_{\mathrm{in}})}{t(\mathcal{C}_{\mathrm{in}})} - \frac{\Phi(\mathcal{C}_{\mathrm{in}}(j,\bm{z}))}{t(\mathcal{C}_{\mathrm{in}}(j,\bm{z}))} \right\rvert 2^{-n}.
  \end{aligned}
  \end{equation}
  To bound quantity (\ref{eq:cut_fn}), the following inequality will prove useful: \rv{defining $\Phi \triangleq \Phi(\mathcal{C}_{\mathrm{in}})$, $\Phi' \triangleq \Phi(\mathcal{C}_{\mathrm{in}}(j,\bm{z}))$, $t \triangleq t(\mathcal{C}_{\mathrm{in}})$ and $t' \triangleq t(\mathcal{C}_{\mathrm{in}}(j,\bm{z}))$, and using $|t - t'| \leq 1$, we have that
  \begin{align}
  \left\lvert \frac{\Phi(\mathcal{C}_{\mathrm{in}})}{t(\mathcal{C}_{\mathrm{in}})} - \frac{\Phi(\mathcal{C}_{\mathrm{in}}(j,\bm{z}))}{t(\mathcal{C}_{\mathrm{in}}(j,\bm{z}))} \right\rvert &= \begin{cases} \frac{\Phi}{t} - \frac{\Phi'}{t'}, & \frac{\Phi}{t} \geq \frac{\Phi'}{t'} \\ \frac{\Phi'}{t'} - \frac{\Phi}{t}, & \frac{\Phi}{t} < \frac{\Phi'}{t'}\end{cases} \nonumber \\
  & = \begin{cases} \frac{\Phi t' - \Phi' t}{t t'}, & \frac{\Phi}{t} \geq \frac{\Phi'}{t'} \\ \frac{\Phi' t - \Phi t'}{t t'}, & \frac{\Phi}{t} < \frac{\Phi'}{t'} \end{cases} \nonumber \\
  & \leq \begin{cases} \frac{\Phi(t+1)-\Phi't}{t t'}, & \frac{\Phi}{t} \geq \frac{\Phi'}{t'} \\ \frac{\Phi' t - \Phi(t-1)}{t t'}, & \frac{\Phi}{t} < \frac{\Phi'}{t'} \end{cases} \nonumber \\
  & = \begin{cases} \frac{\Phi-\Phi'}{t'} + \frac{\Phi}{t t'}, & \frac{\Phi}{t} \geq \frac{\Phi'}{t'} \\ \frac{\Phi'-\Phi}{t'} + \frac{\Phi}{t t'}, & \frac{\Phi}{t} < \frac{\Phi'}{t'} \end{cases} \nonumber \\
  & \hspace{-10em} = \frac{|\Phi(\mathcal{C}_{\mathrm{in}})-\Phi(\mathcal{C}_{\mathrm{in}}(j,\bm{z}))|}{t(\mathcal{C}_{\mathrm{in}}(j,\bm{z}))} + \frac{\Phi(\mathcal{C}_{\mathrm{in}})}{t(\mathcal{C}_{\mathrm{in}}) t(\mathcal{C}_{\mathrm{in}}(j,\bm{z}))}. \nonumber
  \end{align}}

  Following the above inequality, we have that (\ref{eq:cut_fn}) is bounded above by
  \rv{\begin{equation} \nonumber
  \begin{aligned}
  &\sum_{\substack{j\in[2^{\rho n}]: \\ \mathcal{C}_{\mathrm{in}}(j) \in \mathcal{O}_{\bm{\psi}} \text{ and } j \not\in \mathcal{C}_{\mathrm{out}}}} \sum_{\bm{z} \not\in \mathcal{O}_{\bm{\psi}}} \frac{|\Phi(\mathcal{C}_{\mathrm{in}}) - \Phi(\mathcal{C}_{\mathrm{in}}(j,\bm{z}))|}{t(\mathcal{C}_{\mathrm{in}}(j,\bm{z}))} 2^{-n} \\
  &+ \sum_{\substack{j\in[2^{\rho n}]: \\ \mathcal{C}_{\mathrm{in}}(j) \in \mathcal{O}_{\bm{\psi}} \text{ and } j \in \mathcal{C}_{\mathrm{out}}}} \sum_{\bm{z} \not\in \mathcal{O}_{\bm{\psi}}} \frac{|\Phi(\mathcal{C}_{\mathrm{in}}) - \Phi(\mathcal{C}_{\mathrm{in}}(j,\bm{z}))|}{t(\mathcal{C}_{\mathrm{in}}(j,\bm{z}))} 2^{-n} \\
  &+ \sum_{\substack{j \in [2^{\rho n}]: \\ \mathcal{C}_{\mathrm{in}}(j) \in \mathcal{O}_{\bm{\psi}} \text{ and } j \in \mathcal{C}_{\mathrm{out}}}}  \sum_{\bm{z} \not\in \mathcal{O}_{\bm{\psi}}} \frac{ \Phi(\mathcal{C}_{\mathrm{in}})}{t(\mathcal{C}_{\mathrm{in}})t(\mathcal{C}_{\mathrm{in}}(j,\bm{z}))} 2^{-n}
  \end{aligned}
  \end{equation}}

  which in turn, following that $t(\mathcal{C}_{\mathrm{in}}(j,\bm{z}))$ is bounded below by $t(\mathcal{C}_{\mathrm{in}}) - 1$, is bounded above by
  \begin{equation} \label{eq:new_term_2_2}
  \begin{aligned}
  &\sum_{\substack{j\in[2^{\rho n}]: \\ \mathcal{C}_{\mathrm{in}}(j) \in \mathcal{O}_{\bm{\psi}}}} \sum_{\bm{z} \not\in \mathcal{O}_{\bm{\psi}}} \frac{|\Phi(\mathcal{C}_{\mathrm{in}}) - \Phi(\mathcal{C}_{\mathrm{in}}(j,\bm{z}))|}{t(\mathcal{C}_{\mathrm{in}})-1} 2^{-n} \\
  &+ \sum_{\substack{j \in [2^{\rho n}]: \\ \mathcal{C}_{\mathrm{in}}(j) \in \mathcal{O}_{\bm{\psi}} \rv{\text{ and }} j \in \mathcal{C}_{\mathrm{out}}}}  \sum_{\bm{z} \not\in \mathcal{O}_{\bm{\psi}}} \frac{ \Phi(\mathcal{C}_{\mathrm{in}})}{t(\mathcal{C}_{\mathrm{in}})(t(\mathcal{C}_{\mathrm{in}})-1)} 2^{-n}
  \end{aligned}
  \end{equation}
  Following that $\Phi(\mathcal{C}_{\mathrm{in}}) \leq t(\mathcal{C}_{\mathrm{in}})$ and $|\rv{\{}j\in [2^{\rho n}]:\mathcal{C}_{\mathrm{in}}(j) \in \mathcal{O}_{\bm{\psi}},j \in \mathcal{C}_{\mathrm{out}} \rv{\}}| \leq t(\mathcal{C}_{\mathrm{in}})$, we have that (\ref{eq:new_term_2_2}) is bounded above by
  \begin{equation} \label{eq:term_2_2}
  \sum_{\substack{j \in [2^{\rho n}]: \\ \mathcal{C}_{\mathrm{in}}(j) \in \mathcal{O}_{\bm{\psi}}}} \sum_{\bm{z} \not\in \mathcal{O}_{\bm{\psi}}} \frac{|\Phi(\mathcal{C}_{\mathrm{in}}) - \Phi(\mathcal{C}_{\mathrm{in}}(j,\bm{z}))|}{t(\mathcal{C}_{\mathrm{in}})-1} 2^{-n} + \frac{t(\mathcal{C}_{\mathrm{in}})}{t(\mathcal{C}_{\mathrm{in}})-1}
  \end{equation}
   Similar to the bounding of term (\ref{eq:term_1}), we write $\Phi$ as a sum of $\phi_{m}$ terms, partition the set $\{ \bm{z} \not\in \mathcal{O}_{\bm{\psi}} \}$ into $\{ \bm{z} \in \mathcal{O}^c_{\bm{\psi}}\cap\mathcal{B}_{pn}(\bm{y}_m) \}$ and $\{ \bm{z} \in \mathcal{O}^c_{\bm{\psi}} \cap \mathcal{B}^c_{pn}(\bm{y}_m) \}$, and apply inequalities $|\mathcal{B}_{pn}(\bm{y}_m)| \leq 2^{H(p)n}$, $|\mathcal{O}_{\bm{\psi}}\cap\mathcal{C}_n| \leq t(\mathcal{C}_{\mathrm{in}})$, and $\sum_{m: \mathcal{C}'_n(m) \in \mathcal{O}_{\bm{\psi}}} \phi_{m}(\mathcal{C}_{\mathrm{in}}(j,\bm{z})) = \sum_{m: \mathcal{C}_n(m) \in \mathcal{O}_{\bm{\psi}}} \phi_{m}(\mathcal{C}_{\mathrm{in}}(j,\bm{z}))$ where $\mathcal{C}'_{n} = \mathcal{C}_{\mathrm{in}}(j,\bm{z}) \circ \mathcal{C}_{\mathrm{out}}$ when $j \in [2^{\rho n}]:\mathcal{C}_{\mathrm{in}}(j) \in \mathcal{O}_{\bm{\psi}}$ and $\bm{z} \in \mathcal{O}^c_{\bm{\psi}}$, to bound equation (\ref{eq:term_2_2}) above by
  \begin{equation} \nonumber
  \begin{aligned}
   &\sum_{\substack{m \in [2^{Rn}]: \\ \mathcal{C}_n(m) \in \mathcal{O}_{\bm{\psi}}}} \sum_{ \substack{\bm{z} \in \mathcal{O}^c_{\bm{\psi}}: \\ \bm{z} \in \mathcal{B}^c_{pn}(\bm{y}_m)}}  \sum_{\substack{j \in [2^{\rho n}]: \\ \mathcal{C}_{\mathrm{in}}(j) \in \mathcal{O}_{\bm{\psi}}}} \frac{|\phi_{m}(\mathcal{C}_{\mathrm{in}}) - \phi_{m}(\mathcal{C}_{\mathrm{in}}(j,\bm{z}))|}{\left(t(\mathcal{C}_{\mathrm{in}})-1\right) 2^n} \\
   & \hspace{12em} + \frac{2^{- \epsilon_{\rho} n}}{t(\mathcal{C}_{\mathrm{in}})-1} + \frac{t(\mathcal{C}_{\mathrm{in}})}{t(\mathcal{C}_{\mathrm{in}})-1}
   \end{aligned}
  \end{equation}
  which in turn is bounded above by
  \begin{equation} \nonumber
  \left(i+1+\frac{2^{-\epsilon_{\rho} n}}{t(\mathcal{C}_{\mathrm{in}})}+ 1\right) \frac{t(\mathcal{C}_{\mathrm{in}})}{t(\mathcal{C}_{\mathrm{in}})-1} \leq 2(i+2) + 2^{-\epsilon_{\rho}n+1}
  \end{equation}
  following Lemma \ref{thm:global_inequality} and the inequalities \rv{$|\mathcal{O}_{\psi} \cap \mathcal{C}_n| \leq t(\mathcal{C}_{\mathrm{in}})$, $t(\mathcal{C}_{\mathrm{in}}) \geq 1$} and $\frac{t(\mathcal{C}_{\mathrm{in}})}{t(\mathcal{C}_{\mathrm{in}})-1} \leq 2$.
  
  \textbf{Third term:} Next, we bound term (\ref{eq:term_3}). Using a similar approach to the bounding of term (\ref{eq:term_2}), term (\ref{eq:term_3}) is bounded above by 
  \begin{equation} \label{eq:term_3_2}
  \begin{aligned}
  &\sum_{\substack{j \in [2^{\rho n}]: \\ \mathcal{C}_{\mathrm{in}}(j) \not\in \mathcal{O}_{\bm{\psi}}}} \sum_{\bm{z} \in \mathcal{O}_{\bm{\psi}}} \frac{|\Phi(\mathcal{C}_{\mathrm{in}}) - \Phi(\mathcal{C}_{\mathrm{in}}(j,\bm{z}))|}{t(\mathcal{C}_{\mathrm{in}})} 2^{-n} \\
  & + \sum_{\substack{j \in [2^{\rho n}]: \\ (\mathcal{C}_{\mathrm{in}}(j) \not\in \mathcal{O}_{\bm{\psi}}) \cap (j \in \mathcal{C}_{\mathrm{out}})}} \sum_{\bm{z} \in \mathcal{O}_{\bm{\psi}}} \frac{\max \{\Phi(\mathcal{C}_{\mathrm{in}}),\Phi(\mathcal{C}_{\mathrm{in}}(j,\bm{z}))\}}{t^2(\mathcal{C}_{\mathrm{in}})} 2^{-n}
  \end{aligned}
  \end{equation}
  which in turn is bounded above by
  \begin{align} 
  &\sum_{\substack{j\in [2^{\rho n}]: \\ \mathcal{C}_{\mathrm{in}}(j) \not\in \mathcal{O}_{\bm{\psi}}}} \sum_{\bm{z} \in \mathcal{O}_{\bm{\psi}}} \sum_{\substack{m \in [2^{Rn}]: \\ \mathcal{C}_n(m) \in \mathcal{O}_{\bm{\psi}}}} \frac{|\phi_{m}(\mathcal{C}_{\mathrm{in}}) - \phi_{m}(\mathcal{C}_{\mathrm{in}}(j,\bm{z}))|}{t(\mathcal{C}_{\mathrm{in}})} 2^{-n} \label{eq:term_3_3_1} \\
  & + \sum_{\substack{j \in \mathcal{C}_{\mathrm{out}}: \\ \mathcal{C}_{\mathrm{in}}(j) \not\in \mathcal{O}_{\bm{\psi}}}} \sum_{\bm{z} \in \mathcal{O}_{\bm{\psi}}} \left( \frac{\phi_{m_j}(\mathcal{C}_{\mathrm{in}}(j,\bm{z}))}{t(\mathcal{C}_{\mathrm{in}})} + \frac{\Phi_{max}}{t^2(\mathcal{C}_{\mathrm{in}})} \right) 2^{-n} \label{eq:term_3_3_2} 
  \end{align}
  where $\Phi_{max} = \max \{\Phi(\mathcal{C}_{\mathrm{in}}),\Phi(\mathcal{C}_{\mathrm{in}}(j,\bm{z}))\}$ and where $m_j = (\mathcal{O}_{\bm{\psi}} \cap \mathcal{C}_{\mathrm{in}}(j,\bm{z}) \circ \mathcal{C}_{\mathrm{out}} ) \setminus (\mathcal{O}_{\bm{\psi}} \cap \mathcal{C}_n)$ (if $m_j$ is the empty set then we define $\phi_{m_j}(\mathcal{C}_{\mathrm{in}}(j,\bm{z})) = 0$). In the above expression, we are already familiar with how to bound term (\ref{eq:term_3_3_1}); using the same approach used to bound (\ref{eq:term_1}), term (\ref{eq:term_3_3_1}) is bounded above by $i+1+2^{-\epsilon_{\rho} n}$. Thus we only need to bound term (\ref{eq:term_3_3_2}). Since $\Phi_{max} \leq t(\mathcal{C}_{\mathrm{in}})+1$, term (\ref{eq:term_3_3_2}) is bounded above by $\sum_{j \in [2^{\rho n}]:\mathcal{C}_{\mathrm{in}}(j) \not\in \mathcal{O}_{\bm{\psi}},j\in \mathcal{C}_{\mathrm{in}}} \sum_{\bm{z} \in \mathcal{O}_{\bm{\psi}}} 3 \frac{2^{-n}}{t(\mathcal{C}_{\mathrm{in}})}$. By $t(\mathcal{C}_{\mathrm{in}}) \geq t_L$ and the value of $t_L$ given in Definition \ref{def:typ_param}, this in turn is bounded above by 
  \begin{equation} \label{eq:term_3_4}
  \sum_{\substack{j \in [2^{\rho n}]: \\ \mathcal{C}_{\mathrm{in}}(j) \not\in \mathcal{O}_{\bm{\psi}} \text{ and } j \in \mathcal{C}_{\mathrm{out}}}} \sum_{\bm{z} \in \mathcal{O}_{\bm{\psi}}} \frac{3(2^{-n})}{2^{-(1-R)n}|\mathcal{O}_{\bm{\psi}}| - 2^{\frac{3 \delta_0'}{4}}n}.
  \end{equation}
  Following the hypothesis of Lemma \ref{thm:global_coeff}, for large enough $n$ (which only depends on $\delta_0$), $2^{-(1-R)n} |\mathcal{O}_{\bm{\psi}}| - 2^{\frac{3 \delta_0'}{4}n}$ is bounded above by $2^{-(1-R)n} |\mathcal{O}_{\bm{\psi}}| (1/2)$. Hence, for large enough $n$, and by the inequality $|\rv{\{}j \in [2^{\rho n}]:\mathcal{C}_{\mathrm{in}}(j) \not\in \mathcal{O}_{\bm{\psi}},j \in \mathcal{C}_{\mathrm{out}} \rv{\}}| \leq 2^{Rn}$, equation (\ref{eq:term_3_4}) is bounded above by $6$. Hence, equation (\ref{eq:term_3_2}) is bounded above by $i+7+2^{-\epsilon_{\rho} n}$.
  
  \textbf{Fourth term:} Lastly, we bound term (\ref{eq:term_4}). Using the same approach to bound term (\ref{eq:term_1}), term (\ref{eq:term_4}) is bounded above by $i+1+2^{-\epsilon_{\rho} n}$. The desired result follows by summing together the upper bounds of terms (\ref{eq:term_1}) through (\ref{eq:term_4}).

  \end{proof}

  The following Lemma will help us find a small typical variation coefficient of $q'$. The Lemma states that $q_i'(\mathcal{C}_{\mathrm{in}})$ is smooth for all $\mathcal{C}_{\mathrm{in}} \in \mathcal{T}$ in a Lipschitz sense.

  \begin{lemma} \label{thm:Lipshitz}
  Define
  \begin{equation} \label{eq:Lipshitz_coeff_def}
  K_T = K_{T,i}(\vec{\mathcal{O}},\bm{\psi},\bm{e}) = \frac{2 \ell+3}{t_L-1}.
  \end{equation}
  If $\mathcal{C}_{\mathrm{in}} \in \mathcal{T}$, then $q'(\mathcal{C}_{\mathrm{in}})$ is $K_T$-Lipshitz, i.e., $\Delta'(j,\bm{z},\mathcal{C}_{\mathrm{in}}) \leq K_T$ for all $j \in [2^{\rho n}]$, $\bm{z} \in \{0,1\}^n$.
  \end{lemma}
  
  \begin{proof}[Proof of Lemma \ref{thm:Lipshitz}]
  Let $\mathcal{C}_{\mathrm{in}} \in \mathcal{T}$, $j \in [2^{\rho n}]$ and $\bm{z} \in \{0,1\}^n$. For $m \in [2^{Rn}]$, let $\bm{y}_m = \mathcal{C}_{\mathrm{in}} \circ \mathcal{C}_{\mathrm{out}}(m) \rv{\oplus} \bm{e}$. We first count the number of messages $m \in [2^{Rn}]$ such that $\phi_m(\mathcal{C}_{\mathrm{in}}) \neq \phi_m(\mathcal{C}_{\mathrm{in}}(j,\bm{z}))$.
  Since $\mathcal{C}_{\mathrm{in}} \in \mathcal{T}$, $\mathcal{C}_{\mathrm{in}}$ is $[\ell,p]$ list decodable and the $(n, \rho n)$ codebook $\mathcal{C}'_{in}$ resulting from a translation of $\mathcal{C}_{\mathrm{in}}$ by the vector $\bm{e}$ (i.e., $\mathcal{C}_{\mathrm{in}}' = \{\mathcal{C}_{\mathrm{in}}(1) \rv{\oplus} \bm{e}, \mathcal{C}_{\mathrm{in}}(2) \rv{\oplus} \bm{e}, \ldots, \mathcal{C}_{\mathrm{in}}(2^{\rho n}) \rv{\oplus} \bm{e} \}$) is also $[\ell,p]$ list decodable. Hence, there exists at most $2 \ell$ messages $m_1, \ldots, m_{2 \ell}$ such that for $k = 1, \ldots,2\ell$, either  $d(\bm{y}_{m_k},\mathcal{C}_{\mathrm{in}}(j)) \leq pn$ or $d(\bm{y}_{m_k},\bm{z}) \leq pn$. With this observation, we can state the following claim. 
  
  \begin{claim} \label{claim:eq_phi}
  For any message $m \in [2^{Rn}]$ that is \textit{not} in the set $\{m_1, \ldots, m_{2 \ell}\}$,
  \begin{equation} \label{eq:eq_phi}
  \begin{aligned}
  &\mathds{1}\{\bm{w}_i(m,\bm{e},\mathcal{C}_{\mathrm{in}}) \in \mathcal{I}_m \cap \mathcal{L}_{\mathrm{in}}(\bm{y}_m,\mathcal{C}_{\mathrm{in}})\} \\
  & = \mathds{1}\{\bm{w}_i(m,\bm{e},\mathcal{C}_{\mathrm{in}}(j,\bm{z})) \in \mathcal{I}_m \cap \mathcal{L}_{\mathrm{in}}(\bm{y}_m,\mathcal{C}_{\mathrm{in}}(j,\bm{z}))\}.
  \end{aligned}
  \end{equation}
  \end{claim}
  
  We observe that Claim \ref{claim:eq_phi} is a special case of Claim \ref{claim:eq_phi_2}. To see this, let $m \in [2^{Rn}] \setminus \{m_1,\ldots,m_{2\ell}\}$ and $\bm{y}_m = \mathcal{C}_{\mathrm{in}} \circ \mathcal{C}_{\mathrm{out}}(m) \rv{\oplus} \bm{e}$, and first observe that the ball $\mathcal{B}_{pn}(\bm{y}_m)$ contains neither $\mathcal{C}_{\mathrm{in}}(j)$ nor $\bm{z}$. Let $k \in [2^{\rho n}]$ such that $j = \intg(\bm{w}_k(m,\bm{e},\mathcal{C}_{\mathrm{in}}))$. Since $\mathcal{C}_{\mathrm{in}}(j) \not\in \mathcal{B}_{pn}(\bm{y}_m)$, it follows that $\bm{w}_k(m,\bm{e},\mathcal{C}_{\mathrm{in}})\neq \mathcal{C}_{\mathrm{out}}(m)$. These conditions are sufficient to satisfy the hypothesis of Claim \ref{claim:eq_phi_2}. 
  
  Following Claim \ref{claim:eq_phi}, the number of messages $m \in [2^{Rn}]$ such that $\phi_m(\mathcal{C}_{\mathrm{in}}) \neq \phi_m(\mathcal{C}_{\mathrm{in}}(j,\bm{z}))$ is bounded above by $2 \ell + 2$ (where the $2$ is added to account for the 2 possible messages $\{m_1',m_2' \}$ such that for $k=1,2$, $\mathds{1}\{\mathcal{C}_{\mathrm{in}} \circ \mathcal{C}_{\mathrm{out}}(m_k') \in \mathcal{O}_{\bm{\psi}}\} \neq \mathds{1}\{\mathcal{C}_{\mathrm{in}}(j,\bm{z}) \circ \mathcal{C}_{\mathrm{out}}(m_k') \in \mathcal{O}_{\bm{\psi}}\}$, which in turn may result in $\phi_{m_k'}(\mathcal{C}_{\mathrm{in}}) \neq \phi_{m'_k}(\mathcal{C}_{\mathrm{in}}(j,\bm{z}))$). From the triangle inequality, it follows that $|\Phi(\mathcal{C}_{\mathrm{in}}) - \Phi(\mathcal{C}_{\mathrm{in}}(j,\bm{z}))| \leq 2 \ell + 2$.
  
   We are now ready to prove Lemma  \ref{thm:Lipshitz}. The proof involves an upper bound of $\Delta'(j,\bm{z},\mathcal{C}_{\mathrm{in}}) = |\frac{\Phi(\mathcal{C}_{\mathrm{in}})}{t(\mathcal{C}_{\mathrm{in}})} - \frac{\Phi(\mathcal{C}_{\mathrm{in}}(j,\bm{z}))}{t(\mathcal{C}_{\mathrm{in}}(j,\bm{z}))}|$, which we illustrate by walking through the upper bound of the quantity $\frac{\Phi(\mathcal{C}_{\mathrm{in}})}{t(\mathcal{C}_{\mathrm{in}})} - \frac{\Phi(\mathcal{C}_{\mathrm{in}}(j,\bm{z}))}{t(\mathcal{C}_{\mathrm{in}}(j,\bm{z}))}$; the upper bound of the negative of the above quantity follows the same approach. We have that
  \begin{equation} \nonumber
  \begin{aligned}
  & \frac{\Phi(\mathcal{C}_{\mathrm{in}})}{t(\mathcal{C}_{\mathrm{in}})} - \frac{\Phi(\mathcal{C}_{\mathrm{in}}(j,\bm{z}))}{t(\mathcal{C}_{\mathrm{in}}(j,\bm{z}))} \\ 
   & \stackrel{(a)}{\leq}  \frac{\Phi(\mathcal{C}_{\mathrm{in}}) (t(\mathcal{C}_{\mathrm{in}})+1)}{t(\mathcal{C}_{\mathrm{in}}) t(\mathcal{C}_{\mathrm{in}}(j,\bm{z}))} - \frac{\Phi(\mathcal{C}_{\mathrm{in}}(j,\bm{z})) t(\mathcal{C}_{\mathrm{in}})}{t(\mathcal{C}_{\mathrm{in}}) t(\mathcal{C}_{\mathrm{in}}(j,\bm{z}))} \\
   & \stackrel{(b)}{\leq} \frac{2 \ell + 2}{t(\mathcal{C}_{\mathrm{in}}(j,\bm{z}))} + \frac{\Phi(\mathcal{C}_{\mathrm{in}})}{t(\mathcal{C}_{\mathrm{in}}) t(\mathcal{C}_{\mathrm{in}}(j,\bm{z}))} \\
   & \rv{\stackrel{(c)}{\leq} \frac{2 \ell + 2}{t(\mathcal{C}_{\mathrm{in}}(j,\bm{z}))} + \frac{1}{ t(\mathcal{C}_{\mathrm{in}}(j,\bm{z}))}} \\
   & \stackrel{(d)}{\leq} \frac{2\ell+3}{\rv{t_L}-1}
  \end{aligned}
  \end{equation}
  where inequality (a) follows from $t(\mathcal{C}_{\mathrm{in}}(j,\bm{z})) \leq t(\mathcal{C}_{\mathrm{in}})+1$, inequality (b) follows from $|\Phi(\mathcal{C}_{\mathrm{in}}) - \Phi(\mathcal{C}_{\mathrm{in}}(j,\bm{z}))| \leq 2 \ell + 2$ and inequality (c) follows from the inequalities  \rv{$\Phi(\mathcal{C}_{\mathrm{in}}) \leq |\mathcal{C}_n \cap \mathcal{O}_{\psi}| \leq \max\{|\mathcal{C}_{n} \cap \mathcal{O}_{\psi}|, t_L \} \triangleq t(\mathcal{C}_{\mathrm{in}})$, and inequality (d) follows from $t(\mathcal{C}_{\mathrm{in}}(j,\bm{z})) \geq t(\mathcal{C}_{\mathrm{in}})-1$ and $t(\mathcal{C}_{\mathrm{in}}) \geq t_L$}.
  \end{proof}
  
  An immediate corollary of the previous Lemma is that for all $\mathcal{C}_{\mathrm{in}} \in \mathcal{T}$, $V'(\mathcal{C}_{\mathrm{in}}) \leq 2^{\rho n} K_T^2$. In the following Lemma, this bound is tightened by exploiting the fact that the bounded difference $\Delta'(j,\bm{z},\mathcal{C}_{\mathrm{in}})$ inside the definition of $V'(\Cin)$ is often much smaller than the Lipschitz coefficient $K_T$ for many $j \in [2^{\rho n}]$ and $\bm{z} \in \{0,1\}^n$. 
  
  \begin{lemma}[Typical Variation Coefficient] \label{eq:typical_coeff}
  \rv{For $\mathcal{C}_{\mathrm{in}} \in \mathcal{T}$,}
  \begin{equation} \label{eq:typical_coeff_def}
  \rv{V'(\mathcal{C}_{\mathrm{in}}) \leq \left( t_U (\ell +1) + 2^{\delta_0'n + \epsilon_R n} \right) K_T^2.}
  \end{equation}
  
  \end{lemma}
  
  \begin{proof}[Proof of Lemma \ref{eq:typical_coeff}] To prove the upper bound on $V'(\mathcal{C}_{\mathrm{in}})$ for $\mathcal{C}_{\mathrm{in}} \in \mathcal{T}$, the proof uses two facts: (a) $q'_i$ is $K_T$ Lipschitz over $\mathcal{T}$ and (b) difference $\Delta'(j,\bm{z},\mathcal{C}_{\mathrm{in}})$ is zero for several pairs $(j,\bm{z}) \in [2^{\rho n}]\times \{0,1\}^n$. Fact (a) was established above in Lemma \ref{thm:Lipshitz}. We specify and establish Fact (b) below as the following claim.
  
  For $m \in [2^{Rn}]$, define a set $\mathcal{S}_{1,m} = \{\bm{v} \in \mathcal{C}_{\mathrm{in}}: \bm{v} \in \mathcal{B}^c_{pn}(\bm{y}_m) \}$ of all codewords outside the ball $\mathcal{B}_{pn}(\bm{y}_m)$ and define a set $\mathcal{S}_{2,m} = \{\bm{v} \in \{0,1\}^n: \bm{v} \in \mathcal{B}^c_{pn}(\bm{y}_m)  \text{ and } \bm{v} \in \mathcal{O}^c_{\psi} \}$ of all words outside both the ball $\mathcal{B}_{pn}(\bm{y}_m)$ and observation set $\mathcal{O}_{\bm{\psi}}$. For $k=1,2$, define $\mathcal{S}_k = \cap_{m: \mathcal{C}_n(m) \in \mathcal{O}_{\bm{\psi}}} \mathcal{S}_{k,m}$.
  \begin{claim}[Fact (b)] \label{thm:fact_b}
  For an $(n, \rho n)$ codebook $\Cin$, for any $j \in [2^{\rho n}]$ such that $\mathcal{C}_{\mathrm{in}}(j) \in \mathcal{S}_1$ and for any $\bm{z} \in \mathcal{S}_2$, $\Delta'(j,\bm{z},\Cin)=0$.
  \end{claim}
  
  We first prove Claim \ref{thm:fact_b}, which is a special case of Claim \ref{claim:eq_phi_3}. For an $(n, \rho n)$ codebook $\Cin$, let $j \in [2^{\rho n}]$ such that $\mathcal{C}_{\mathrm{in}}(j) \in \mathcal{S}_1$ and let $\bm{z} \in \mathcal{S}_2$. First, it is easy to verify that this choice of parameters satisfies the hypothesis of Claim \ref{claim:eq_phi_3}. Second, note that $\mathcal{C}_{\mathrm{in}}(j)$ is not in $\mathcal{O}_{\bm{\psi}} \cap \mathcal{C}_n$ and $\bm{z}$ is not in $\mathcal{O}_{\bm{\psi}}$. Hence, $\mathcal{O}_{\bm{\psi}} \cap \mathcal{C}_n = \mathcal{O}_{\bm{\psi}} \cap \mathcal{C}_{\mathrm{in}}(j,\bm{z}) \circ \mathcal{C}_{\mathrm{out}}$ and thus $t(\mathcal{C}_{\mathrm{in}}) = t(\mathcal{C}_{\mathrm{in}}(j,\bm{z}))$. In turn, the difference $\Delta'(j,\bm{z},\mathcal{C}_{\mathrm{in}}) = |q'(\mathcal{C}_{\mathrm{in}}) - q'(\mathcal{C}_{\mathrm{in}}(j,\bm{z}))|$ is equal to 
  \begin{equation} \label{eq:fact_b_1}
  \bigg\lvert \sum_{\substack{m \in [2^{Rn}]: \\ \mathcal{C}_n(m) \in \mathcal{O}_{\bm{\psi}}}} \frac{\phi_m(\Cin) - \phi_m(\Cin(j,\bm{z}))}{t(\Cin)} \bigg\rvert.
  \end{equation}
  Following Claim \ref{claim:eq_phi_3}, $\phi_m(\Cin)=\phi_m(\Cin(j,\bm{z}))$ for all $m \in [2^{Rn}]$, and thus, (\ref{eq:fact_b_1}) is zero. This completes the proof of Claim \ref{thm:fact_b}.

  Let $\Cin \in \mathcal{T}$. Following Claim \ref{thm:fact_b}, 
  $V'(\mathcal{C}_{\mathrm{in}})$ is equal to
  \begin{equation} \label{eq:typical_coeff_eq_1}
  \sum_{\substack{j\in[2^{\rho n}]: \\ \mathcal{C}_{\mathrm{in}}(j) \in \mathcal{S}_1^c}} \mathbb{E}_{\rv{\bm{z}}}[\Delta'(j,\bm{z},\mathcal{C}_{\mathrm{in}})^2] +  \sum_{\substack{j\in [2^{\rho n}]: \\ \mathcal{C}_{\mathrm{in}}(j) \in \mathcal{S}_1}} \sum_{\bm{z} \in \mathcal{S}^c_2} \frac{\Delta'(j,\bm{z},\mathcal{C}_{\mathrm{in}})^2}{2^n}.
  \end{equation}
  To finish our proof of Lemma \ref{eq:typical_coeff}, we upper bound (\ref{eq:typical_coeff_eq_1}). Since $\mathcal{C}_{\mathrm{in}} \in \mathcal{T}$, codebook $\mathcal{C}_{\mathrm{in}}$ is $[\ell,p]$ list decodable and we have that $|\mathcal{S}^c_1|$ is bounded above by $|\mathcal{O}_{\bm{\psi}} \cap \mathcal{C}_{n}| \ell$. By a simple union bound, $|S_2^c|$ is bounded above by $|\mathcal{O}_{\bm{\psi}} \cap \mathcal{C}_{n}|2^{H(p)n} + |\mathcal{O}_{\bm{\psi}}|$ and $|\mathcal{S}_1| \leq 2^{\rho n}$. By Lemma \ref{thm:Lipshitz}, $\Delta'(j,\bm{z},\mathcal{C}_{\mathrm{in}}) \leq K_T$ for all $j \in [2^{\rho n}]$ and $\bm{z} \in \{0,1\}^n$, and thus, equation (\ref{eq:typical_coeff_eq_1}) is upper bounded by $(|\mathcal{S}^c_1| + |\mathcal{S}_1| |\mathcal{S}^c_2| 2^{-n} ) K_T^2$, which in turn, is upper bounded by 
  \begin{equation} \nonumber
  \begin{aligned}
  & |\mathcal{O}_{\bm{\psi}} \cap \mathcal{C}_n| \ell K_T^2 \\
  & \hspace{0.5em} + \left(|\mathcal{O}_{\bm{\psi}} \cap \mathcal{C}_n|2^{(H(p)-1+\rho)n} + |\mathcal{O}_{\bm{\psi}}|2^{-(1-\rho)n} \right)K_T^2. 
  \end{aligned}
  \end{equation}
  Finally, (\ref{eq:typical_coeff_def}) follows by applying the bound $|\mathcal{O}_{\bm{\psi}} \cap \mathcal{C}_n| \leq t_U$.
  \end{proof}

  \subsection{Concentration of $q$} \label{sec:conc_q}
  
  The following Lemma shows that if $q'(\mathcal{C}_{\mathrm{in}}) = q(\mathcal{C}_{\mathrm{in}})$ w.h.p. over $\rv{Q}(n, \rho n)$, then $q'$ concentrated implies that $q$ is concentrated.
  \begin{lemma} \label{thm:approx_ineq}
  For any $\lambda > 0$,
  \begin{equation} \nonumber
  \begin{aligned}
  &\mathbb{P}_{\mathcal{C}_{\mathrm{in}}}(q - \mathbb{E}_{\mathcal{C}_{\mathrm{in}}}[q] > \lambda) \leq \mathbb{P}_{\mathcal{C}_{\mathrm{in}}}(q' - \mathbb{E}_{\mathcal{C}_{\mathrm{in}}}[q'] > \lambda) \\
  & \hspace{17em} + \mathbb{P}_{\mathcal{C}_{\mathrm{in}}}(q \neq q').
  \end{aligned}
  \end{equation}
  \end{lemma}
  
  \begin{proof}[Proof of Lemma \ref{thm:approx_ineq}]
  \rv{Let $\mathscr{C}$ denote the set of all $(n, \rho n)$ codebooks, and for $\mathcal{C}_{\mathrm{in}} \in \mathscr{C}$ let $Q(\mathcal{C}_{\mathrm{in}})$ denote the probability of drawing $\mathcal{C}_{\mathrm{in}}$.} We have that $\mathbb{P}_{\mathcal{C}_{\mathrm{in}}}(q - \mathbb{E}_{\mathcal{C}_{\mathrm{in}}}[q] > \lambda)$ is equal to 
  \rv{\begin{equation} \nonumber
  \begin{aligned}
  &  \sum_{\mathcal{C}_{\mathrm{in}} \in \mathscr{C}} \mathds{1}\{q(\mathcal{C}_{\mathrm{in}}) - \mathbb{E}_{\mathcal{C}_{\mathrm{in}}}[q] > \lambda\} Q(\mathcal{C}_{\mathrm{in}}) \\
  & = \sum_{\stackrel{\mathcal{C}_{\mathrm{in}}\in \mathscr{C}:}{ q(\mathcal{C}_{\mathrm{in}}) = q'(\mathcal{C}_{\mathrm{in}})}} \mathds{1}\{q(\mathcal{C}_{\mathrm{in}}) - \mathbb{E}_{\mathcal{C}_{\mathrm{in}}}[q] > \lambda \} Q(\mathcal{C}_{\mathrm{in}}) \\
  & \hspace{5 em} + \sum_{\stackrel{\mathcal{C}_{\mathrm{in}}\in \mathscr{C}:}{ q(\mathcal{C}_{\mathrm{in}}) \neq q'(\mathcal{C}_{\mathrm{in}})}} \mathds{1}\{q(\mathcal{C}_{\mathrm{in}}) - \mathbb{E}_{\mathcal{C}_{\mathrm{in}}}[q] > \lambda\} Q(\mathcal{C}_{\mathrm{in}}) \\
  & \stackrel{(a)}{\leq} \sum_{\stackrel{\mathcal{C}_{\mathrm{in}}\in \mathscr{C}:}{ q(\mathcal{C}_{\mathrm{in}}) = q'(\mathcal{C}_{\mathrm{in}})}} \mathds{1}\{q'(\mathcal{C}_{\mathrm{in}}) - \mathbb{E}_{\mathcal{C}_{\mathrm{in}}}[q'] > \lambda\} Q(\mathcal{C}_{\mathrm{in}}) \\
  & \hspace{15 em} + \sum_{\stackrel{\mathcal{C}_{\mathrm{in}}\in \mathscr{C}:}{ q(\mathcal{C}_{\mathrm{in}}) \neq q'(\mathcal{C}_{\mathrm{in}})}} Q(\mathcal{C}_{\mathrm{in}}) \\
  & \leq \mathbb{P}_{\mathcal{C}_{\mathrm{in}}}(q' - \mathbb{E}_{\mathcal{C}_{\mathrm{in}}}[q']>\lambda) + \mathbb{P}_{\mathcal{C}_{\mathrm{in}}}(q \neq q')
  \end{aligned}
  \end{equation}}
  where inequality (a) follows from $\mathbb{E}_{\mathcal{C}_{\mathrm{in}}}[q'] \leq \mathbb{E}_{\mathcal{C}_{\mathrm{in}}}[q]$.
  \end{proof}
  
  We now state and prove our concentration inequality for $q'$, which follows from a modified logarithmic Sobolev inequality \cite[Theorem 2]{Boucheron2003ConcentrationMethod}.
  
  \begin{lemma} \label{thm:global_typical_conc}
  Suppose that $|\mathcal{O}_{\bm{\psi}}| \geq 2^{(1-R)n}2^{\delta_0 n}$, let $i \in [L]$ and suppose that the following Assumptions hold:
  \begin{enumerate}
  \item For large enough $n$ (depending only on $\delta_0$, $\epsilon_{\rho}$ and $L$), there exists a global variation coefficient $a_G \in (0,\infty)$ such that $V'(\mathcal{C}_{\mathrm{in}}) \leq a_G$ for all $(n, \rho n)$ codebooks $\mathcal{C}_{\mathrm{in}}$. 
  \item There exists a typical variation coefficient $a_T \in (0,\rv{\min\{1, a_G \}})$ such that $V'(\mathcal{C}_{\mathrm{in}}) \leq a_T$ for all $\mathcal{C}_{\mathrm{in}} \in \mathcal{T}$.
  \item As a sequence in $n$, the ratio $\frac{a_G}{a^2_T}$ is $o(- \ln \mathbb{P}(\mathcal{C}_{\mathrm{in}} \not\in \mathcal{T}))$.
  \end{enumerate}
  Then for $\lambda \in (\sqrt{a_T},1)$ and for large enough $n$ (depending only on $\delta_0'$, $\epsilon_{\rho}$, $L$),
  \begin{equation}
  \mathbb{P}_{\mathcal{C}_{\mathrm{in}}} \left( q' - \mathbb{E}_{\mathcal{C}_{\mathrm{in}}}[q'] > \lambda \right) \leq \exp \left\{ -\frac{\lambda^2}{8 a_T} \right\}.
  \end{equation}
  
  \end{lemma}
  
  \begin{proof}[Proof of Lemma \ref{thm:global_typical_conc}]
  The proof follows a conventional ``entropy-method'' proof for deriving concentration inequalities \cite{Boucheron2003ConcentrationMethod}. A slight modification of the conventional proof is needed to incorporate the typical and global variation coefficients and prevent the inequality from blowing up over a small set of $(n, \rho n)$ codebooks. We begin by restating a modified logarithmic Sobolev inequality in a general form.

  \rv{\begin{lemma}[{\cite[Theorem 2]{Boucheron2003ConcentrationMethod}}] \label{thm:mlSi}
  Suppose that $X_1, X_2, \ldots, X_n$ are independent random variables taking values in a measurable set $\mathcal{X}$, and define $Z = g(X_1,\ldots,X_n)$ for some given measureable function $g: \mathcal{X}^n \rightarrow \mathbb{R}$. Furthermore, suppose that $X_1',\ldots,X_n'$ are independent copies of $X_1,\ldots,X_n$ and define both $Z(j) = g(X_1,\ldots,X_{j-1},X'_j,X_{j+1},\ldots,X_n)$ and  $$V_+ = \sum_{j=1}^{n} \mathbb{E}_{X_j'}\left[  \left(Z-Z(j)\right)^2 \mathds{1}\{Z>Z(j)\} \right].$$ Then for $\theta>0$ and $\mu \in (0,1 / \theta)$, 
  \begin{equation} \nonumber
  \ln \mathbb{E} \left[ e^{\mu(Z-\mathbb{E}[Z])} \right]  \leq \frac{\mu \theta}{1 -\mu \theta} \ln \mathbb{E} \left[ e^{\frac{\mu V_+}{\theta} }\right].
  \end{equation}
  \end{lemma}

  Next, we apply the framework of Lemma \ref{thm:mlSi} to our coding problem. Recall that the $2^{\rho n}$ codewords of $\mathcal{C}_{\mathrm{in}}$ are independent random variables that take values in $\mathcal{X} = \{0,1\}^n$. Setting $g = q'$, it follows that $(Z-Z(j))^2 = \Delta(j, \bm{z},\mathcal{C}_{\mathrm{in}})^2$, and in turn,
  \begin{equation} \nonumber
  V_+ = \sum_{j=1}^{2^{\rho n}} \mathbb{E}_{\bm{z}} \left[ \Delta(j,\bm{z},\mathcal{C}_{\mathrm{in}})^2 \mathds{1} \{q'(\mathcal{C}_{\mathrm{in}}) > q'(\mathcal{C}_{\mathrm{in}}(j,\bm{z}))\}\right].
  \end{equation}
  Note that $V_+ \leq V'$ following the indicator bound $\mathds{1}\{ \cdot \} \leq 1$ and the inequalities $\Delta(j,\bm{z},\mathcal{C}_{\mathrm{in}})^2 \leq \Delta(j,\bm{z},\mathcal{C}_{\mathrm{in}}) \leq 1$. Thus, Lemma \ref{thm:mlSi} implies that for $\theta>0$ and $\mu \in (0,1/\theta)$

  \begin{equation} \label{eq:log_Sobolev}
  \ln \mathbb{E}_{\mathcal{C}_{\mathrm{in}}} \left[ e^{\mu(q'-\mathbb{E}_{\mathcal{C}_{\mathrm{in}}}[q'])} \right]  \leq \frac{\mu \theta}{1 -\mu \theta} \ln \mathbb{E}_{\mathcal{C}_{\mathrm{in}}} \left[ e^{\frac{\mu V'}{\theta} }\right].
  \end{equation} 
  
  By Bayes formula and Assumptions 1 and 2, the expectation in the RHS of (\ref{eq:log_Sobolev}) is bounded such that
  \begin{align}
  \mathbb{E}_{\mathcal{C}_{\mathrm{in}}} \left[ e^{\frac{\mu V'}{\theta}} \right] &\leq e^{ \frac{\mu a_T}{\theta}} \mathbb{P}_{\mathcal{C}_{\mathrm{in}}} (\mathcal{C}_{\mathrm{in}} \in \mathcal{T}) +  e^{ \frac{\mu a_G}{\theta}} \mathbb{P}_{\mathcal{C}_{\mathrm{in}}} (\mathcal{C}_{\mathrm{in}} \not\in \mathcal{T}) \nonumber \\
  & \leq e^{ \frac{\mu a_T}{\theta}}  +  e^{ \frac{\mu a_G}{\theta}} \mathbb{P}_{\mathcal{C}_{\mathrm{in}}} (\mathcal{C}_{\mathrm{in}} \not\in \mathcal{T}) \label{eq:new_ls}
  \end{align}
  We set $\theta$ such that the two terms in the sum of (\ref{eq:new_ls}) are equal, i.e.,
  set $\theta = \frac{(a_G-a_T)\mu}{-\ln \mathbb{P}_{\mathcal{C}_{\mathrm{in}}}(\mathcal{C}_{\mathrm{in}} \not\in \mathcal{T})}$, and note that $\theta>0$ following $a_G > a_T > 0$.} Given this choice of $\theta$, it follows from (\ref{eq:new_ls}) that $ \mathbb{E}_{\mathcal{C}_{\mathrm{in}}} \left[ \exp \{\frac{\mu V'}{\theta} \}\right] \leq \exp \{\mu a_T/\theta + \ln(2) \}$, and in turn,  following (\ref{eq:log_Sobolev}),
  \begin{equation} \nonumber
  \mathbb{E}_{\mathcal{C}_{\mathrm{in}}} \left[ e^{\mu(q'-\mathbb{E}_{\mathcal{C}_{\mathrm{in}}}[q'])} \right] \leq \exp \left\{ \frac{\mu^2 a_T}{1- \mu \theta} + \frac{\mu \theta \ln 2}{1- \mu \theta} \right\}.
  \end{equation}
  Applying Markov's inequality to the above inequality yeilds
  \begin{equation} \label{eq:markov}
  \mathbb{P}(q' - \mathbb{E}_{\mathcal{C}_{\mathrm{in}}}[q'] > \lambda) \leq \exp \left\{ \frac{\mu^2 a_T}{1- \mu \theta} + \frac{\mu \theta \ln 2}{1- \mu \theta} - \mu \lambda \right\}.
  \end{equation}

  To finish the proof, we choose some round numbers to simplify the RHS of (\ref{eq:markov}). \rv{We set $\mu = \lambda/(2 a_T)$, which is the value that minimizes the RHS of (\ref{eq:markov}) over the optimization variable $\mu \in (0,\infty)$ when $\theta$ is treated as a constant fixed at $0$. Note that this choice of $\mu$ satisfies for large enough $n$ the requirement of Lemma \ref{thm:mlSi} that $\mu \in (0,1/\theta)$ since the quantity $$\mu \theta=\frac{\lambda^2}{4 a_T^2} \left( \frac{a_G-a_T}{- \ln \mathbb{P}(\mathcal{C}_{\mathrm{in}} \not\in \mathcal{T})} \right) \rightarrow 0 \text{ as } n \rightarrow \infty$$
  where the limit follows from the inequalities $\lambda \in [0,1]$ and $a_T < a_G$, and Assumption 3 which states that $\frac{a_G}{-a_T^2 \ln \mathbb{P}_{\mathcal{C}_{\mathrm{in}}}(\mathcal{C}_{\mathrm{in}} \not\in \mathcal{T})}$ tends to $0$ as $n$ tends to $\infty$. Following a substitution of our choices of $\theta$ and $\mu$, the RHS of (\ref{eq:markov}) is equal to
  \begin{equation} \label{eq:markov_2}
  \Scale[1]{\exp \left\{ \frac{\lambda^2}{4a_T} \left( \frac{1}{1- o(1)} \right) + \frac{o(1)}{1-o(1)}  - \frac{\lambda^2}{2a_T} \right\}}
  \end{equation}
  which in turn, is bounded above by $\exp \left\{- \frac{\lambda^2}{8a_T} \right\}$ for large enough $n$ following the condition $\lambda > \sqrt{a_T}$. Together with (\ref{eq:markov}), this yields the desired inequality.} 
  \end{proof}

  The following concentration inequality for $q$ is bootstrapped from the above concentration inequality for $q'$ using Lemma \ref{thm:approx_ineq}.
  
  \begin{lemma} \label{thm:concentration_application}
  If $\epsilon_R$ satisfies Condition \ref{cond:small_R}, $|\mathcal{O}_{\bm{\psi}}| \geq 2^{(1-R)n} 2^{\delta_0 n}$ and $i \in [L]$, then for large enough $n$ (depending only on $\delta_0$, $L$, and $\epsilon_{\rho}$),
  \begin{equation} \nonumber
  \begin{aligned}
  & \mathbb{P}_{\mathcal{C}_{\mathrm{in}}} \left(q - \mathbb{E}_{\mathcal{C}_{\mathrm{in}}}[q] > 1/n \right) \leq 2\exp \left\{ - \frac{2^{\frac{\delta_0 n}{30}}}{8n^2} \right\}.
  \end{aligned}
  \end{equation}
  \end{lemma}
  
  \begin{proof}[Proof of Lemma \ref{thm:concentration_application}]
  The proof is a straightforward application of Lemma \ref{thm:global_coeff} through Lemma \ref{thm:global_typical_conc}, but requires a little accounting to ensure that the chosen parameters check out. Recall that $\delta_0' \geq \delta_0$ is the unique constant such that $|\mathcal{O}_{\bm{\psi}}| = 2^{(1-R)n} 2^{\delta_0' n}$.
  
  We first bound the probability that $\mathcal{C}_{\mathrm{in}}$ is not typical. Note that $\mathbb{E}_{\mathcal{C}_{\mathrm{in}}}|\mathcal{O}_{\bm{\psi}}\cap\mathcal{C}_n| = 2^{-(1-R)n} |\mathcal{O}_{\bm{\psi}}| = 2^{\delta_0' n}$. Recall from Definition \ref{def:typ_param} that $\ell = 2^{\frac{4}{13}\delta_0' n}$, $t_U = 2^{\delta_0' n} + 2^{\frac{3}{4}\delta_0'n}$ and $t_L = 2^{\delta_0' n} - 2^{\frac{3}{4}\delta_0' n}$. It follows from Lemma \ref{thm:LD_lb} and Lemma \ref{thm:size_AandC} that $\mathbb{P}(q' \neq q) \leq \mathbb{P}(\mathcal{C}_{\mathrm{in}} \not\in \mathcal{T}) \leq 2^{-2^{\frac{\delta_0' n}{13}}}$ and $\mathbb{P}(|\mathcal{O}_{\bm{\psi}}\cap\mathcal{C}_n|<t_L) \leq 2 \exp\{\frac{-2^{\delta_0' n /2}}{4} \}$ for large enough $n$ (depending only on $\delta_0$).
  
  \rv{Define $a_{T}^{LB}$ as the RHS of (\ref{eq:typical_coeff_def}).} Next, we bound $K_T$ (defined by equation (\ref{eq:Lipshitz_coeff_def}) and $a_T^{LB}$. Recall that $K_T$ is equal to $\frac{2 \ell +1 }{t_L-1}$. From substitution of the typical parameters, $K_T$ is equal to $2^{-\frac{9}{13} \delta_0' n}+2$ for large enough $n$ (depending only on $\delta_0$). Similarly, $a_T^{LB}$ is equal to $2^{-\delta_0'n/13+6} + 2^{-\frac{5}{13} \delta_0'n +\epsilon_R n+4}$ for large enough $n$ (depending only on $\delta_0$ and $\epsilon_R$). Hence, for large enough $n$ (depending only on $\delta_0$ and $\epsilon_R$) we can choose any $a_T$ such that
  \begin{equation} \label{eq:aT_bound_w_parameters}
  a_T \geq a_T^{LB} = 2^{-\delta_0'n/13+6} + 2^{-\frac{5}{13} \delta_0'n +\epsilon_R n+4}
  \end{equation}
  \rv{and have that $V(\mathcal{C}_{\mathrm{in}}) \leq a_T$ for all $\mathcal{C}_{\mathrm{in}} \in \mathcal{T}$.}
  
  Finally, we are ready to apply Lemma \ref{thm:global_typical_conc}. We first check that Assumption 3 of Lemma \ref{thm:global_typical_conc} holds. We have that 
  \begin{equation} \label{eq:ass3_bound}
  \frac{a_G}{a_T^2 (-\ln \mathbb{P}(\mathcal{C}_{\mathrm{in}} \in \mathcal{T}))} < \frac{a_G}{a_T^2 2^{\frac{\delta_0' n}{13}}}
  \end{equation}
  for large enough $n$ (depending only on $\delta_0$ and $\epsilon_{\rho}$). Set $a_T = 2^{-\frac{\delta_0' n}{30}}$; this choice of $a_T$ is possible under Condition \ref{cond:small_R} and satisfies equation (\ref{eq:aT_bound_w_parameters}). Following Lemma \ref{thm:global_coeff}, $a_G$ is bounded above by $5 L + 14$ for large enough $n$ (depending only on $\delta_0$, $\epsilon_{\rho}$, and $L$). In turn, using $\delta_0' \geq \delta_0$, the RHS of quantity (\ref{eq:ass3_bound}) is bounded above by  
  \begin{equation} \nonumber
  \frac{5L+14}{2^{\frac{2 \delta_0 n}{195}}}
  \end{equation}
  for large enough $n$ (depending on $\delta_0$, $\epsilon_{\rho}$ and $L$), and therefore, is $o(1)$ and Assumption 3 holds. 
  
  To complete the proof, we apply Lemma \ref{thm:approx_ineq} to bound the quantity $\mathbb{P}(q - \mathbb{E}_{\mathcal{C}_{\mathrm{in}}}[q] > 1/n)$ above by $\mathbb{P}(q' - \mathbb{E}_{\mathcal{C}_{\mathrm{in}}}[q'] > 1/n) + \mathbb{P}(q' \neq q)$. Lemma \ref{thm:concentration_application} follows after applying Lemma \ref{thm:global_typical_conc} to bound $\mathbb{P}(q' - \mathbb{E}_{\mathcal{C}_{\mathrm{in}}}[q'] > 1/n)$ above by $\exp\{ -2^{\delta_0 n/30} / (8n^2)\}$ for large enough $n$ (depending only on $\delta_0$, $\epsilon_{\rho}$ and $L$), and after observing that $\mathbb{P}(q' \neq q)$ is bounded above by $\mathbb{P}(|\mathcal{O}_{\bm{\psi}} \cap \mathcal{C}_n| < t_L) \leq \exp\{-2^{\delta_0 n / 30} / (8n^2)\}$.
  \end{proof}
  
  \subsection{Proof of Theorem \ref{thm:main_result}} \label{sec:subsec_pf_main_result}

  We are now ready to prove Theorem \ref{thm:main_result}. Let $L > 1/\epsilon_{\rho}$ be an integer. Our strategy is to apply the sufficient condition (Lemma \ref{thm:overview_goal}) for the rate $R$ to be $(c,s)$-achievable. Recall that for one to apply Lemma \ref{thm:overview_goal}, one must show that for any $\epsilon_e>0$ and for large enough $n$, $\mathbb{P}_{\Cin}(\Cin \notin \mathcal{H}(L,\epsilon_e)) < 1 - 1/n$. To show this, we construct a set $\mathcal{G}$ of \textit{good} $(n, \rho n)$ codebooks such that $\mathcal{G}$ is contained in the set $\mathcal{H}(L,\epsilon_e)$ and $\lim _{n \rightarrow \infty}\mathbb{P}_{\Cin}(\Cin \notin \mathcal{G}) = 0$, and conclude that for large enough $n$, $\mathbb{P}_{\Cin}(\Cin \notin \mathcal{H}(L,\epsilon_e)) \leq \mathbb{P}_{\Cin}(\Cin \not\in \mathcal{G}) < 1-1/n$.
  
  For integer $n \geq 1$, we define the set of good $(n, \rho n)$ codebooks using the following sets: Let $\mathcal{E}_1$ be the set of $(n, \rho n)$ codebooks $\mathcal{C}_{\mathrm{in}}$ where there exists some $(i,\vec{\mathcal{O}},\bm{\psi},\bm{e})\in \mathcal{P}(L): |\mathcal{O}_{\bm{\psi}}|\geq 2^{(1-R)n} 2^{\delta_0 n}$ such that $q_i(\vec{\mathcal{O}},\bm{\psi},\bm{e},\Cin)$ is \textit{not} concentrated, i.e., $q_i(\vec{\mathcal{O}},\bm{\psi},\bm{e},\mathcal{C}_{\mathrm{in}}) > \mathbb{E}_{\mathcal{C}_{\mathrm{in}}}[q_i(\vec{\mathcal{O}},\bm{\psi},\bm{e},\mathcal{C}_{\mathrm{in}})] + 1/n$. Let $\mathcal{E}_2$ be the set of $(n, \rho n)$ codebooks $\mathcal{C}_{\mathrm{in}}$ where some small observation set is \textit{not} typical, i.e.,  there exists some $(i,\vec{\mathcal{O}},\bm{\psi},\bm{e})\in \mathcal{P}(L): |\mathcal{O}_{\bm{\psi}}| < 2^{(1-R)n} 2^{\delta_0 n}$ such that $|\mathcal{O}_{\bm{\psi}} \cap \mathcal{C}_n| > 2^{(\delta_0 + \delta_1) n}$. Finally, let $\mathcal{G} = (\mathcal{E}_1 \cup \mathcal{E}_2)^c$ denote the set of good $(n, \rho n)$ codebooks. We say that an $(n, \rho n)$ codebook $\mathcal{C}_{\mathrm{in}}$ is not good if $\mathcal{C}_{\mathrm{in}}$ is not in $\mathcal{G}$. To see that $\mathcal{G} \subseteq \mathcal{H}(L,\epsilon_e)$ for large enough $n$, we first observe that by Lemma \ref{thm:expected_t_new} and for large enough $n$,
  \begin{equation} \label{eq:thm1_exp_q}
  \max_{(i,\vec{\mathcal{O}},\bm{\psi},\bm{e}) \in \mathcal{P}(L)} \mathbb{E}_{\Cin}[q_i(\vec{\mathcal{O}},\bm{\psi},\bm{e},\Cin)] +1/n \leq \frac{\epsilon_e}{2L}.
  \end{equation}
  Then for large enough $n$ such that (\ref{eq:thm1_exp_q}) holds, $\Cin \in \mathcal{G}$ implies that for all $ (i,\vec{O},\bm{\psi},\bm{e}) \in \mathcal{P}(L)$,
  \begin{equation} \label{eq:thm1_GcontH1}
  q_i(\vec{\mathcal{O}},\bm{\psi},\bm{e},\Cin) \leq \frac{\epsilon_e}{2L}, \text{ if } |\mathcal{O}_{\bm{\psi}}| \leq 2^{(1-R)n}2^{\delta_0 n}
  \end{equation} 
  and 
  \begin{equation} \label{eq:thm1_GcontH2}
  |\mathcal{O}_{\bm{\psi}} \cap \mathcal{C}_n|\leq 2^{(\delta_0+\delta_1)n}, \text{ if } |\mathcal{O}_{\bm{\psi}}| > 2^{(1-R)n}2^{\delta_0 n}.
  \end{equation}
  Following (\ref{eq:thm1_GcontH2}), $\Cin \in \mathcal{G}$ implies that for all $(i,\vec{\mathcal{O},\bm{\psi},\bm{e}}) \in \mathcal{P}(L)$,
  \begin{align}
  \mathbb{P}_{m_0}(\bm{\Psi}(m_0) = \bm{\psi}) &= \mathbb{P}_{m_0}(\mathcal{C}_n(m_0) \in \mathcal{O}_{\bm{\psi}}) \nonumber \\
  & = |\mathcal{O}_{\psi} \cap \mathcal{C}_n| 2^{-Rn} \nonumber \\
  &\leq 2^{(\delta_0+\delta_1-R)n}, \text{ if } |\mathcal{O}_{\bm{\psi}}| > 2^{(1-R)n}2^{\delta_0 n}. \label{eq:thm1_GcontH3}
  \end{align}
  Using the definition of set $\mathcal{H}(L,\epsilon_e)$, it is easy to verify that for any $(n, \rho n)$ codebook $\Cin$ such that both (\ref{eq:thm1_GcontH1}) and (\ref{eq:thm1_GcontH3}) hold, we have that $\Cin \in \mathcal{H}(L,\epsilon_e)$.

  We now bound the probability that $\mathcal{C}_{\mathrm{in}}$ is not good by bounding $\mathbb{P}(\mathcal{E}_1)$ and $\mathbb{P}(\mathcal{E}_2)$. The adversary's computational bound will help us to bound both $\mathbb{P}(\mathcal{E}_1)$ and $\mathbb{P}(\mathcal{E}_2)$. Let $S$ denote the number of unique observation sets in $\CPX$, i.e., $S = |\rv{\{}\mathcal{O} \subseteq \{0,1\}^n: \vec{\mathcal{O}} \in \CPX, \mathcal{O} = \mathcal{O}_{\bm{\psi}} \text{ for some } \bm{\psi} \in \{0,1\}^{rn} \rv{\}}|$. We can bound $S$ by counting the number of Boolean circuits with $cn^s$ logic gates. 
  \begin{lemma} \label{thm:circuit_cmplx}
   For large enough $n$ (depending only on $c$ and $s$), the number of functions in $\textset{CKT}(r,cn^s)$ is bounded above by $2^{n^{s+2}}$, and thus $S = S(r,cn^s) \leq 2^{n^{s+3}}$. Proof is in Appendix \ref{sec:proof_circuit_cmplx}.
  \end{lemma} 
  For large enough $n$ (depending only on $\delta_0$, $L$, $\epsilon_{\rho}$, $c$ and $s$),
  \begin{align}
  &\mathbb{P}(\mathcal{E}_1) = \mathbb{P}_{\Cin} \left( \bigcup_{\substack{(i,\vec{O},\bm{\psi},\bm{e}) \in \mathcal{P}(L): \\ |\mathcal{O}_{\bm{\psi}}| \geq 2^{(1-R+\delta_0)n}}} \left\{ q > \mathbb{E}_{\Cin}[q] + 1/n \right\} \right) \nonumber \\
  & \leq \sum_{\substack{(i,\vec{O},\bm{\psi},\bm{e}) \in \mathcal{P}(L): \\ |\mathcal{O}_{\bm{\psi}}| \geq 2^{(1-R+\delta_0)n}}} \mathbb{P}_{\Cin} (q > \mathbb{E}_{\Cin}[q] + 1/n) \label{eq:E1_union} \\
  & \leq S2^nL2\exp \left\{-\frac{2^{\delta_0 n/30}}{8n^2} \right\} \label{eq:E1_ub}
  \end{align}
  where (\ref{eq:E1_union}) follows from a simple union bound and (\ref{eq:E1_ub}) follows from Lemma \ref{thm:concentration_application}. Furthermore, for large enough $n$ (depending only on $\delta_0$, $\delta_1$, $c$ and $s$),
  \begin{equation} \nonumber
  \mathbb{P}(\mathcal{E}_2) \leq S 2 \exp \left\{ 2^{- \delta_0 n} \right\}
  \end{equation}
  which follows from a simple union bound and Lemma \ref{thm:size_AandC}. In turn, given the bound on $S$ established in Lemma \ref{thm:circuit_cmplx}, it is clear that both $\mathbb{P}(\mathcal{E}_1)$ and $\mathbb{P}(\mathcal{E}_2)$ are going to $0$ in the limit as $n\rightarrow \infty$. Hence, for $\epsilon_e>0$ and for large enough $n$, $\mathbb{P}_{\mathcal{C}_{\mathrm{in}}}(\mathcal{C}_{\mathrm{in}} \not\in \mathcal{H}(L,\epsilon_e)) \leq  \mathbb{P}_{\mathcal{C}_{\mathrm{in}}}(\mathcal{C}_{\mathrm{in}} \not\in \mathcal{G}) \leq \mathbb{P}(\mathcal{E}_1)+\mathbb{P}(\mathcal{E}_2) < 1 - 1/n$. This completes the proof of Theorem \ref{thm:main_result}.

  \section{Conclusion} \label{sec:conc}
  In this work, we define and study a binary channel controlled by a $\CPX$-oblivious adversary (an adversary that can observe a fraction $r$ of all bits in the transmitted codeword via some function $f \in \CPX$ of bounded complexity and flip a fraction $p$ of all bits). The capacity $C(p,r,\rv{c,s})$ of this channel is characterized for the parameter range $r < 1- H(p)$ (i.e., a sufficiently myopic adversary) under deterministic codes and average error criterion. \rvv{We give a proof of this result which is based on a new application logarithmic Sobolev inequalities.

  An alternative proof of the above result can be stated using the proof techniques for sufficiently myopic channels developed by Dey, Jaggi and Langberg \cite{Dey2019a}. The advantage of the alternative proof is that it uses a simpler random coding scheme, involves a simpler analysis, and can provide more general results than the proof of Section \ref{sec:analysis}. An outline of this alternative proof is provided in Appendix \ref{sec:alt_proof}.

  Lastly, we remark that a $\CPX$-oblivious adversary can be strictly less powerful than a $\textset{CKT}(r,\infty)$-oblivious adversary (i.e., an adversary with no complexity constraint), in the sense that $C(p,r,\infty,\infty)$ is strictly less than $C(p,r,c,s)$ for some values of $p \in (0,1/2)$ and $r < 1 - H(p)$. A proof sketch is as follows. If no complexity constraint is imposed, then the adversary can choose a function $f$ (dependent on the codebook $\mathcal{C}_n$) that does the following:
  \begin{enumerate}
  \item Take the transmitted codeword $\bm{x}$ as input. Compute the nearest codeword $\bm{x}'$ to $\bm{x}$ and, in turn, compute an ``auxiliary'' error vector $\bm{s} \in \{0,1\}^n$ such that $\bm{x} \oplus \bm{s}$ is equal Hamming distance to both $\bm{x}$ and $\bm{x}'$. 
  \item Let $w(\bm{s})$ denote the Hamming weight of $\bm{s}$. If $w(\bm{s})$ is small enough such that the total number of length-$n$ binary vectors of weight $w(\bm{s})$ or less (call this number $A_{w(\bm{s})}$) is at most $2^{rn}$, then ``compress'' $\bm{s}$ into an $rn$ bit vector and output this compressed vector. Otherwise, output an error.
  \end{enumerate}
  Let $\mathrm{LP}(\delta)$ denote the linear programming bound for binary codes with minimum distance $\delta n$, and let $\mathrm{LP}^{-1}$ denote its inverse. By the linear programming bound (i.e., MRRW bound) \cite{McEliece1977a}, $\bm{x}$ and $\bm{x}'$ are (with positive probability) within $n  \mathrm{LP}^{-1}(R)$ bits where $R$ is the rate, and thus, $w(\bm{s})$ is no more than about $\frac{n}{2}  \mathrm{LP}^{-1}(R)$ and $A_{w(\bm{s})}$ is no more than about $2^{nH(\frac{1}{2} \mathrm{LP}^{-1}(R))}$. Hence, if $r>H(\frac{1}{2} \mathrm{LP}^{-1}(R))$ and $p>\frac{1}{2} \mathrm{LP}^{-1}(R)$, with positive probability, the adversary can reconstruct $\bm{s}$ from the output of $f$ and, in turn, choose the true error vector $\bm{e} = \bm{s}$ to confuse Alice as to whether $\bm{x}$ or $\bm{x}'$ was transmitted. As can be verified numerically, the above adversarial strategy can be used to upper bound $C(p,r,\infty,\infty)$ and show that $C(p,r,\infty,\infty) < C(p,r,c,s)$ for a certain range of $r < 1 - H(p)$.
  }

  \appendices

  \rvv{\section{Alternative Proof of Theorem \ref{thm:main_result}} \label{sec:alt_proof}

  In this appendix, we present an outline of an alternative proof of Theorem \ref{thm:main_result}. The proof closely follows the achievability proof of Dey, Jaggi and Langberg \cite[Theorem III.1]{Dey2019a} for sufficiently myopic channels. To follow this alternative proof, we point the reader to this reference, and structure our outline to emphasize the difference between the alternative proof and the proof of \cite[Theorem III.1]{Dey2019a}. For encoding, we replace our concatenated code construction with a simple random code where the codewords of code $\mathcal{C}_{n}$ are i.i.d. uniform in $\{0,1\}^n$. For decoding, use the Hamming ball decoder as in \cite{Dey2019a}.
  \begin{itemize}
  \item Modify \cite[Lemma IV.2]{Dey2019a} such that for any function $f:\{0,1\}^{n}\rightarrow \{0,1\}^{rn}$ (not necessarily with polynomial circuit complexity) and any $rn$-bit observation vector $\bm{\psi}$, the probability (over random code design) that there are fewer than about $|\mathcal{O}_{\bm{\psi}}|2^{(R-1)n}$ codewords compatible with the output of the function $f$ is $\exp\{-2^{\Omega(n)}\}$.
  \item In \cite[Lemma IV.3]{Dey2019a}, instead of analyzing the event that the number of codewords in a set $\mathcal{V} \subseteq \{0,1\}^n$ exceeds $\Theta(n^2)$, show that this number exceeds $\Theta(n^{s+4})$ with probability $2^{-\Omega(n^{s+4})}$ over random code design.
  \item In \cite[Corollary IV.4]{Dey2019a}, the $n^2$ is replaced by $n^{s+4}$, and in \cite[Lemma IV.5]{Dey2019a} the $n^4$ is replaced by $n^{2(s+4)}$.
  \item The subsequent arguments in \cite[Lemma IV.6]{Dey2019a} are similarly modified, with $n^2$ being replaced by $n^{s+4}$.
  \item In each of these Lemmas, instead of union bounding over all error vectors (numbering $2^{O(n)}$), one union bounds
over all error vectors and circuits in $\CPX$ (numbering $O(n^{s+3})$ following Lemma \ref{thm:circuit_cmplx}).
  \item In the analysis, allow decoding to fail over small observation sets as described in Section \ref{sec:approx_q} of this paper. In the event that Alice's transmitted codeword belongs to a small set, this means that adversary has high certainty of Alice's codeword/message upon observing $\bm{\psi}$, and thus, may be able design $\bm{e}$ well-tailored for this codeword/message and induce a decoding error. However, we can ignore this event in the analysis, since such an event is unlikely and thus makes a negligible contribution to the probability of error.
  \end{itemize}}

  \section{A Talagrand-type Concentration Inequality}
  Let $g(\cdot)$ be a function mapping the set of $(n, \rho n)$ codebooks to $(-\infty,\infty)$. For $b > 0$, $g$ is said to be \textit{$b$-Lipshitz} if for any $(n, \rho n)$ codebooks $\mathcal{C}_{\mathrm{in}}$ and $\mathcal{C}_{\mathrm{in}}'$ differing by at most 1 codeword, then $|g(\mathcal{C}_{\mathrm{in}}) - g(\mathcal{C}_{\mathrm{in}}')|\leq b$. An index set $J(\cdot) \subseteq [2^{\rho n}]$ is said to be a \textit{certificate} of $g$ if for any $(n, \rho n)$ codebook $\mathcal{C}_{\mathrm{in}}$, $g(\mathcal{C}_{\mathrm{in}}) \geq |J(\mathcal{C}_{\mathrm{in}})|$ and $g(\mathcal{C}_{\mathrm{in}}') \geq g(\mathcal{C}_{\mathrm{in}})$ for any $\mathcal{C}_{\mathrm{in}}'$ that agrees with $\mathcal{C}_{\mathrm{in}}$ on the codewords indexed in $J(\mathcal{C}_{\mathrm{in}})$. Lastly, for $c>0$, $g$ is said to be \text{$c$-certifiable} if there exists a certificate $J$ of $g$ such that $|J(\mathcal{C}_{\mathrm{in}})| \leq c g(\mathcal{C}_{\mathrm{in}})$ for all $(n, \rho n)$ codebooks $\mathcal{C}_{\mathrm{in}}$.
  
  \begin{lemma}[{\cite[Theorem~11.3]{Dubhashi2009ConcentrationAlgorithms}}] \label{thm:Talagrand} Let $\mathbb{M}[g]$ denote a median of $g$. For any $t > 0$, 
  \begin{equation} \nonumber
  \mathbb{P}_{\mathcal{C}_{\mathrm{in}}}(g - \mathbb{M}[g] > t) \leq 2 \exp \left\{ \frac{-t^2}{4b^2 c (\mathbb{M}[g]+t)} \right\}
  \end{equation}
  and
  \begin{equation} \nonumber
  \mathbb{P}_{\mathcal{C}_{\mathrm{in}}}(g - \mathbb{M}[g] < -t) \leq 2 \exp \left\{ \frac{-t^2}{4b^2 c \mathbb{M}[g]} \right\}.
  \end{equation}
  
  \end{lemma}

  \section{Proof of Lemma \ref{thm:LD_lb}} \label{sec:proof_LD_lb}
  For $\bm{y} \in \{0,1\}^n$, define $g_{\bm{y}}(\mathcal{C}_{\mathrm{in}}) = |\mathcal{C}_{\mathrm{in}} \cap \mathcal{B}_{pn}(\bm{y})|$. Our goal is to show that $g_{\bm{y}}$ is strongly concentrated around its expectation $\mathbb{E}_{\mathcal{C}_{\mathrm{in}}}[g_{\bm{y}}]$. Note the following: $g_{\bm{y}}$ is 1-Lipschitz and $J(\mathcal{C}_{\mathrm{in}}) = \{ k \in [2^{\rho n}]: \mathcal{C}_{\mathrm{in}}(k) \in \mathcal{B}_{pn}(\bm{y}) \}$ is a certificate of $g_{\bm{y}}(\mathcal{C}_{\mathrm{in}})$ where it follows that $g_{\bm{y}}$ is $1$-certifiable.
  
  Since $g_{\bm{y}}(\mathcal{C}_{\mathrm{in}})$ is a binomial random variable, the value $\text{floor}(\mathbb{E}_{\mathcal{C}_{\mathrm{in}}}[g_{\bm{y}}])$ is a median. Set $\mathbb{M}[g_{\bm{y}}] = \text{floor}(\mathbb{E}_{\mathcal{C}_{\mathrm{in}}}[g_{\bm{y}}])$. Note that $R < 1 - H(p)$ implies that $\mathbb{E}_{\mathcal{C}_{\mathrm{in}}}[g_{\bm{y}}]$ (which is equal to $\sum_{i = 1}^{2^{Rn}} \mathbb{P}_{\mathcal{C}_{\mathrm{in}}}( x_i \in \mathcal{B}_{pn}(\bm{y}) ) \leq 2^{(R - 1 + H(p))n}$) is going to zero in $n$. It follows that for large enough $n$, $\mathbb{M}[g_{\bm{y}}] = 0$.
  
  By Lemma \ref{thm:Talagrand}, for $\ell > 0$ the probability that $g_{\bm{y}} > \ell$ is at most $2^{- \log(e) \frac{\ell}{4}+1}$. In conclusion, $\mathbb{P}_{\mathcal{C}_{\mathrm{in}}}(\exists \bm{y} \in \{0,1\}^n \text{ s.t. } g_{\bm{y}}(\mathcal{C}_{\mathrm{in}}) > \ell) = \mathbb{P}_{\mathcal{C}_{\mathrm{in}}}(\cup_{\bm{y} \in \{0,1\}^n} \{g_{\bm{y}}(\mathcal{C}_{\mathrm{in}}) > \ell \}) < 2^n2^{- \log(e) \frac{\ell(n)}{4}+1}$.
 
  \section{Proof of Lemma \ref{thm:size_AandC}} \label{sec:proof_size_AandC}
  
  Define $g(\mathcal{C}_{\mathrm{in}}) = |\mathcal{A} \cap \mathcal{C}_n|$. Note the following: $g(\cdot)$ is 1-Lipshitz and $J(\mathcal{C}_{\mathrm{in}}) = \{m \in [2^{Rn}]: \mathcal{C}_{\mathrm{in}} \circ \mathcal{C}_{\mathrm{out}}(m) \in \mathcal{A} \}$ is a certificate of $g(\mathcal{C}_{\mathrm{in}})$ where it follows that $g(\mathcal{C}_{\mathrm{in}})$ is 1-certifiable.
  
  Since $g(\mathcal{C}_{\mathrm{in}})$ is a binomial random variable, the expected value $\mathbb{E}_{\mathcal{C}_{\mathrm{in}}}[g]$ is a median. Set $\mathbb{M}[g] = \mathbb{E}_{\mathcal{C}_{\mathrm{in}}}[g] = 2^{-(1-R)n}|\mathcal{A}|$. The desired result follows from Lemma \ref{thm:Talagrand}.
   
  \section{Proof of Lemma \ref{thm:circuit_cmplx}} \label{sec:proof_circuit_cmplx}
  
  Let $W$ be the number of functions of the form $g_n:\{0,1\}^n \rightarrow \{0,1\}$ that can be computed by a Boolean circuit (of $n$ inputs and $1$ output) of size $cn^s$. We first show that $W < 2^{2(s+1)n^{s+1}}$. 
  
  Note that each gate can compute one of $16$ unique functions from $\{0,1\}^2$ to $\{0,1\}$. Furthermore, for a given gate, the number of ways to choose $2$ gate inputs from $n$ circuit inputs, $cn^s-1$ gate outputs, and 2 constant inputs (i.e., $0$ and $1$) is bounded above by $(n+cn^s+1)^2$. It follows that $W$ is bounded above by $(16 (n+cn^s+1)^2)^{cn^s}$ which in turn, for large enough $n$, is bounded above by $(n^{s+1})^{c2n^s} = 2^{2c(s+1)n^s \log n}$. Done. 
  
  We now prove Lemma \ref{thm:circuit_cmplx}. Any function in $\mathcal{F}_{n,r}$ that is computable by a Boolean circuit (of $n$ inputs and $rn$ outputs) of size $c n^s$ can be computed by some $r n$ Boolean circuits (of $n$ inputs and $1$ output) each of size $c n^s$. Hence, the number of functions in $\mathcal{F}_{n,r}$ that can be computed by a Boolean circuit (of $n$ inputs and $rn$ outputs) of size $cn^s$ is bounded above by $W^{rn}$. We finish the proof by observing that $W^{rn}$ is smaller than $2^{n^{s+2}}$ for large enough $n$.

  \rvv{\section*{Acknowledgement}

  We would like to thank the anonymous reviewers and the associate editor for the various constructive feedback during the review process, including the alternative proof of Theorem~\ref{thm:main_result} in Appendix~\ref{sec:alt_proof}, the ideas of the proof in Section \ref{sec:conc} confirming $C(p,r,\infty,\infty)<C(p,r,c,s)$ for some range of $r< 1-H(p)$, and ideas of the stochastic code construction in Section \ref{sec:main_result} for achieving rates above the GV bound for some range of $r > 1-H(p)$.}

\bibliographystyle{IEEEtran}
\bibliography{refs}

\IEEEpeerreviewmaketitle

\end{document}